\documentclass[sn-mathphys-num,pdflatex]{sn-jnl}

\usepackage[english]{babel}
\usepackage[bb=boondox,bbscaled=1.05,scr=dutchcal]{mathalfa}	%
\usepackage{dsfont}
\usepackage[T1]{fontenc}
\normalfont

\usepackage{slashed}

\usepackage[dvipsnames]{xcolor}

\usepackage{amssymb,amsmath,amsthm,anyfontsize,enumerate,layout,mathrsfs,mdwlist,nicematrix,stmaryrd,tikz,thmtools}

\makeatletter
\renewcommand\thmt@autorefsetup{%
  \@xa\def\csname\thmt@envname autorefname\@xa\endcsname\@xa{\thmt@thmname}%
}
\makeatother

\usepackage[title]{appendix}
\usepackage{dsfont} 

\usepackage{mathtools}

\usepackage[english,noabbrev]{cleveref}
\crefname{lemma}{lemma}{lemmata}
\Crefname{lemma}{Lemma}{Lemmata}

\allowdisplaybreaks 

\declaretheorem[numberwithin=section,style=thmstyleone,refname={theorem,theorems}]{theorem}
\RenewCommandCopy{\theHtheorem}{\thetheorem}
\declaretheorem[sibling=theorem,style=thmstyleone,refname={proposition,propositions}]{proposition}    
\declaretheorem[sibling=theorem,style=thmstyleone,refname={lemma,lemmata}]{lemma}
\declaretheorem[sibling=theorem,style=thmstyleone,refname={corollary,corollaries}]{corollary}

\declaretheorem[sibling=theorem,style=thmstyletwo,refname={remark,remarks}]{remark}
\declaretheorem[sibling=theorem,style=thmstyletwo,refname={example,examples}]{example}
\declaretheorem[sibling=theorem,style=thmstylethree,refname={definition,definitions}]{definition}

\numberwithin{equation}{section}

\renewcommand{\theta}{\vartheta}
\renewcommand{\phi}{\varphi}
\renewcommand{\epsilon}{\varepsilon}
\newcommand{\normord}[1]{\vcentcolon\mathrel{#1}\vcentcolon}
\newcommand{\mb}[1]{\mathbb{#1}} 
\newcommand{\mf}[1]{\mathfrak{#1}}
\newcommand{\mc}[1]{\mathcal{#1}}

\newcommand{\N}{\mb{N}} 
\newcommand{\C}{\mb{C}} 
\newcommand{\Z}{\mb{Z}} 

\renewcommand{\P}{\mb{P}}

\newcommand{\DA}{\widehat{\mathcal{D}}_A^\hslash}
\newcommand{\I}{\mathcal{I}}

\renewcommand{\>}{\rangle}
\newcommand{\del}{\partial}

\newcommand{\2}{\frac{1}{2}}

\DeclareMathOperator{\Ker}{Ker}

\DeclareMathOperator{\Tr}{Tr}

\DeclareMathOperator{\End}{End}

\newcommand{\Id}{\mathord{\mathrm{Id}}}
\DeclareMathOperator*{\Res}{Res}

\DeclareMathOperator{\ad}{ad}
\DeclareMathOperator{\GL}{GL}

\newcommand{\DAh}{\mathcal{D}_A^{\hslash}}

\makeatletter 
\def\@acknow{}%
\long\def\EarlyAcknow#1 \par{%
\def\@acknow{\abstractfont\abstracthead*{Acknowledgments}
#1\par}}%

\def\printabstract{\ifx\@acknow\empty\else\@acknow\fi\par%
    \ifx\@abstract\empty\else\@abstract\fi\par}
\makeatother

\begingroup\expandafter\expandafter\expandafter\endgroup
\expandafter\ifx\csname pdfsuppresswarningpagegroup\endcsname\relax
\else
\fi
{}

\begin{document}

\title{Highest weight vectors, shifted topological recursion and quantum curves}

\author[1,2,3]{\fnm{Rapha\"el} \sur{Belliard}} \email{belliard@imada.sdu.dk} 

\author[1]{\fnm{Vincent} \sur{Bouchard}}\email{vincent.bouchard@ualberta.ca} 

\author*[1,4,5]{\fnm{Reinier} \sur{Kramer}}\email{reinier.kramer@unimib.it} 

\author[1]{\fnm{Tanner} \sur{Nelson}}\email{tnelson1@ualberta.ca}

\affil[1]{\orgdiv{Department of Mathematics and Statistics}, \orgname{University of Alberta}, \orgaddress{\city{Edmonton}, \postcode{T6G 2G1}, \state{AB}, \country{Canada}}}

\affil[2]{\orgdiv{Centre for Quantum Topology and Its Applications (quanTA) and Department of Mathematics and Statistics}, \orgname{University of Saskatchewan}, \orgaddress{\city{Saskatoon}, \postcode{10587}, \state{SK}, \country{Canada}}}

\affil[3]{\orgdiv{Centre for Quantum Mathematics, IMADA},
\orgname{Syddansk Universitet}, \orgaddress{\street{Campusvej 55}, \postcode{5230} \city{Odense}, \country{Denmark}}}

\affil*[4]{\orgdiv{Dipartimento di Matematica e Applicazioni}, \orgname{Organization}, \orgaddress{\street{Via Roberto Cozzi, 55}, \city{Milano}, \postcode{20125}, \country{Italy}}}

\affil[5]{\orgdiv{Sezione di Milano-Bicocca}, \orgname{INFN}, \orgaddress{\city{Milano}, \country{Italy}}}

\EarlyAcknow{We would like to thank N. Chidambaram, A. Giacchetto, J. Hurtubise, P. Lorenzoni and S. Shadrin for interesting discussions, and J. Hurtubise in particular for explaining the construction of formal WKB solutions relevant to our situation.

The authors acknowledge support from the National Science and Engineering Research Council of Canada. R.K. is partially supported by funds of the Istituto Nazionale di Fisica Nucleare, by IS-CSN4 Mathematical Methods of Nonlinear Physics. R.K. is also thankful to GNFM (Gruppo Nazionale di Fisica Matematica) for supporting activities that contributed to the research reported in this paper.\par
The University of Alberta respectfully acknowledges that they are situated on Treaty 6 territory, traditional lands of First Nations and Métis people.

This version of the article has been accepted for publication, after peer review but is not the Version of Record and does not reflect post-acceptance improvements, or any corrections. The Version of Record is available online at: \url{http://dx.doi.org/10.1007/s00220-025-05448-6}}

\abstract{
  We extend the theory of topological recursion by considering Airy ideals (also known as Airy structures) whose partition functions are highest weight vectors of particular $\mc{W}$-algebra representations. Such highest weight vectors arise as partition functions of Airy ideals only under certain conditions on the representations. In the spectral curve formulation of topological recursion, we show that this generalization amounts to adding specific terms to the correlators $ \omega_{g,1}$, which leads to a ``shifted topological recursion'' formula. We then prove that the wave-functions constructed from this shifted version of topological recursion are WKB solutions of families of quantizations of the spectral curve with $ \hslash$-dependent terms. In the reverse direction, starting from an $\hslash$-connection, we find that it is of topological type if the exact same conditions that we found for the Airy ideals are satisfied. When this happens, the resulting shifted loop equations can be solved by the shifted topological recursion obtained earlier.
  }

\maketitle

\tableofcontents


\section{Introduction}

\subsection{Motivation: quantum curves}

The Eynard-Orantin topological recursion \cite{EO07} is a method to calculate invariants associated to Riemann surfaces by a formula which is recursive on the negative of the Euler characteristic $ 2g - 2 + n$. It has as input only the cases $(g,n) = (0,1), (\frac{1}{2}, 1), (0,2)$ where $2g-2+n \leq 0$ -- these define the (spectral) curve of the problem -- and produces as output a set $ \{ \omega_{g,n} \}_{g \in \frac{1}{2}{\mathbb{N}, n \in \mathbb{N}^*}} $ of symmetric multidifferentials on the spectral curve. Topological recursion has applications to matrix models \cite{CEO06}, volumes of moduli spaces \cite{EO09}, Gromov--Witten theory \cite{BKMP09,EO15,DOSS14,GKLS22}, Hurwitz numbers and hypergeometric KP tau-functions \cite{BM08,DKPS19a,BDKS20a}, and WKB analysis of Lax systems \cite{BBE15,BEM17,IMS18}, among others.\par

Topological recursion can be understood as a quantization formalism \cite{BE09}. The spectral curve can often be understood as an algebraic curve $P(x,y) = 0$. We then consider the following question: how can the spectral curve be quantized? I.e., how do we construct a function $ \psi (\hslash; x)$ and a differential operator $ \hat{P} (\hslash; \hat{x}, \hat{y})$ with $ \hat{x} = x\cdot$ and $ \hat{y} = \hslash \frac{d}{dx}$ such that
\begin{equation}\label{eq:quant}
  \hat{P} (\hslash; \hat{x}, \hat{y}) \psi (\hslash; x) = 0 \,, \qquad \hat{P}(0; x, y) = P(x,y)\,.
\end{equation}

There are many operators $\hat{P}$ that reduce to $P$ this way, but due to non-commutativity of $ \hat{x}$ and $ \hat{y}$, there is no canonical choice. So the real question is: how do we quantize in a meaningful way?

Topological recursion provides an answer to this question, as was originally suggested in \cite{BE09}. Out of the differentials $\omega_{g,n}$ produced by topological recursion, one can construct a wave-function that is annihilated by a quantization of the spectral curve. For genus $0$ spectral curves, the wave function $ \psi $ is constructed by integrating the $ \omega_{g,n}$ along a correctly chosen divisor $D$ of degree $ -1$, and assembling them in a multivalued WKB-type generating series:
\begin{equation}
\psi(z) = \exp \left(\sum_{g \in \frac{1}{2}\mathbb{N}, n \in \mathbb{N}^*} \frac{\hslash^{2g-2+n}}{n!} \int_{D+[z]} \cdots \int_{D+[z]} \left(\omega_{g,n} - \delta_{g,0} \delta_{n,2} \frac{dx(z_1) dx(z_2)}{(x(z_1)-x(z_2))^2} \right) \right) \,.
\end{equation}
For higher genus, the correct wave function is a transseries obtained from this function as a generalized theta series.

This quantization procedure was proved for a large class of genus $0$ curves by Bouchard--Eynard \cite{BE17} and for higher genus by Eynard--Garcia-Failde--Marchal--Orantin \cite{EGMO21}. More precisely, in the original formulation of topological recursion, the projection of the spectral curve to the first coordinate, $ x \colon \Sigma \to \P^1$, has to be simply ramified. This was generalized to spectral curves with arbitrary ramification in \cite{BE13}. The simple ramification condition is also a requirement in the proof of \cite{EGMO21} for higher genus spectral curves, but not in the proof of \cite{BE17} for genus $0$ spectral curves, which uses the higher ramification generalization of \cite{BE13}.\par

This quantization method however raises an intriguing question. There are many ways to quantize a plane curve as in \eqref{eq:quant} -- one needs to choose an ordering of the non-commutative operators $\hat{x}, \hat{y}$, and one could add further $\hslash$-corrections. Nonetheless, topological recursion seems to "select" a particular quantization. Moreover, it is often not the naively expected one, such as the normal-ordered quantization. It may not even be the quantization in \emph{any} ordering! (I.e. it may include further $\hslash$-corrections.) Why is topological recursion selecting such particular quantizations?

To make things concrete, consider the following spectral curve:
\begin{equation}\label{eq:ss1}
x^{r-1} y^r - 1 =0.
\end{equation}
This is the $s=1$ case in the notation of \cite{BBCCN18}. This spectral curve falls into the class considered in \cite{BE17}. In there, it is shown that the quantization procedure above gives rise to the following quantum curve:
\begin{equation}
\left( \hat{y} \hat{x} \right)^{r-1} \hat{y} - 1,
\end{equation}
which is of course a quantization of the spectral curve, but a rather strange one! For instance, it is not the normal-ordered quantization, which one could naively expect to be singled out by topological recursion. Why is topological recursion selecting this particular quantization? Is it possible to modify the quantization procedure to obtain other choices of quantization of the spectral curve?

As explained in \cite{BE17}, there is a freedom in the quantization procedure, which is in the choice of integration divisor $D$. For some spectral curves, constructing wave-functions with different choices of integration divisors does produce solutions to distinct quantizations of the spectral curve. However, this freedom is rather limited, and is not sufficient to obtain all possible choices of quantizations. For instance, in most cases, one should take $D$ to consist of a pole of $x$ (understood as a meromorphic function on the normalization of the plane curve); but for the spectral curve \eqref{eq:ss1}, there is only one such choice (the pole at $\infty$), and thus this freedom cannot account for other choices of quantizations.

The motivation behind this paper is to figure out how we can modify topological recursion and its corresponding quantization procedure to obtain more general quantizations of spectral curves. We propose a "shifted" version of topological recursion and loop equations, which, as we show, allows us to reconstruct the WKB solution to more general quantizations of the spectral curve.

\subsection{A trifecta of viewpoints}

We will also approach this question from two other different viewpoints: from the point of view of WKB solutions to differential systems, and from the reformulation of topological recursion as Airy ideals coming from representations of $\mathcal{W}$-algebras.

On the one hand, it is natural to consider the question of quantization from the point of view of WKB solutions of differential systems. In this context, we can start with any quantization of a spectral curve, which produces a differential system. The question then becomes: for what such quantizations can we reconstruct the WKB solution through topological recursion? This question was answered in part in \cite{BBE15,BEM18,IMS18}: if the system satisfies certain conditions, called the topological type property, the solution is given by topological recursion. In \cite{BEM17}, Belliard--Eynard--Marchal formulated a set of six assumptions that imply the topological type property, and proved that they hold in many natural examples. In this context, what we show is that we can sharpen one of the assumptions of \cite{BEM17}; we obtain a larger class of quantum curves for which the WKB solution can be reconstructed recursively, but it is now via our proposed shifted topological recursion.

On the other hand, topological recursion was reformulated in an algebraic language by Kontsevich and Soibelman \cite{KS17}, who showed that the $ \omega_{g,n}$ can be assembled in a partition function which is annihilated by a particular set of differential operators called an Airy ideal (also known as Airy structure). These Airy ideals encode the fact that topological recursion gives a solution to loop equations \cite{BEO15,BS17} which only have prescribed poles and holomorphic components. For the original topological recursion of \cite{EO07}, the Airy ideal can be obtained as a representation of a number of copies of the Virasoro algebra, one for each ramification point of the spectral curve. The partition function can then be thought of as a vacuum vector (or highest weight vector with weight zero) of the Virasoro algebra. The strength of the Airy ideal formulation is that it gives an immediate proof that there exists symmetric solutions to topological recursion, something which otherwise is quite difficult to prove directly from topological recursion. 

This approach via Airy ideals was generalized to higher order ramification points in \cite{BBCCN18,BKS23} (also allowing poles of $y$ at the ramification points, keeping $\omega_{0,1} =y \, dx$ holomorphic). In particular, this approach proves that the topological recursion formulas obtained in \cite{BE13,DN18} have symmetric solutions. Surprisingly, requiring symmetry gave conditions on the kind of ramification orders $r$ and pole orders $r-s$ of $ y$ that are allowed: one must have $ r = \pm 1 \pmod{s}$, otherwise already $\omega_{0,3}$ is non-symmetric.

In this generalization, the Airy ideals are obtained as representations of $\mathcal{W}(\mathfrak{gl}_r)$-algebras (one copy for each ramification point of order $r$). The partition function is again a vacuum vector (or highest weight vector with weight zero). From this point of view, what we show is that our proposed shifted topological recursion (and shifted loop equations) arise by simply considering more general partition functions obtained as highest weight vectors with non-zero weights. We show that those also form Airy ideals, and thus we know that our proposed shifted topological recursion has a symmetric solution. As this formulation in terms of Airy ideals is clean and simple, this will be our starting point.

\subsection{Contributions of this paper}

We mostly investigate the spectral curves which can be parametrized by

\begin{equation}
  \begin{cases}
    x 
    &= 
    z^r 
    \\
    y 
    &= 
    z^{s-r}
  \end{cases}
\end{equation}
for some $ r \geq 2$, $ s \geq 1$. These curves are the the local model for any (smooth) ramification point of a spectral curve, and from the analysis of \cite{BBCCN18}, we know that topological recursion is well-behaved on these curves if and only if $ r = \pm 1 \pmod{s}$. Their plane curve equation is
\begin{equation}
\label{rAirySC}
  P(x,y) = x - y^r
\end{equation}
for $ s= r +1$ and
\begin{equation}
\label{rsSC}
  P(x,y) = x^{r-s} y^r - 1
\end{equation}
else. We call those curves the $(r,s)$-spectral curves.

We are particularly interested in possible quantizations of these curves. From the plane curves \cref{rAirySC,rsSC}, one may think that in the $ s = r+1 $ case, there is no ambiguity in quantization, while in the other cases, there are several possible orderings of the quantization of the monomial $ x^{r-s} y^r$.\par

We find, rather, that if $r = 1 \pmod{s}$, we can obtain infinite-dimensional families of quantum curves, whose solutions can be calculated via an explicit and consistent\footnote{i.e. producing symmetric multidifferentials} modification of the topological recursion formula which we call "shifted topological recursion". In the particular case $ s = 1$, this family is even larger than in the other cases. However, only for $ s \in \{ 1, r-1\}$ do these families contain all possible orderings of the naive quantization.

We start our investigation from the point of view of Airy ideals. In \cref{s:shifted}, we investigate the theory of Airy ideals for different $ (r,s)$. The $(r,s)$-Airy ideals corresponding to topological recursion on the $(r,s)$-spectral curves were constructed in \cite{BBCCN18} as representations of $\mathcal{W}(\mathfrak{gl}_r)$-algebras. The corresponding partition functions are vacuum vectors or highest weight vectors with weight zero. We show that we can construct more general families of Airy ideals, which we call ``shifted $(r,s)$-Airy ideals'', whose partition functions correspond to highest weight vectors with non-zero weights of the $\mathcal{W}(\mathfrak{gl}_r)$-algebras (\cref{t:shifts}).

More precisely, for the cases $(r,s) = (r,1)$, in the usual construction the partition function is annihilated by all the non-negative modes $W^i_k$, $k \geq 0$, $1 \leq i \leq r$ of the generators of the $\mathcal{W}(\mathfrak{gl}_r)$-algebra. It is thus a highest weight vector with highest weight zero. We show that we can construct general highest weight vectors from Airy ideals; the highest weights, which correspond to the non-zero weights of the zero modes $W^i_0$, appear in the differential operators as $r$ scalars $ S_i \in \hslash \C \llbracket \hslash \rrbracket $, for $ 1 \leq i \leq r$. For the cases  $r = 1 \pmod{s}$ with $s \geq 2$, the partition function is now annihilated by some negative modes of the generators as well as the non-negative modes, and the construction is more limited. We show that we can construct more general Airy ideals, but the only freedom is in giving a non-zero weight $S_1 \in \hslash \C \llbracket \hslash \rrbracket$ to the zero mode $W^1_0$ of the conformal weight $1$ generator. Finally, for $r = -1 \pmod{s}$ and $s \geq 3$, we show that we cannot introduce any non-zero weights from Airy ideals.

In \cref{s:shiftedle}, we generalize the translation from Airy ideals to topological recursion (via loop equations) to these shifted $(r,s)$-Airy ideals. The shifted Airy ideals are equivalent to a modification of the loop equations for correlators $\omega_{g,n}$, which we call "shifted loop equations" (\cref{p:shifteloopeq}). We can solve these shifted loop equations in the same way as topological recursion solves the usual loop equations, and we obtain a variation on the topological recursion formula, which we call "shifted topological recursion" (\cref{ShiftedTR}). The only difference with the usual topological recursion formula is that the highest weights introduce corrections to the correlators $ \omega_{g,1}$, which have to be added into the topological recursion formula explicitly. Aside from this, the recursive structure remains the same. Moreover, because we obtain shifted topological recursion starting from Airy ideals, it is guaranteed to produce symmetric multidifferentials.

Now that we have more general shifted loop equations and topological recursion for the $(r,s)$-spectral curves, we can ask whether the corresponding quantization formalism produces wave-functions for more general quantizations of the $(r,s)$-spectral curves. We answer this question in \cref{s:qcurves}. We generalize the construction of quantum curves from \cite{BE17} to the shifted loop equations. We find the appropriate system of differential equations and hence quantum curves that annihilate the wave-function constructed from the correlators produced by shifted topological recursion (\cref{PsiDifferentialSystem} and \cref{t:QC}). In this way, we obtain families of quantizations of the $(r,s)$-spectral curves. In particular, for $s=1$ and $s=r-1$, we obtain families that contain all possible quantizations of the spectral curve corresponding to distinct choices of ordering of the non-commutative operators $\hat{x}$ and $\hat{y}$.

Finally, in \cref{s:diffsyst} we close the loop by considering the converse question: given a quantum curve, or rather the associated differential system, when can its solution be constructed by topological recursion? I.e., when is the system of topological type? We find that after sharpening one of the assumptions of \cite{BEM17}, the conditions for this to work are exactly the same as the ones obtained in the Airy ideals framework, namely that $r= 1 \pmod{s}$. In fact, we generalize the construction to allow for highest weight shifts in the differential systems, and we obtain that the differential system has a WKB solution constructed from shifted topological recursion under exactly the same conditions as in \cref{t:shifts} (\cref{t:ds}). 
We identify certain key elements in both constructions, explaining the correspondence between the two languages.

Moreover, the aim of this paper is also partly expository. We connect several important viewpoints on topological recursion: the original geometric definition via residues of multidifferentials, the algebraic formulation via Airy ideals, and the integrable aspect via the WKB analysis of an $ \hslash$-connection of topological type. The central concept which connects all of these points of view is that of loop equations, and they will appear in different guises throughout the paper. This trifecta of viewpoints is represented pictorially in \cref{f:intro}.

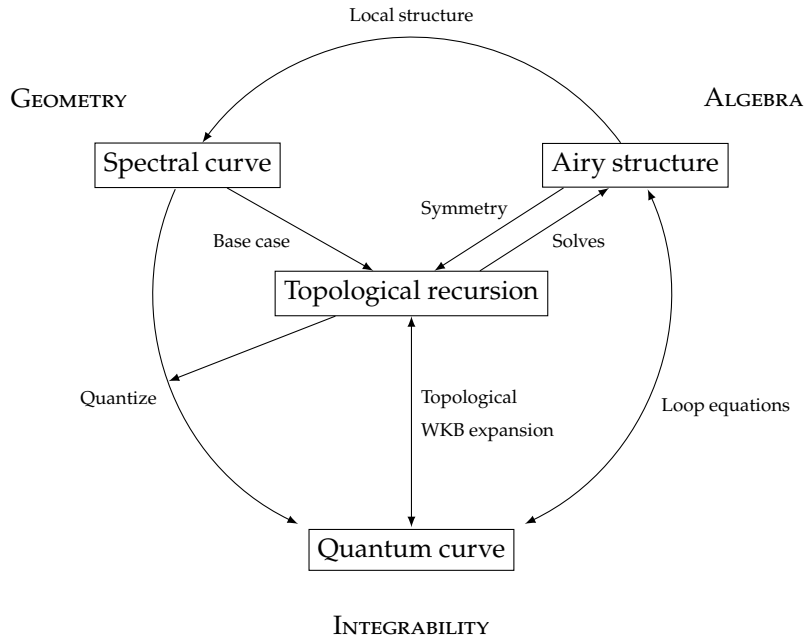
\begin{figure}[ht]
  \begin{tikzpicture}
    \def \radius {3.4cm}
    \def \margin {6} 
    \def \inside {0.3cm}

    \node[draw] at ({30}:\radius) {Airy ideal};
    \node at (30:\radius+1.8cm) {\textsc{Algebra}};
    \draw[->, >=latex] ({30+\margin}:\radius) 
      arc ({30+\margin}:{150-\margin}:\radius);
    \node[draw] at ({150}:\radius) {Spectral curve};
    \node at (150:\radius+1.8cm) {\textsc{Geometry}};
    \draw[->, >=latex] ({150+\margin}:\radius) 
      arc ({150+\margin}:{270-2\margin}:\radius);
    \node[draw] at ({270}:\radius) {Quantum curve};
    \node at (270:\radius+1cm) {\textsc{Integrability}};
    \draw[<->, >=latex] ({270+2\margin}:\radius) 
      arc ({270+2\margin}:{390-\margin}:\radius);
    \node[draw] at (0,0) {Topological recursion};
    \draw[<->, >=latex] (0cm,-\radius+\inside) to (0cm,-\inside);
    \node[label={[align=left] \scriptsize Topological\\ \scriptsize WKB \scriptsize expansion}] at (1.3cm,-0.65*\radius) {};
    \draw[->, >=latex] (-1cm,-\inside) to (200:\radius);
    \node at ({200}:\radius+0.7cm) {\scriptsize Quantize};
    \node at ({340}:\radius+1.2cm) {\scriptsize Loop equations};
    \draw[->, >=latex] (0.9cm,\inside) to (2.6cm,1.4cm);
    \node at (2.2cm,0.7cm) {\scriptsize Solves};
    \node at (90:\radius+0.3cm) {\scriptsize Local structure};
    \draw[->, >=latex] (150:\radius-2*\inside) to (-0.5cm,\inside);
    \node at (-2.1,0.7cm) {\scriptsize Base case};
    \draw [->, >=latex] (2cm,1.4cm) to (0.3cm,\inside);
    \node at (0.5cm,1.1cm) {\scriptsize Symmetry};
  \end{tikzpicture}
  \caption{A trifecta of viewpoints.}
  \label{f:intro}
\end{figure}

\subsection{Open questions}

We have only considered in detail very specific spectral curves, with a single ramification point relevant for topological recursion. These give all of the commonly considered local models, but the more general global situation still poses significant challenges, at least at a computational level. We have also not considered higher-genus spectral curves, as \cite{EGMO21}, and the required resummations there may also pose problems.

In the semi-simple case, i.e. the case where all ramifications are of type $ (r,s) = (2,3)$, local topological recursion is identified \cite{DOSS14} with Givental's reconstruction of cohomological field theories (CohFTs) \cite{Giv01a}, which reconstructs all semi-simple CohFTs from genus $0$ data \cite{Tel12}, starting from the correspondence between the Airy curve $ x - y^2 = 0$ and the unit CohFT. In the cases $ (r,r+1) $ and $(r,r-1)$, the curves also correspond to CohFTs, namely the $r$-spin Witten class \cite{Wit93,FSZ10,BBCCN18} and the class $ \Theta^r$ \cite{Nor23a,CGG22}, and the Givental group action still acts on such CohFTs \cite{FSZ10} and can still be identified with topological recursion by \cite{DOSS14}. However, our results show that for $ r = 1 \pmod{s}$, topological recursion can actually get corrections in positive genus, and this suggests that in these cases the Givental group action has to be extended as well. Therefore, in these cases an analogue of Teleman's reconstruction theorem may not hold, as the Givental group does not act transitively.

In a similar direction, an open question is to find a geometric interpretation for the correlators $\omega_{g,n}$ calculated by topological recursion on the $(r,s)$-spectral curves, or equivalently, for the partition function of the $(r,s)$-Airy ideals. As mentioned above, in the cases with $s=r+1$ and $s=r-1$, such an interpretation is known: the partition function is the descendent potential of the $r$-spin Witten class and the class $\Theta^r$ respectively \cite{CGG22}. However, it remains unknown for other choices of $s$. It is perhaps even more interesting to study whether there is a geometric interpretation for shifted $(r,s)$-Airy ideals, in particular in the case $s=1$, where we can shift all zero modes. In fact, in upcoming work, one of the authors, in collaboration with N. K. Chidambaram, A. Giacchetto and S. Shadrin, shows that for $s=1$, and for a specific choice of the highest weights, the partition function is the descendant potential of the $\Theta^{r,1}$-class proposed in \cite{CGG22} (see Remark 2.10).

\subsection{Notation}

We introduce $\hbar$ along the conventions in \cite{BCJ22}. To connect to other literature on Airy structures, cf. \cite[Remark~2.16]{BCJ22}.\par
We use the convention that $\mathbb{N} = \{0,1,2,\ldots \}$ and $\mathbb{N}^* = \{1,2,3,\ldots \}$. We write $[r] = \{1,\ldots, r\}$. For a set $ N$ and a variable $z$, we write $ z_N \coloneq \{ z_n \, | \, n \in N \}$.\par
We consider fields in vertex operator algebras as differential forms of degree equal to the conformal weight of the state. I.e. if in a VOA $V$, the state $ v \in V$ has conformal weight $ \Delta$, then we index its field by
\begin{equation}
  Y(v;x) = \sum_{k \in \Z} v_k \frac{(dx)^\Delta}{x^{\Delta + k}}\,.
\end{equation}
We use $x$ for the variable instead of the conventional $ z$, because this conforms with our interpretation via the spectral curve of topological recursion, cf. \cite{BKS23}.\par
When considering a spectral curve with local coordinate $ z$, and functions $  x(z)$ and $ y(z)$, we may write $ x_j = x(z_j)$ and $ y_j = y(z_j)$ to lighten notation.

\section{Shifted \texorpdfstring{$(r,s)$}{rs}-Airy ideals and highest weight vectors}

\label{s:shifted}

In this section, we explain how Airy ideals \cite{KS17}, by which we mean higher quantum Airy structures with crosscaps in the sense of \cite{BBCCN18} (or rather the associated Airy ideals \cite{BCJ22}), can be used to reconstruct highest weight vectors for $\mathcal{W}(\mathfrak{gl}_r)$ at self-dual level. This involves a generalization of the $(r,s)$-Airy ideals introduced in \cite{BBCCN18}, which we call \emph{shifted $(r,s)$-Airy ideals}. We assume familiarity with the relevant concepts in these papers, and only refer to main results. We follow the approach to Airy ideals presented in \cite{Bo24}, following \cite{BCJ22}.

\subsection{Airy ideals}

Let us start by reviewing the definition of Airy ideals (also called Airy structures). We follow \cite{Bo24}; proofs of the results stated here can be obtained either there or in \cite{BCJ22,KS17}.

\subsubsection{The Rees Weyl algebra}

Let $A$ be a finite or countably infinite index set. We use the notation $x_A$ for the set of variables $\{x_a \}_{a \in A}$, and $\partial_A$ for the set of differential operators $\left\{ \frac{\partial}{\partial x_a} \right \}_{a \in A}$.  The Weyl algebra $ \mathbb{C}[x_A]\langle \partial_A \rangle$ is the algebra of differential operators with polynomial coefficients. We define the completed Weyl algebra $\mathcal{D}_A$ to be the completion of the Weyl algebra, where we allow infinite sums in the derivatives (when $A$ is a countably infinite index set) but not in the variables. 

$\mathcal{D}_A$ has many filtrations, one of which is the Bernstein filtration (see Definition 2.3 in \cite{Bo24}). Using this filtration, we construct a graded algebra via the Rees construction:
\begin{definition}\label{d:rees}
The \emph{Rees Weyl algebra} $\mathcal{D}_A^\hslash$ associated to $\mathcal{D}_A$ with the Bernstein filtration is
\begin{equation}
\mathcal{D}_A^\hslash = \bigoplus_{n \in \mathbb{N}} \hslash^n F_n \mathcal{D}_A,
\end{equation}
where the $F_n \mathcal{D}_A$ refer to the subspaces in the Bernstein filtration of $\mathcal{D}_A$.
\end{definition}

When $A$ is countably infinite, we want to be able to take infinite linear combinations of operators $P_a$ without divergent sums appearing. To this end, we define the notion of a bounded collection of differential operators:
\begin{definition}\label{def:bounded}
Let $I$ be a finite or countably infinite index set, and $\{P_i\}_{i \in I}$ a collection of differential operators $P_i \in \DA$ of the form
\begin{equation}
P_i = \sum_{n \in \mathbb N} \hslash^n \sum_{\substack{m,k \in \mathbb N \\ m+k = n} }\sum_{a_1,\ldots, a_m \in A} p^{(n,k)}_{i;a_1,\ldots, a_m}(x_A) \partial_{a_1}\ldots \partial_{a_m}.
\end{equation}
We say that the collection is \emph{bounded} if, for all fixed choices of indices $ a_1,\ldots, a_m, n  $ and $ k $, the polynomials  $ p^{(n,k)}_{i;a_1,\ldots, a_m}(x_A) $ vanish for all but finitely many indices $ i \in I $.
\end{definition} 
It is easy to see that for any bounded collection $\{P_i\}_{i \in I}$, linear combinations $\sum_{i \in I} c_i P_i$ for any $c_i \in \DAh$ are well defined operators in $\DAh$, regardless of whether $I$ is finite or countably infinite.

\subsubsection{Airy ideals}

We now define the notion of Airy ideals, which is a particular class of left ideals in $\DAh$. 

 \begin{definition}\label{d:airy}
 Let $\mathcal{I} \subseteq \DAh$ be a left ideal. We say that it is an \emph{Airy ideal} (also known as \emph{Airy structure}) if there exists a bounded generating set $\{H_a \}_{a \in A}$ for $\mathcal{I}$ such that:\footnote{We abuse notation slightly here. We say that $\mathcal{I}$ is generated by the $H_a$, even though in standard terminology the ideal generated by the $H_a$ should only contain finite linear combinations of the generators. Here we allow finite and infinite (when $A$ is countably infinite) linear combinations, which is allowed since the collection $\{H_a \}_{a \in A}$ is bounded.}
  \begin{enumerate}
    \item The operators $H_a$ take the form
    \begin{equation}\label{eq:form}
       H_a= \hslash \partial_a + \hslash p_a(x_A) + O(\hslash^2),
    \end{equation}
    where the $p_a(x_A)$ are linear polynomials.
    \item The left ideal $\mathcal{I}$ satisfies the property:
    \begin{equation}
      [ \mathcal{I}, \mathcal{I}] \subseteq \hslash^2 \mathcal{I}.
    \end{equation}
  \end{enumerate}
 \end{definition}    
 
 \subsubsection{Partition function}
 
 The main reason that Airy ideals are interesting is because they are annihilator ideals for some partition functions. 
 
 \begin{definition}\label{d:pf}
A \emph{partition function} in the set of variables $x_A$ is an expression of the form
\begin{equation}\label{eq:pfFgncoeff}
Z = \exp\left( \sum_{g \in \frac{1}{2} \mathbb{N}, n \in \mathbb{N}^*} \frac{\hslash^{2g-2+n}}{n!}  \sum_{k_1, \ldots, k_n \in A}  F_{g,n}[k_1, \ldots,k_n]  x_{k_1} \cdots x_{k_n}\right).
\end{equation}
We say that it is \emph{stable} if $F_{0,1}[k_1]=F_{0,2}[k_1,k_2]=F_{\frac{1}{2},1}[k_1] =0$, \emph{semistable} if $F_{0,1}[k_1]=0$, and unstable otherwise.
 \end{definition}
 
 Recall the definition of an annihilator ideal:
 
 \begin{definition}
Let $Z$ be a partition function as in \eqref{eq:pfFgncoeff}. The \emph{annihilator ideal} $\mathcal{I} = \text{Ann}_{\DAh}(Z)$ of $Z$ in $\DAh$ is the left ideal in $\DAh$ defined by
\begin{equation}
 \text{Ann}_{\DAh}(Z) = \{ P \in \DAh~|~ P Z = 0 \}. 
\end{equation}
\end{definition}

The main result in the theory of Airy ideals, which was originally proved in \cite{KS17}, is the following theorem:

 \begin{theorem}\label{t:airy}
  Let $\mathcal{I} \subset \DAh$ be an Airy ideal. Then there exists a unique partition function $Z$ of the form \eqref{eq:pfFgncoeff} such that $\mathcal{I}$ is the annihilator ideal of $Z$ in $\DAh$. Moreover, $Z$ is semistable, and if $p_a(x_A) = 0$ for all $a \in A$, then it is stable.
 \end{theorem}
 
 In other words, given any Airy ideal $\I$, there always exists a unique partition function $Z$ such that $\I Z =0$. Since the operators $H_a$ that generate $\I$ are finite degree in $\hslash$, the differential constraints $H_a Z = 0$ for all $a \in A$ give rise to recursion relations for the $F_{g,n}[k_1,\ldots,k_n]$ that can be used to reconstruct $Z$ uniquely.
 
 \begin{remark}
 In the literature on Airy structures, the $O(\hslash)$ terms $\hslash p_a(x_A)$ are usually omitted from the operators $H_a$ in \cref{d:airy}. The resulting partition function is then always stable (that is, the sum in \eqref{eq:pfFgncoeff} starts with $2g-2+n>0$). It is straightforward however to extend the proof of \cref{t:airy} (for instance, following step-by-step the approach in \cite{BCJ22}) to the case of non-zero linear polynomials $\hslash p_a(x_A)$, with the only difference being that the resulting partition function becomes semistable (i.e. with the sum starting with $2g-2+n \geq 0$). 
 \end{remark}
 
 \subsubsection{Airy ideals in universal enveloping algebras}
 
 Many Airy ideals are constructed via representations of either Lie algebras or non-linear Lie algebras -- see for instance \cite{BBCCN18}. We briefly explain the main idea.
 
 Let $\mathfrak{g}$ be either a Lie algebra or a non-linear Lie algebra (see for instance Section 3 of \cite{DK05} for a precise definition of non-linear Lie algebras), and $U(\mathfrak{g})$ the universal enveloping algebra. Suppose that there is an exhaustive ascending filtration on $U(\mathfrak{g})$ (such as the filtration by conformal weight); then we construct the Rees universal enveloping algebra $U^\hslash(\mathfrak{g}) = \bigoplus_{n \in \mathbb{N}} \hslash^n F_n U(\mathfrak{g})$ using the Rees construction as in Definition \ref{d:rees}. 
 
 To construct Airy ideals, we proceed as follows:
\begin{lemma}\label{l:uea}
Let $\rho: U^\hslash(\mathfrak{g}) \to \DAh$ be a representation of the Rees enveloping algebra in the Rees Weyl algebra, for some index set $A$. Let $\mathcal{I}_{U^\hslash} \subseteq U^\hslash(\mathfrak{g})$ be a left ideal in $U^\hslash(\mathfrak{g})$, and $\mathcal{I} = \DAh \rho(\mathcal{I}_{U^\hslash}) \subseteq \DAh$ be the corresponding left ideal in $\DAh$ generated by $\rho(\mathcal{I}_{U^\hslash})$.

Suppose that  $\mathcal{I}_{U^\hslash}$ satisfies the property $[\mathcal{I}_{U^\hslash},\mathcal{I}_{U^\hslash}] \subseteq \hslash^2 \mathcal{I}_{U^\hslash}$, and that there exists a generating set $\{H_a \}_{a \in A}$ for $\mathcal{I}_{U^\hslash}$ such that $\rho(H_a) = \hslash \partial_a + \hslash p_a(x_A)+ O(\hslash^2)$,
where the $p_a(x_A)$ are linear polynomials, 
 and such that the collection $\{\rho(H_a)\}_{a \in A}$ is bounded. Then $\mathcal{I}$ is an Airy ideal.
\end{lemma}

In this construction we see that the two conditions in the definition of Airy ideals, \cref{d:airy}, are obtained independently. The condition $[\mathcal{I}_{U^\hslash},\mathcal{I}_{U^\hslash}] \subseteq \hslash^2 \mathcal{I}_{U^\hslash}$ is a condition on the left ideal $\mathcal{I}_{U^\hslash} \subseteq U^\hslash(\mathfrak{g})$  in the Rees universal enveloping algebra, while the second condition that there exists a generating set $\{H_a \}_{a \in A}$ for $\mathcal{I}_{U^\hslash}$ such that $\rho(H_a) = \hslash \partial_a + \hslash p_a(x_A)+O(\hslash^2)$ depends on the choice of representation.

The condition $[\mathcal{I}_{U^\hslash},\mathcal{I}_{U^\hslash}] \subseteq \hslash^2\mathcal{I}_{U^\hslash}$ is in fact fairly easy to satisfy. We first define an operation that maps elements of $U(\mathfrak{g})$ to elements of $U^\hslash(\mathfrak{g})$:

\begin{definition}\label{d:homogenization}
Let 
 $p \in U(\mathfrak{g})$, and let  $i = \min\{ k \in \mathbb{N}~|~ p \in F_k U(\mathfrak{g}) \}$. We define the \emph{homogenization} $h(p)$ of $p$ to be $h(p) = \hslash^i p \in U^\hslash(\mathfrak{g})$. We define the homogenization $h(\mathcal{I}_U)$ of a left ideal $\mathcal{I}_U \subseteq U(\mathfrak{g})$ to be the left ideal in $U^\hslash(\mathfrak{g})$ generated by all homogenized elements $h(p)$, $p \in \mathcal{I}_U$. 
 \end{definition}
 
 Then we have the following simple lemma:
 
\begin{lemma}\label{l:hom}
Let $\mathcal{I}_U \subseteq U(\mathfrak{g})$ be a left ideal. Then its homogenization $h(\mathcal{I}_U) \subseteq U^\hslash(\mathfrak{g})$ satisfies $[ h(\mathcal{I}_U), h(\mathcal{I}_U)] \subseteq \hslash^2 h(\mathcal{I}_U)$.
\end{lemma}

Thus any left ideal $\mathcal{I}_{U^\hslash} \subseteq U^\hslash(\mathfrak{g})$ that is obtained as the homogenization of a left ideal in $U(\mathfrak{g})$ automatically satisfies $[\mathcal{I}_{U^\hslash},\mathcal{I}_{U^\hslash}] \subseteq \hslash^2\mathcal{I}_{U^\hslash}$. This gives a clear recipe on how to obtain Airy ideals from universal enveloping algebras.

\begin{enumerate}
\item We start with a left ideal $\mathcal{I}_U \subseteq U(\mathfrak{g})$ or, equivalently, a cyclic left module $M \simeq U(\mathfrak{g})/ \mathcal{I}_U$ generated by a vector $v$ whose annihilator is $\mathcal{I}_U = \text{Ann}_{U(\mathfrak{g})}(v)$.
\item We construct the homogenization $\mathcal{I}_{U^\hslash}= h(\mathcal{I}_U)$, which is a left ideal in $U^\hslash(\mathfrak{g})$.  By construction, we know that $[\mathcal{I}_{U^\hslash}, \mathcal{I}_{U^\hslash}] \subseteq \hslash^2 \mathcal{I}_{U^\hslash}$. From the point of view of modules, we obtain a cyclic left module $M[\hslash] \simeq U^\hslash(\mathfrak{g})/\mathcal{I}_{U^\hslash}$ generated by the vector $v$ and where $\hslash$ acts by multiplication; the annihilator of $v$ in $U^\hslash(\mathfrak{g})$ is $\mathcal{I}_{U^\hslash} = \text{Ann}_{U^\hslash(\mathfrak{g})}(v)$.
\item We find a representation $\rho: U^\hslash(\mathfrak{g}) \to \DAh$, for some index set $A$, such that there exists a  generating set $\{H_a\}_{a \in A}$ for $\mathcal{I}_{U^\hslash}$ with $\rho(H_a) = \hslash \partial_a + \hslash p_a(x_A)+ O(\hslash^2)$ and the collection $\{\rho(H_a)\}_{a \in A}$ bounded.
\end{enumerate}
By Lemma \ref{l:uea}, the left ideal $\mathcal{I} \subseteq \DAh$ generated by $\{\rho(H_a)\}_{a \in A}$ is an Airy ideal.

\subsection{\texorpdfstring{$(r,s)$}{rs}-Airy ideals}

\label{s:rsairystruct}

In this section we apply the ideas of the previous section to construct Airy ideals from the universal enveloping algebra of the modes of the strong generators of the $\mathcal{W}(\mathfrak{gl}_r)$-algebra at self-dual level. We follow the three-step approach explained above. This construction was originally presented in \cite{BBCCN18}.

\subsubsection{The \texorpdfstring{$\mathcal{W}(\mathfrak{gl}_r)$}{Wglr}-algebra at self-dual level}

Let us introduce the $\mathcal{W}(\mathfrak{gl}_r)$-algebra at self-dual level via its realization as a subalgebra of the Heisenberg VOA $\mathcal{H}(\mathfrak{gl}_r)$.

Let $\mathfrak{h} \subset \mathfrak{gl}_r$ be the Cartan subalgebra with orthogonal canonical basis $\{\chi^j \}_{j=1}^r$.
The Heisenberg VOA is the vertex operator algebra freely generated by the vectors $\chi_{-1}^j |0 \rangle$, $j=1,\ldots,r$, where $| 0 \rangle$ is the vacuum vector. We define the fields
\begin{equation}
J^j(z) = Y(\chi_{-1}^j |0 \rangle, z) = \sum_{n \in \mathbb{Z}} J_n^j \frac{dz}{z^{n+1}}.
\end{equation}

The $\mathcal{W}(\mathfrak{gl}_r)$-algebra at self-dual level is the VOA strongly freely generated by the vectors 
\begin{equation}\label{WrGenerators}
  w^j = e_j (\chi^1_{-1}, \dotsc, \chi^r_{-1})|0\> \,, \qquad j\in [r],
\end{equation} 
where $e_j$ denotes the $j$'th elementary symmetric polynomial. The corresponding fields take the form
\begin{equation}\label{eq:expl1}
W^j(z) = Y( e_j (\chi^1_{-1}, \dotsc, \chi^r_{-1})|0\>, z) =e_j\left( J^{1}(z), \ldots, J^r(z) \right)  
= \sum_{n \in \mathbb{Z}} W^j_n \frac{dz^j}{z^{n+j}}.
\end{equation}
This gives the explicit relation
\begin{equation}\label{eq:explicit}
W^j_n = \sum_{1 \leq i_1 < \ldots < i_j \leq r} \sum_{m_1 + \ldots + m_j = n}\left( \prod_{k=1}^j J^{i_k}_{m_k} \right).
\end{equation}

The modes $\{W^j_n\}_{j \in [r], n \in \mathbb{Z}}$ of the strong generators span a non-linear Lie algebra. Let us denote by $U_r$ the universal enveloping algebra of the modes.

There is a natural filtration on $U_r$ by conformal weight, where the modes $W^j_n$ have degree $j$. More precisely, the subspaces in the filtration $F_n  U_r$ consist of sums of monomials of the forms $W^{j_1}_{n_1} \cdots W^{j_k}_{n_k}$ with $j_1 + \ldots + j_k \leq n$. We use this filtration to construct the Rees universal enveloping algebra $U^\hslash_r$, which in essence amounts to redefining $W^j_n \mapsto \hslash^j W^j_n$. 

\subsubsection{A few preliminary lemmas}

We prove a few preliminary lemmas that will be useful shortly. We first prove a simple result about partitions and elementary symmetric polynomials.

\begin{definition}\label{d:ps}
Let $\lambda = (\lambda_1, \lambda_2, \ldots, \lambda_p)$ be an integer partition of $r$, that is, $\lambda_1 \geq \lambda_2 \geq \ldots \geq \lambda_p \geq 1$ and $\sum_{i=1}^p \lambda_i = r$. We define the partial sums $\mu_k = \sum_{i=1}^k \lambda_i$ for $k \in [p]$. By convention we set $\mu_0 = 0$.
\end{definition}

\begin{lemma}\label{l:esp}
Let $\lambda = (\lambda_1, \lambda_2, \ldots, \lambda_p)$ be an integer partition of $r$. Let $e_j$ be the $j$'th symmetric polynomial. Then:
\begin{equation}
e_j(x_1, \ldots, x_r) = \sum_{j_1=0}^{\lambda_1} \cdots \sum_{j_p=0}^{\lambda_p} \delta_{j_1+\ldots+j_p,j}  \prod_{k=1}^p e_{j_k}(x_{\mu_{k-1}+1},\ldots,x_{\mu_k}),
\end{equation}
where $\delta_{m,n}$ is the Kronecker delta.
\end{lemma}

\begin{proof}
This follows directly from the generating function for elementary symmetric polynomials. We know that
\begin{equation}
G(z;x_1,\ldots, x_r) := \prod_{i=1}^r (1 + x_i z) = \sum_{j=0}^r e_j(x_1,\ldots,x_r) z^j,
\end{equation}
where $e_0(x_1,\ldots,x_r) = 1$.
But:
\begin{align}
G(z;x_1,\ldots, x_r) =&  \prod_{k=1}^{p} G(z; x_{\mu_{k-1}+1}, \ldots, x_{\mu_{k}})\\
=& \prod_{k=1}^{p}  \left(  \sum_{j_k=0}^{\lambda_k} e_{j_k}(x_{\mu_{k-1}+1},\ldots,x_{\mu_k}) z^{j_k}\right)\\
=& \sum_{j=0}^r \left( \sum_{j_1=0}^{\lambda_1} \cdots \sum_{j_p=0}^{\lambda_p} \delta_{j_1+\ldots+j_p,j}  \prod_{k=1}^p e_{j_k}(x_{\mu_{k-1}+1},\ldots,x_{\mu_k})  \right) z^j.
\end{align}
\end{proof}

Using this lemma we can exploit the realization of the $\mathcal{W}(\mathfrak{gl}_r)$-algebra at self-dual level in terms of elementary symmetric polynomials to see that there is a natural embedding of $\mathcal{W}(\mathfrak{gl}_r)$ in $\mathcal{W}(\mathfrak{gl}_{\lambda_1}) \times \ldots \times \mathcal{W}(\mathfrak{gl}_{\lambda_p})$ for any integer partition $\lambda$ of $r$.

\begin{lemma}\label{l:embed}
Let $\lambda = (\lambda_1, \lambda_2, \ldots, \lambda_p)$ be an integer partition of $r$.  Let $W^j(z)$, $j \in [r]$ be the strong generators of $\mathcal{W}(\mathfrak{gl}_r)$ and $W^j_m$ their modes. There is a natural embedding $\mathcal{W}(\mathfrak{gl}_r) \subset \mathcal{W}(\mathfrak{gl}_{\lambda_1}) \times \ldots \times \mathcal{W}(\mathfrak{gl}_{\lambda_p})$ given by the explicit formula for the modes:
\begin{equation}\label{eq:embedding}
W^j_m =\sum_{j_1=0}^{\lambda_1} \cdots \sum_{j_p=0}^{\lambda_p} \delta_{j_1+\ldots+j_p,j}   \sum_{m_1+ \ldots + m_p = m} \left( \prod_{k=1}^p X^{k,j_k}_{m_k}  \right),
\end{equation}
where the $X^{k,j_k}_{m_k}$, $j_k \in [\lambda_k]$, $m_k \in \mathbb{Z}$  are the modes of the strong generators of the $\mathcal{W}(\mathfrak{gl}_{\lambda_k}) $ factors.  By convention we set $X^{k,0}_{m_k} = \delta_{m_k,0}$.
\end{lemma}

\begin{proof}
This follows from \cref{l:esp}. 
By \eqref{eq:expl1}, and using \cref{l:esp}, we get:
\begin{align}
W^j(z) =& e_j\left( J^{1}(z), \ldots, J^r(z) \right) \\
=&\sum_{j_1=0}^{\lambda_1} \cdots \sum_{j_p=0}^{\lambda_p} \delta_{j_1+\ldots+j_p,j} 
\prod_{k=1}^p e_{j_k}(J_{\mu_{k-1}+1}(z),\ldots,J_{\mu_k}(z))\\
=&\sum_{j_1=0}^{\lambda_1} \cdots \sum_{j_p=0}^{\lambda_p} \delta_{j_1+\ldots+j_p,j}  \prod_{k=1}^p X^{k,j_k}(z),
\end{align}
where by convention we set $X^{k,0}(z) = 1$ (and thus $X^{k,0}_{m_k} = \delta_{m_k,0}$). Then the explicit formula \eqref{eq:embedding} for the modes follows directly.
\end{proof}

Next we introduce a few simple definitions: 

\begin{definition}\label{d:lambdann}
Let $\lambda = (\lambda_1, \lambda_2, \ldots, \lambda_p)$ be an integer partition of $r$, and consider the embedding from \cref{l:embed}. For $d \in [p]$, we say that the mode $W^j_m$ is \emph{non-negative of level $d$ with respect to $\lambda$} if either $m \geq 0$, or for $m<0$, all terms in the sum over $m_1 + \ldots +m_p = m$ in \eqref{eq:embedding} satisfy one of the following two conditions:
\begin{enumerate}[(a)]
\item $m_k > 0$ for at least one $k \in [p]$;
\item there are at least $d$ distinct $k_1, \ldots, k_d \in [p]$ such that $m_{k_i} = 0$ and $j_{k_i} >0$ for all $i \in [d]$. 
\end{enumerate}
\end{definition}

To put it simply, a mode $W^j_m$ with $m<0$ is non-negative of level $d$ with respect to $\lambda$ if all monomials in the sum \eqref{eq:embedding} contain either one positive mode or at least $d$ non-trivial zero modes.

\begin{definition}\label{d:lambdaj}
Let $\lambda = (\lambda_1, \lambda_2, \ldots, \lambda_p)$ be an integer partition of $r$. For $j \in [r]$, we define
\begin{equation}
\lambda(j) =\min \{ s \in [p]~|~\lambda_1 + \ldots + \lambda_s \geq j \}.
\end{equation}
\end{definition}

The notions are related as follows:

\begin{lemma}\label{l:lambdann}
Let $\lambda = (\lambda_1, \lambda_2, \ldots, \lambda_p)$ be an integer partition of $r$. For $d \in [p]$, the mode $W^j_m$ is non-negative of level $d$ with respect to $\lambda$ if and only if $m \geq 0$ if $\lambda(j) \leq  d$ and $m \geq d - \lambda(j)$ if $\lambda(j) > d$.
\end{lemma}

\begin{proof}
\Cref{eq:embedding} expresses the mode $W^j_m$ as a sum of terms of the form $ \prod_{k=1}^p X^{k,j_k}_{m_k}$, with $\sum_{k=1}^p j_k = j$ and $\sum_{k=1}^p m_k = m$. Since by convention $X^{k,0}_{m_k} = \delta_{m_k,0}$, we define the degree of a term $\prod_{k=1}^p X^{k,j_k}_{m_k}$ to be the number of non-zero $j_k$'s (which counts the number of $X$'s in the product after replacing $X^{k,0}_{m_k} = \delta_{m_k,0}$). It is clear that for any given $j$, all terms that appear in the sum \eqref{eq:embedding} will have degree greater than or equal to $\lambda(j)$. 

As a result, given any $j$ and $m$, for $W^j_m$ to be non-negative of any level, we must have $m > - \lambda(j)$, otherwise the sum would contain a term with only negative modes.  For $W^j_m$ to be non-negative of level $d$, if $\lambda(j) > d$, we must have $m \geq d - \lambda(j)$, so that all terms contain either at least $d$ zero modes or at least one positive mode. For the cases with $\lambda(j) \leq d$, the only modes that are non-negative of level $d$ are those with $m \geq 0$, since whenever $m<0$ there will be terms with fewer than $d$ zero modes and no positive mode.
\end{proof}

We can rewrite the condition above in terms of a new partition of $r$.

\begin{lemma}\label{l:newpart}
Let $\lambda = (\lambda_1, \lambda_2, \ldots, \lambda_p)$ be an integer partition of $r$. For $d \in [p]$, define a new partition $\tilde{\lambda} =(\tilde{\lambda}_1, \ldots, \tilde{\lambda}_{p-d+1}) = (\mu_d, \lambda_{d+1} ,\ldots, \lambda_p)$, where $\mu_d = \sum_{i=1}^d \lambda_i$. The mode $W^j_m$ is non-negative of level $d$ with respect to $\lambda$ if and only if $m \geq 1 - \tilde{\lambda}(j)$.
\end{lemma}

\begin{proof}
First, we note that $\tilde{\lambda}(j) = 1$ for all $j \in [\mu_d]$, and $\lambda(j) \leq d$ if and only if $j \in [\mu_d]$. Therefore the condition $m \geq 0$ for all $j$ such that $\lambda(j) \leq d$ is reproduced. For $j > \mu_d$, we need to show that $1 - \tilde{\lambda}(j) = d - \lambda(j)$, that is, $\tilde{\lambda}(j) + d-1 = \lambda(j)$, which is clear by construction of the partition $\tilde{\lambda}$.
\end{proof}

\subsubsection{Step 1: constructing left ideals \texorpdfstring{$\mathcal{I}_{U_r}(\lambda) \subset U_r$}{IUr(l)inUr}}

We move on to the construction of the Airy ideals. The first step is to construct a family of proper left ideals $\mathcal{I}_{U_r} \subset U_r$ in the universal enveloping algebra of modes associated to partitions of $r$. The construction presented here is from Section 3.3 of \cite{BBCCN18}. We provide a proof of the main result so that we can generalize it in the next section.

\begin{proposition}\label{p:leftideals}
Let $\lambda = (\lambda_1, \lambda_2, \ldots, \lambda_p)$ be an integer partition of $r$. 
Let $\mathcal{I}_{U_r}(\lambda)$ be the left ideal generated by the modes $W^j_m$, with $j \in [r]$ and $m \geq1- \lambda(j)$. Then $U_r / \mathcal{I}_{U_r}(\lambda)$ is a cyclic left module generated by a non-zero vector $v$. Furthermore, 
$W^j_m \notin \mathcal{I}_{U_r}(\lambda)$ for all $j \in [r]$ and $m < 1- \lambda(j)$.
\end{proposition}

\begin{proof}
We start with the embedding $\mathcal{W}(\mathfrak{gl}_r) \subset \mathcal{W}(\mathfrak{gl}_{\lambda_1}) \times \ldots \times \mathcal{W}(\mathfrak{gl}_{\lambda_p})$ from \cref{l:embed}, with the explicit formula \eqref{eq:embedding}.

Let $v_k$ be a highest-weight vector with highest weight zero for $\mathcal{W}(\mathfrak{gl}_{\lambda_k} )$. That is, $X^{k,j}_{m} v_k = 0$ for all $j \in [\lambda_k]$ and $m \geq 0$, and the cyclic module generated by $v_k$ is spanned by elements of the form $X^{k,j_1}_{m_1 }\cdots X^{k,j_n}_{m_n} v_k$ with $m_1,\ldots,m_n < 0$.  Construct the tensor product $v = v_1 \otimes \cdots \otimes v_p$. $v$ generates a cyclic module for $U_r$ via the embedding \eqref{eq:embedding}. It is annihilated by all the modes $W^j_m$ that are non-negative of degree one with respect to $\lambda$ (see \cref{d:lambdann}), since all monomials in \eqref{eq:embedding} contain at least one non-negative mode. From \cref{l:lambdann}, we know that $W^j_m$ is non-negative of degree one with respect to $\lambda$ if and only if $m \geq 1 - \lambda(j)$ for all $j \in [r]$.

Furthermore, for each mode $W^j_m$ with $m <1- \lambda(j)$, there will be at least one monomial that will only involve negative modes. Since the $X^{k,j}_m$ form a PBW basis for the $\mathcal{W}(\mathfrak{gl}_{\lambda_k})$ factors, this means that these terms will act non-trivially on $v$, and therefore we conclude that the cyclic module generated by $v$ is spanned by elements of the form $W^{j_1}_{m_1} \cdots W^{j_n}_{m_n} v$ with $m_i < 1-\lambda(j_i)$. We conclude that the left ideal $\mathcal{I}_{U_r}(\lambda)$ generated by the $W^j_m$, with $j \in [r]$ and $m \geq 1-\lambda(j)$, is the annihilator ideal of $v$, and that $W^j_m \notin \mathcal{I}_{U_r}(\lambda)$ for all $j \in [r]$ and $m <1- \lambda(j)$.

\end{proof}

Let us clarify the statement of the proposition with a couple of examples.
\begin{example}
Pick $\lambda = (r)$. Then $\lambda(j) = 1$ for all $j \in [r]$. The ideal $\mathcal{I}_{U_r}(\lambda)$ is generated by all non-negative modes, that is $W^j_m$ with $j \in [r]$ and $m \geq 0$. The corresponding vector $v$ is a highest weight vector with weight zero, as it satisfies $W^j_m v  =0$ for all $j \in [r]$ and $m \geq0$.

Pick $\lambda=(1 ,1,\ldots, 1)$. Then $\lambda(j) = j$. The ideal $\mathcal{I}_{U_r}(\lambda)$ is generated by all modes $W^j_m$ with $j \in [r]$ and $m \geq 1-j$. The corresponding vector $v$ is the vacuum vector, which satisfies $W^j_m v = 0$ for all $j \in [r]$ and $m \geq 1-j$.
\end{example}

\subsubsection{Step 2: determining the homogenization of \texorpdfstring{$\mathcal{I}_{U_r}(\lambda)$}{IUr(l)}}

\label{s:Url}

Associated to a partition $\lambda$ of $r$ we constructed a left ideal $\mathcal{I}_{U_r}(\lambda)$ in the universal enveloping algebra of modes. The homogenization $\mathcal{I}_{U^\hslash_r}(\lambda) := h(\mathcal{I}_{U_r}(\lambda)) \subset U_r^\hslash$ of  $\mathcal{I}_{U_r}(\lambda)$ is obtained by homogenizing all elements of $\mathcal{I}_{U_r}(\lambda)$. By \cref{l:hom}, we know that
$$
[ \mathcal{I}_{U_r^\hslash}(\lambda), \mathcal{I}_{U_r^\hslash}(\lambda)] \subseteq \hslash^2 \mathcal{I}_{U_r^\hslash}(\lambda).
$$

For the  ideals that we constructed above, the homogenization is easy to describe. Since the modes $W^j_m$ form a PBW basis for $U_r$, and $W^j_m \in \mathcal{I}_{U_r}(\lambda)$ for $m \geq 1- \lambda(j)$ but $W^j_m \notin \mathcal{I}_{U_r}(\lambda)$ for $m < 1- \lambda(j)$, we conclude that the homogenization $\mathcal{I}_{U^\hslash_r}(\lambda) \subset U^\hslash_r$ is generated by the homogenization of the modes, that is, by $W^{\hslash,j}_m := \hslash^j W^j_m$ for $j \in [r]$ and $m \geq  1-\lambda(j)$. 

\subsubsection{Step 3: finding a good representation of \texorpdfstring{$U^\hslash_r$}{Uhr} in \texorpdfstring{$\mathcal{D}_{A}^\hslash$}{DAh}}

To a partition $\lambda$ of $r$ we constructed a left ideal $ \mathcal{I}_{U_r^\hslash}(\lambda) \subset U^\hslash_r$ that satisfies the condition $[ \mathcal{I}_{U_r^\hslash}(\lambda), \mathcal{I}_{U_r^\hslash}(\lambda)] \subseteq \hslash^2 \mathcal{I}_{U_r^\hslash}(\lambda)$. For each of those, can we find a representation  $\rho: U^\hslash_r \to \DAh$, for some index set $A$, such that there exists a  generating set $\{H_a\}_{a \in A}$ for $\mathcal{I}_{U^\hslash_r}$ with $\rho(H_a) = \hslash \partial_a + \hslash p_a(x_A) + O(\hslash^2)$ and the collection $\{\rho(H_a)\}_{a \in A}$ bounded?

One way to do that for a subset of those ideals is to consider representations of $U_r^\hslash$ that come from $\mathbb{Z}_r$-twisted representations for the Heisenberg VOA $\mathcal{H}(\mathfrak{gl}_r)$. This construction was proposed in \cite{BBCCN18}. We will not explain it in detail here, but simply state the final result, which is the following proposition:

\begin{proposition}[{\cite[Proposition~4.5 \& Corollary~4.7]{BBCCN18}}]
There exists a representation $\mu: U^\hslash_r \to \mathcal{D}_{\mathbb{N}^*}^\hslash$ that takes the form
  \begin{equation}\label{eq:ww}
   \mu( W^{\hslash,i}_k)
    =
    \left(\frac{\hslash}{r} \right)^i \sum_{j=0}^{\lfloor \frac{i}{2} \rfloor} \frac{i!}{2^j j! (i-2j)!} \sum_{\substack{p_{2j+1}, \dotsc p_i \in \Z \\ \sum p_l = rk}} \Psi^{(j)}_r (p_{2j+1}, \dotsc, p_i) \normord{ \prod_{l=2j+1}^i J_{p_l}} \,,
  \end{equation}
  where, with $ \theta = e^{2 \pi i /r}$,
  \begin{equation}
    \Psi^{(j)}_r (a_{2j+1}, \dotsc, a_i) \coloneqq \frac{1}{i!} \sum_{\substack{m_1, \dotsc, m_i = 0\\ m_l \neq m_k}}^{r-1} \prod_{k = 1}^j \frac{\theta^{m_{2k-1} + m_{2k}}}{(\theta^{m_{2k-1}} - \theta^{m_{2k}})^2} \prod_{l=2j+1}^i \theta^{-m_l a_l} \, ,
  \end{equation}
  and
  \begin{equation}\label{eq:Js}
  J_m  = \begin{cases} \del_{x_m} & m > 0 \\ 0 & m =0 \\ -m x_{-m} & m < 0\end{cases} \,.
\end{equation}
In \eqref{eq:ww}, for cases such that $j=i/2$ the condition $\sum p_l = rk$ is understood as the Kronecker delta $\delta_{k,0}$.
\end{proposition}

This is not yet in the form that we want though, since for $i \geq 2$ the $   \mu( W^{\hslash,i}_k)$ are $O(\hslash^i)$. However, for $ s \in [r+1]$ and $s$ coprime with $r$,
we can define a new representation $\rho: U^\hslash_r \to \mathcal{D}_{\mathbb{N}^*}^\hslash$ via conjugation:
\begin{equation}\label{eq:rep}
\rho( W^{\hslash,i}_k) = \hat{T}_s \mu(W^{\hslash,i}_k) \hat{T}_s^{-1}, \qquad \text{with} \qquad \hat{T}_s = \exp \Big( -\frac{J_s}{s\hslash} \Big).
\end{equation}
One can calculate that, if we keep only the modes such that $k  \geq - \lfloor \frac{s(i-1)}{r} \rfloor $, we get that
\begin{equation}
\rho( W^{\hslash,i}_k) = \hslash J_{r k + s (i-1)} + O(\hslash^2),
\end{equation}
which is in the right form.

Combining Steps 2 and 3, we need to find partitions $\lambda$ of $r$ such that $1-\lambda(i) = - \lfloor \frac{s(i-1)}{r} \rfloor$. As shown in Appendix B of \cite{BBCCN18}, the result is that there exists a partition $\lambda$ of $r$ such that $1-\lambda(i) = - \lfloor \frac{s(i-1)}{r} \rfloor$ if and only if $r = \pm 1 \pmod{s}$. For $s=1$, the partition is $\lambda=(r)$. For $2 \leq s \leq r-1$,  we can write $r = r' s + r''$ with $r'' \in \{1, s-1\}$, and the partition is given by $\lambda = (\lambda_1,\ldots, \lambda_s)$ with
\begin{equation}
\lambda_1 = \ldots = \lambda_{r''} = r'+1, \qquad \lambda_{r''+1} = \ldots = \lambda_s = r'.
\end{equation}
For $s=r+1$, the partition is  $\lambda = (1,1,\ldots,1)$.

We can summarize this in the following theorem:

\begin{theorem}[{\cite[Theorem~4.9]{BBCCN18}}]\label{t:rsAs}
  Let $ r \geq 2$, and $ s \in [r+1]$ such that $ r = \pm 1 \pmod{s}$. Let $\rho: U^\hslash_r \to \mathcal{D}_{\mathbb{N}^*}^\hslash$ be the representation defined in \eqref{eq:rep}. Let $\mathcal{I}_{U_r^\hslash} \subset U^\hslash_r$ be the left ideal generated by the  modes $W^{\hslash,j}_m$ with $j \in [r]$ and $m \geq - \lfloor \frac{s(i-1)}{r} \rfloor $, 
 and $\mathcal{I}$ the corresponding left ideal in $\mathcal{D}_{\mathbb{N}^*}^\hslash$ generated by $\rho(\mathcal{I}_{U_r^\hslash})$. Then $\mathcal{I}$ is an Airy ideal, which we call the \emph{$(r,s)$-Airy ideal}.
\end{theorem}{}

Since $\mathcal{I}$ is an Airy ideal, there exists a unique partition function $Z$ such that $\mathcal{I} Z = 0$. Concretely, what this means is that
\begin{equation}
\rho(W_m^{\hslash,i}) Z = 0 \qquad \text{for $i \in [r], m \geq - \lfloor \frac{s(i-1)}{r} \rfloor$.}
\end{equation}
This set of differential constraints can be used to uniquely reconstruct $Z$ recursively. This is equivalent to topological recursion on the $(r,s)$ spectral curves $x^{r-s} y^r - 1 = 0$, as shown in \cite{BBCCN18}.

\subsubsection{More general representations}

We can generalize the construction of the Airy ideals in \cref{t:rsAs} by constructing more general representations $\rho: U^\hslash_r \to \mathcal{D}_{\mathbb{N}^*}^\hslash$. The idea is simple: instead of conjugating by $T_s$ as in \eqref{eq:rep}, we conjugate by more complicated operators. This idea was explored in \cite{BBCCN18} -- see also \cite[Section~4.1]{BKS23}.

Pick a collection of complex numbers
\begin{equation}\label{eq:complex}
  F_{0,1}[-k]\,, k \geq \min \{ s, r\} \,, \qquad F_{\frac{1}{2},1}[-k]\,, k > 0\,, \qquad F_{0,2}[-k,-l] \,, k,l > 0,
\end{equation}

such that $ F_{0,1}[-s] \neq 0$ and $ F_{0,2}[-k,-l] = F_{0,2}[-l,-k]$. Define the operators
\begin{align}\label{eq:TPhi}
  \hat{T} 
  &\coloneqq
  \exp \Big( \sum_k \left( \frac{1}{\hslash} F_{0,1}[-k] + F_{\frac{1}{2},1}[-k] \right) \frac{J_k}{k}\Big) \,,
  \\
  \hat{\Phi}
  &\coloneqq
  \exp \Big( \frac{1}{2} \sum_{k,l > 0} F_{0,2}[-k,-l]\frac{J_kJ_l}{kl} \Big)\,.
\end{align}
We define a new representation $\rho': U^\hslash_r \to \mathcal{D}_{\mathbb{N}^*}^\hslash$ via conjugation:
\begin{equation}\label{eq:repp}
\rho'( W^{\hslash,i}_k) = \hat{\Phi}\hat{T} \mu(W^{\hslash,i}_k) \hat{T}^{-1} \hat{\Phi}^{-1}.
\end{equation}
Then it is not too difficult to show that \cref{t:rsAs} generalizes to this new class of representations:

\begin{proposition}[{\cite[Proposition~4.14]{BBCCN18} \& \cite[Theorem~2.14]{BKS23}}] \label{p:rsAsgen}
Let $ r \geq 2$, and $ s \in [r+1]$ such that $ r = \pm 1 \pmod{s}$. Let $\rho': U^\hslash_r \to \mathcal{D}_{\mathbb{N}^*}^\hslash$ be the representation defined in \eqref{eq:repp}. Let $\mathcal{I}_{U_r^\hslash} \subset U^\hslash_r$ be the left ideal generated by the  modes $W^{\hslash,j}_m$ with $j \in [r]$ and $m \geq - \lfloor \frac{s(i-1)}{r} \rfloor $, 
 and $\mathcal{I}$ the corresponding left ideal in $\mathcal{D}_{\mathbb{N}^*}^\hslash$ generated by $\rho'(\mathcal{I}_{U_r^\hslash})$. Then $\mathcal{I}$ is an Airy ideal, which we call the \emph{deformed $(r,s)$-Airy ideal}.
\end{proposition}
%

\subsection{Shifted \texorpdfstring{$(r,s)$}{rs}-Airy ideals}
\label{s:shiftedrs}

The construction of the previous section can be naturally generalized by starting with highest weight vectors with non-zero weights. This gives rise to new left ideals that can be used to construct Airy ideals. We continue using the three-step approach.

\subsubsection{Step 1: constructing left ideals \texorpdfstring{$\mathcal{I}_{U_r}(\lambda; S)$}{IUr(l;S)}}

We start by generalizing the previous construction, but now starting with highest weight vectors.

\begin{lemma}\label{l:leftidealsshifted}
Let $\lambda = (\lambda_1, \lambda_2, \ldots, \lambda_p)$ be an integer partition of $r$. For $d \in [p]$, define a new partition $\tilde{\lambda} =(\tilde{\lambda}_1, \ldots, \tilde{\lambda}_{p-d+1}) = (\mu_d, \lambda_{d+1} ,\ldots, \lambda_p)$, where $\mu_k = \sum_{i=1}^k \lambda_i$.  Let $S_j \in \mathbb{C}$ for $j \in [\mu_{d-1}]$ and $S_j =0$ for $j > \mu_{d-1}$.

Let $\mathcal{I}_{U_r}(\tilde\lambda;S)$ be the left ideal generated by the shifted modes $W^j_m - S_j \delta_{m,0}$, with $j \in [r]$ and $m \geq 1 - \tilde{\lambda}(j)$. Then $U_r / \mathcal{I}_{U_r}(\tilde \lambda; S)$ is a cyclic left module generated by a non-zero vector $v$. Furthermore, $W^j_m v \neq0$ for all $j \in [r]$ and $m < 1 - \tilde{\lambda}(j)$, and thus
$W^j_m \notin \mathcal{I}_{U_r}(\tilde \lambda; S)$ for all $j \in [r]$ and $m < 1 - \tilde{\lambda}(j)$.
\end{lemma}

\begin{proof}
The proof goes along the same lines as for \cref{p:leftideals}. The main difference is that we consider highest weight vectors with non-zero weights. We use again the embedding of $\mathcal{W}(\mathfrak{gl}_r)$ in $\mathcal{W}(\mathfrak{gl}_{\lambda_1}) \times \ldots \times \mathcal{W}(\mathfrak{gl}_{\lambda_p})$, with the explicit formula from \eqref{eq:embedding}.

Let $v_1$ be a highest weight vector with highest weight $Q^1 = (Q^1_1, \ldots, Q^1_{\lambda_1})$ for $\mathcal{W}(\mathfrak{gl}_{\lambda_1} )$. That is, $X^{1,j}_{m} v_1 = \delta_{m,0} Q^1_j v_1$ for all $j \in [\lambda_1]$ and $m \geq 0$, and the cyclic module generated by $v_1$ is spanned by elements of the form $X^{1,j_1}_{m_1 }\cdots X^{1,j_n}_{m_n} v_1$ with $m_1,\ldots,m_n < 0$. 

Let $v_2, \ldots, v_p$ be highest weight vector with weight zero as in the proof of \cref{p:leftideals}. Construct the tensor product $v = v_1 \otimes \cdots \otimes v_p$. We want to find the annihilator of $v$ in $\mathcal{W}(\mathfrak{gl}_r)$. It is clear that $v$ is annihilated by all modes $W^j_m$ that are non-negative of degree two with respect to $\lambda$ (see \cref{d:lambdann}) and such that $j > \lambda_1$, since all monomials in those modes will contain either a positive mode or two zero modes from different factors (and only the zero modes from the first factor $\mathcal{W}(\mathfrak{gl}_{\lambda_1} )$ act non-trivially). Furthermore, it is clear that the zero modes $W^j_0$ for $j \in [\lambda_1]$ act as $W^j_0 v = Q^1_j v$. Let us set $S_j = Q^1_j$ for $j \in [\lambda_1]$ and $S_j=0$ for $j > \lambda_1$. Using \cref{l:newpart}, we conclude that the annihilator is the left ideal generated by the modes $W^j_m - \delta_{m,0} S_j$, with $m \geq 1 - \tilde{\lambda}(j)$, for the new partition $\tilde{\lambda} = (\lambda_1+\lambda_2, \lambda_3, \ldots, \lambda_p)$ of $r$. This is the case $d=2$.

If we started instead with a highest weight vector $v_k$ for any other factor $2 \leq k \leq p$, we would reach the conclusion (we leave this as an exercise) that the annihilator is the left ideal generated by the modes $W^j_m - \delta_{m,0} S_j$ with $m \geq 1 - \tilde{\lambda}(j)$, for the same partition $\tilde{\lambda}$ as above, but this time with only the weights $S_j$ for $j \in [\lambda_k] \subseteq [\lambda_1]$ being non-zero. Thus we can see it as a subcase of the previous one. This concludes the case $d=2$.

For general $d \in [p]$, consider highest weight vectors $v_1,\ldots,v_{d-1}$ with non-zero weights for $\mathcal{W}(\mathfrak{gl}_{\lambda_j})$ with $j \in [d-1]$, and highest weight vectors $v_d, \ldots, v_p$ with zero weights. The tensor product $v = v_1 \otimes \cdots \otimes v_p$ is annihilated by all modes $W^j_m$ that are non-negative of degree $d$ with respect to $\lambda$ and such that $j > \lambda_1 + \ldots + \lambda_{d-1} = \mu_{d-1}$, since all monomials in these modes contain either a positive mode or $d$ zero modes from distinct factors. Further, the zero modes $W^j_0$ for $j \in [\mu_{d-1}]$ will act as $W^j_0 v = S_j v$ for some constants $S_j$ that are obtained as polynomials in the highest weights of the vector $v_1,\ldots,v_{d-1}$. Set $S_j = 0$ for $j > \mu_{d-1}$; we conclude that the annihilator is the left ideal generated by the modes  $W^j_m - \delta_{m,0} S_j$, with $m \geq 1 - \tilde{\lambda}(j)$, for the new partition $\tilde{\lambda} = (\mu_d, \lambda_{d+1}, \ldots, \lambda_p)$ of $r$.  As in the $d=2$ case, considering the tensor product of $d-1$ other highest weight vectors is a sub-case of this one.
\end{proof}

Looking at the statement of the lemma, we notice that in the end, the ideal $\mathcal{I}_{U_r}(\tilde\lambda;S)$ is almost entirely defined in terms of the new partition $\tilde{\lambda}$; the original partition only appears in the choice of non-zero weights. So we can think of the partition $\tilde{\lambda}$ as our starting point. After renaming $\tilde{\lambda}$ as $\lambda$ for simplicity, this leads to the following simpler reformulation of the result.

\begin{theorem}\label{t:leftidealsshifted}
Let $\lambda = (\lambda_1, \lambda_2, \ldots, \lambda_p)$ be an integer partition of $r$. Let $S_j \in \mathbb{C}$ for $j \in [\lambda_1-\lambda_2]$ and $S_j =0$ for $j > \lambda_1-\lambda_2$. Let $\mathcal{I}_{U_r}(\lambda; S)$ be the left ideal generated by the shifted modes $W^j_m - S_j \delta_{m,0}$, with $j \in [r]$ and $m \geq 1 - \lambda(j)$. Then $U_r / \mathcal{I}_{U_r}(\lambda; S)$ is a cyclic left module generated by a non-zero vector $v$. Furthermore, $W^j_m v \neq0$ for all $j \in [r]$ and $m < 1 - \lambda(j)$, and thus
$W^j_m \notin \mathcal{I}_{U_r}( \lambda)$ for all $j \in [r]$ and $m < 1 - \lambda(j)$.
\end{theorem}

\begin{proof}
Consider the statement of \cref{l:leftidealsshifted}, but think of the partition $\tilde{\lambda} = (\tilde{\lambda}_1, \ldots, \tilde{\lambda}_q)$ as the starting point. The ideal $\mathcal{I}_{U_r}(\tilde{\lambda}; S)$ is the left ideal generated by the shifted modes $W^j_m - S_j \delta_{m,0}$, with $j \in [r]$ and $m \geq 1 - \tilde{\lambda}(j)$. But which of the shifts are non-zero? If we think of $\tilde{\lambda}$ as the starting point, then it can be realized as $\tilde{\lambda} = (\mu_d, \lambda_{d+1}, \ldots, \lambda_p)$ for a number of distinct partitions $\lambda$ of $r$. For each such choice, the non-zero shifts are $S_j$ for $j \in [\mu_{d-1}]$. In the end, those are all subcases of the choice for which $\mu_{d-1}$ is largest. But $\mu_{d-1} = \mu_d - \lambda_d$ is largest when $\lambda_d$ is smallest. Since $\lambda_d \geq \lambda_{d+1}$, the maximal value of $\mu_{d-1}$ is $\mu_d - \lambda_{d+1} = \tilde{\lambda}_1 - \tilde{\lambda}_2$. Therefore, the maximal number of non-zero shifts that we can allow is $\tilde{\lambda}_1 - \tilde{\lambda}_2$; that is, the shifts $S_j$ for $j \in [\tilde{\lambda}_1 - \tilde{\lambda}_2]$ can be non-zero.

Replacing $\tilde{\lambda}$ by $\lambda$ in the final result we get the statement of the theorem.
\end{proof}

\begin{remark}
In essence, what this means is that, given any partition $\lambda$ of $r$ and left ideal generated by the modes $W^j_m$ with $j \in [r]$ and $m \geq 1 - \lambda(j)$, we can shift the zero modes $W^j_0$ for $j \in [\lambda_1-\lambda_2]$, and the shifted modes generate a new left ideal such that all $W^j_m$ with $m < 1 - \lambda(j)$ are not in the ideal. In the language of \cite{BBCC21}, one can say that the modes $W^j_0$ for $ j \in [\lambda_1-\lambda_2]$ are \emph{extraneous}. 
\end{remark}

\subsubsection{Step 2: determining the homogenization of \texorpdfstring{$\mathcal{I}_{U_r}(\lambda;S)$}{IUr(l;S)}}

By \cref{l:hom} the homogenization $\mathcal{I}_{U^\hslash_r}(\lambda;S):= h(\mathcal{I}_{U_r}(\lambda;S)) \subset U^\hslash_r$ satisfies 
\begin{equation}
[ \mathcal{I}_{U_r^\hslash}(\lambda;S), \mathcal{I}_{U_r^\hslash}(\lambda;S)] \subseteq \hslash^2 \mathcal{I}_{U_r^\hslash}(\lambda;S).
\end{equation}
Just as for the $(r,s)$-Airy ideals, the homogenization is easy to describe. By the same argument as in \cref{s:Url}, the homogenization  $\mathcal{I}_{U^\hslash_r}(\lambda;S)$ is generated by the homogenization of the shifted modes, that is, by $W^{\hslash,j}_m(S) := \hslash^j (W^j_m- \delta_{m,0} S_j)$ for $j \in [r]$ and $m \geq  1-\lambda(j)$.

\subsubsection{Step 3: finding a good representation of \texorpdfstring{$U^\hslash_r$}{Uhr} in \texorpdfstring{$\mathcal{D}_{A}^\hslash$}{DAh}}

For the modes $W^{\hslash,j}_m(S)$ with $m >0$, we use the same representation $\mu: U^\hslash_r \to \mathcal{D}_{\mathbb{N}^*}^\hslash$ as before from \eqref{eq:ww}. We extend it to the shifted modes by{}
\begin{equation}
\mu(W^{\hslash,i}_0(S) ) = \mu(W^{\hslash,i}_0) - \sum_{n=1}^\infty \hslash^n S_{i,n},
\end{equation}
where the $S_{i,n} \in \mathbb{C}$ for{} $i \in [\lambda_1-\lambda_2]$ and $S_{i,n} = 0$ for $i > \lambda_1-\lambda_2$. It is easy to see that mapping the shifts $\hslash^i S_i$ to the series $ \sum_{n=1}^\infty \hslash^n S_{i,n}$ still produces a representation of the universal enveloping algebra.

As before, for $ s \in [r+1]$ and $s$ coprime with $r$,
we define a new representation $\rho: U^\hslash_r \to \mathcal{D}_{\mathbb{N}^*}^\hslash$ via conjugation:
\begin{equation}\label{eq:rep2}
\rho( W^{\hslash,i}_k(S)) = T_s \mu(W^{\hslash,i}_k(S)) T_s^{-1}, \qquad \text{with} \qquad T_s = \exp \Big( -\frac{J_s}{s\hslash} \Big).
\end{equation}
If we keep only the modes such that $k  \geq - \lfloor \frac{s(i-1)}{r} \rfloor $, we get that
\begin{equation}
\rho( W^{\hslash,i}_k) = \hslash J_{r k + s (i-1)} + \hslash S_{i,1} + O(\hslash^2),
\end{equation}
which is in the right form.

As before, a partition $\lambda$ of $r$ such that $1-\lambda(i) = - \lfloor \frac{s(i-1)}{r} \rfloor$ exists if and only if $r = \pm 1 \pmod{s}$. For $s=1$, the partition is $\lambda=(r)$. For $2 \leq s \leq r-1$,  we can write $r = r' s + r''$ with $r'' \in \{1, s-1\}$, and the partition is given by $\lambda = (\lambda_1,\ldots, \lambda_s)$ with
\begin{equation}
\lambda_1 = \ldots = \lambda_{r''} = r'+1, \qquad \lambda_{r''+1} = \ldots = \lambda_s = r'.
\end{equation}
For $s=r+1$, the partition is  $\lambda = (1,1,\ldots,1)$.

We notice that, for $s \geq 2$, if $r= 1 \pmod{s}$, $\lambda_1 = \lambda_2+1$, which means that the only non-zero shifts are $S_{1,n}$; that is, we can only shift the zero mode $W^1_0$. For $s \geq 3$, if $r=-1 \pmod{s}$, $\lambda_1 = \lambda_2$, and all shifts are zero; we are back to the $(r,s)$-Airy ideals. 

In the case $s=1$, things are more interesting. The partition is $\lambda = (r)$. We are then allowed to shift all zero modes, that is, $S_{i,n} \neq 0$ for all $i \in [r]$ and $n \geq 1$. 

To summarize these conditions, we define the notion of a set of $s$-consistent shifts:
\begin{definition}\label{d:consistent}
Let $S=\{ S_{i,n} \}_{i \in [r], n \in \mathbb{N}^*}$ be a set of complex numbers. We say that it is \emph{$s$-consistent} if the following two conditions are satisfied:
\begin{itemize}
\item If $s \geq 2$ and $r = 1 \pmod{s}$, then $S_{i,n} = 0$ for all $2 \leq i \leq r$, and:
\item If $s \geq 3$ and $r = -1 \pmod{s}$, then $S_{i,n} = 0$ for all $i \in [r]$.
\end{itemize}
\end{definition}

We then obtain the following theorem:

\begin{theorem}\label{t:shifts}
  Let $ r \geq 2$, and $ s \in [r+1]$ such that $ r = \pm 1 \pmod{s}$. Let $\rho: U^\hslash_r \to \mathcal{D}_{\mathbb{N}^*}^\hslash$ be the representation defined in \eqref{eq:rep2}. Let $\mathcal{I}_{U_r^\hslash}(S) \in U^\hslash_r$ be the left ideal generated by the shifted modes $W^{\hslash,j}_m(S)$ with $j \in [r]$ and $m \geq - \lfloor \frac{s(i-1)}{r} \rfloor $, where the set of shifts $S$ is $s$-consistent,
 and $\mathcal{I}(S)$ the corresponding left ideal in $\mathcal{D}_{\mathbb{N}^*}^\hslash$ generated by $\rho(\mathcal{I}_{U_r^\hslash}(S))$. Then $\mathcal{I}(S)$ is an Airy ideal, which we call the \emph{shifted $(r,s)$-Airy ideal.}
 
 For $s=1$, all zero modes are shifted, that is,
 \begin{equation}
 \rho(W^{\hslash,j}_0(S)) = \rho(W^{\hslash,j}_0) - \sum_{n=1}^\infty \hslash^n S_{j,n}.
 \end{equation}
 For $s \geq 2$ and $r = 1 \pmod{s}$, only the first zero mode is shifted, that is,
 \begin{equation}
  \rho(W^{\hslash,j}_0(S)) = \rho(W^{\hslash,j}_0) - \delta_{j,1} \sum_{n=1}^\infty \hslash^n S_{1,n}.
 \end{equation}
 For $s \geq 3$ and $r = -1 \pmod{s}$, no shifts are allowed.
\end{theorem}

The $s=1$ case is particularly interesting. Since $\mathcal{I}(S)$ is an Airy ideal, there exists a unique partition function $Z$ such that $\mathcal{I}(S) Z = 0$. Explicitly, this means that
\begin{equation}
 \rho(W^{\hslash,j}_m(S)) Z = 0 \qquad \text{for $j\in[r], m \geq 0$.} 
\end{equation}
In other words, this means that
\begin{equation}
 \rho(W^{\hslash,j}_m) Z = \left(\sum_{n=1}^\infty \hslash^n S_{j,n} \right) Z \qquad \text{for $j\in[r], m \geq 0$.} 
\end{equation}
Thus we can think of the partition function $Z$ for the shifted $(r,s)$-Airy ideals as being a highest weight vector for $\mathcal{W}(\mathfrak{gl}_r)$ at self-dual level.

\begin{remark}\label{r:extra}
In the language of \cite{BBCC21}, the statement of \cref{t:shifts} is that for the $(r,s)$-Airy ideals, there are extraneous zero modes only for the cases $s=1$ or $r = 1 \pmod{s}$. For $s=1$, all zero modes $W^j_0$, $j \in [r]$, are extraneous, while for $r=1 \pmod{s}$ only the zero mode $W^1_0$ is extraneous.
\end{remark}

\subsubsection{More general representations}

Just as for the $(r,s)$-Airy ideals, we can construct more general shifted $(r,s)$-Airy ideals via conjugation. As before, we construct a new class of representations
$\rho': U^\hslash_r \to \mathcal{D}_{\mathbb{N}^*}^\hslash$ via conjugation:
\begin{equation}\label{eq:rep2p}
\rho'( W^{\hslash,i}_k(S)) = \hat{\Phi}\hat{T} \mu(W^{\hslash,i}_k(S)) \hat{T}^{-1} \hat{\Phi}^{-1}.
\end{equation}
with $\hat{T}$ and $\hat{\Phi}$ defined in \eqref{eq:TPhi}.

Following the same arguments as in \cite{BBCCN18} and \cite{BKS23}, \cref{t:shifts} generalizes to this new class of representations:

\begin{proposition}\label{p:shiftsgen}
 Let $ r \geq 2$, and $ s \in [r+1]$ such that $ r = \pm 1 \pmod{s}$. Let $\rho': U^\hslash_r \to \mathcal{D}_{\mathbb{N}^*}^\hslash$ be the representation defined in \eqref{eq:rep2p}. Let $\mathcal{I}_{U_r^\hslash}(S) \in U^\hslash_r$ be the left ideal generated by the shifted modes $W^{\hslash,j}_m(S)$ with $j \in [r]$ and $m \geq - \lfloor \frac{s(i-1)}{r} \rfloor $, where the set of shifts $S$ is $s$-consistent,
 and $\mathcal{I}(S)$ the corresponding left ideal in $\mathcal{D}_{\mathbb{N}^*}^\hslash$ generated by $\rho'(\mathcal{I}_{U_r^\hslash}(S))$. Then $\mathcal{I}(S)$ is an Airy ideal, which we call the \emph{deformed and shifted $(r,s)$-Airy ideal}.
\end{proposition}

\subsection{Other shifts}

In the previous section we showed that we can shift some zero modes to get new shifted $(r,s)$-Airy ideals. But are we allowed to shift other modes that are not zero modes? The answer is no, because of the following simple lemma.

\begin{lemma}\label{l:othershifts}
Let $\lambda = (\lambda_1, \lambda_2, \ldots, \lambda_p)$ be an integer partition of $r$.  Fix a pair $(\alpha,\beta)$, with $\alpha \in [r]$ and $0 \neq \beta \in \mathbb{Z}$. Let $\mathcal{I}$ be the left ideal generated by the modes $W^j_m - S  \delta_{j,\alpha} \delta_{m,\beta}$ with $j \in [r]$ and $m \geq 1 - \lambda(j)$, where $0 \neq S \in \mathbb{C}$. In other words, we shift only one mode, but it is a non-zero mode. Then $\mathcal{I} \simeq U_r$. That is, the left ideal is not proper.
\end{lemma}

\begin{proof}
The Virasoro zero-mode $W^2_0$ is always in the ideal $\mathcal{I}$. Thus, for any mode $W^j_m$, we have
\begin{equation}
[W^2_0, W^j_m] = m W^j_m.
\end{equation}
This means that if we shift the mode $W^\alpha_\beta$, we get
\begin{equation}
[ W^2_0, W^\alpha_\beta - S] = [W^2_0, W^\alpha_\beta] = \beta W^{\alpha}_\beta = \beta(W^\alpha_\beta - S) + \beta S.
\end{equation}
The left-hand-side is clearly in the ideal $\mathcal{I}$, and thus the right-hand-side must be too. Since the first term on the right-hand-side is in the ideal, we conclude that $\beta S \in\mathcal{I}$. But $\beta S \in \mathbb{C}$, and we conclude that $\mathcal{I}\simeq U_r$.
\end{proof}

The upshot of this simple lemma is that the homogenization of $\mathcal{I}$ is the whole Rees universal enveloping algebra $U_r^\hslash$. It is thus impossible to find a representation that maps its generators to operators of the required form in a Rees Weyl algebra, and we conclude that we cannot obtain Airy ideals in this way.

\begin{remark}
In the language of \cite{BBCC21}, \cref{l:othershifts} can be reformulated as the statement that for the $(r,s)$-Airy ideals, only zero modes can be extraneous. As we already classified in \cref{t:shifts} what zero modes are extraneous, this concludes the analysis of extraneous modes for the $(r,s)$-Airy ideals.
\end{remark}

\section{Shifted loop equations and shifted topological recursion}

\label{s:shiftedle}

In section \ref{s:shifted} we constructed new Airy ideals, which we called ``shifted $(r,s)$-Airy ideals''. In the case $s=1$, the partition function associated to these shifted $(r,1)$-Airy ideals is a highest weight vector for the $\mathcal{W}(\mathfrak{gl}_r)$-algebra at self-dual level.

In general, as shown in \cite{BBCCN18}, the differential constraints associated to the $(r,s)$-Airy ideals can be reformulated as loop equations for a system of correlators on the $(r,s)$-spectral curves. Along similar lines, in this section we show that the differential constraints associated to the shifted $(r,s)$-Airy ideals can be recast as ``shifted loop equations'' for another system of correlators on shifted $(r,s)$-spectral curves. We then find a recursive formula that solves these shifted loop equations; it turns out to look like the usual topological recursion formula, but with some correlators appropriately shifted. Unsurprisingly (or perhaps uncreatively) we call this recursive formula ``shifted topological recursion''.
%

\subsection{Spectral curves, loop equations and topological recursion}

We refer the reader to \cite{BBCKS} for a careful treatment of spectral curves, loop equations and topological recursion. An introduction to these concepts is also available in \cite{Bo24}. Here we summarize the main concepts.

We start with the general definition of spectral curves.

\begin{definition}\label{d:sc}
  An \emph{admissible local spectral curve} $ \mc{S} = ( C, x, \omega_{0,1}, \omega_{\frac{1}{2},1}, \omega_{0,2}) $ is a collection of small disks $ C = \bigsqcup_{j=1}^N C_j $ for some positive integer $N$ together with maps $ x \colon C_j \to \P^1 \colon z \mapsto z^{r_j} + x_j$ for distinct $ x_j \in \P^1$,  two one-forms $\omega_{0,1}$ and $ \omega_{\frac{1}{2},1} $ which on each $C_j$ have expansions
  \begin{align}
    \omega_{0,1}^j (z) 
    &= 
    \sum_{k \geq s_j} F^j_{0,1}[-k] z^{k-1} dz \,,
    \\
    \omega_{\frac{1}{2},1}^j (z) 
    &= 
    \sum_{k \geq 0} F^j_{\frac{1}{2},1}[-k] z^{k-1} dz \,,
  \end{align}
  where $ F^j_{0,1} [-s_j] \neq 0$ and $s_j \in [r_j+1] $ such that $ r_j = \pm 1 \pmod{s_j}$, and a \emph{fundamental bidifferential of the second kind}
  \begin{equation}
    \omega_{0,2} \in H^0 ( C^2 ; K_C^{\boxtimes 2}(2 \Delta))^{\mf{S}_2}
  \end{equation}
  with biresidue $ 1$ on the diagonal.
  \end{definition}
  
  Given a spectral curve, we construct a particular basis of one-forms that will play an important role in the following.
  
  \begin{definition}\label{d:xibasis}
  Let $\mc{S}$ be an admissible local spectral curve. For each component $C_j$ with $j \in [N]$, we define a basis of one-forms:
  \begin{align}
  \xi^{(j)}_{k} (z)
  &\coloneqq
  z^{k-1} dz \,,
  \\
  \xi^{(j)}_{-k} (z) 
  &\coloneqq \Res_{w = 0} \Big( \int_0^{w} \omega_{0,2}( \mathord{\cdot}, z) \Big) \frac{dw}{w^{k+1}} = \left( \frac{1}{z^{k+1}} +\text{holomorphic} \right)\ dz \,.
  \end{align}
  \end{definition}

We also introduce the notation:
\begin{definition}\label{d:fz}
Let $\mc{S}$ be an admissible local spectral curve. For each component $C_j$ with $j \in [N]$, we define $\mathfrak{f}(z) = \{ \theta^k z \}_{k\in [r_j]}$, where $\theta = \exp\left(\frac{2 \pi i}{r_i} \right)$. $\mathfrak{f}(z)$ is the set of sheets of the branched covering $x: C_j \to \mathbb{P}^1$ near the ramification point $z=0$.
\end{definition}
  
  The main object of study is a system of correlators. 
  
  \begin{definition}\label{d:system}
  A \emph{system of correlators} on an admissible local spectral curve $ \mc{S}$ is a collection $ \{ \omega_{g,n} \}_{g \in \frac{1}{2} \N, n\in \N^*}$ such that $ \omega_{0,1}$, $ \omega_{\frac{1}{2},1}$, and $ \omega_{0,2}$ are the ones given as part of the data of the spectral curve, and all $ \omega_{g,n}$ for $ 2g-2+n > 0$ are symmetric meromorphic $n$-differentials on $ C^n$, with only possible poles at the origins of the $ C_j$ with vanishing residue.
  \end{definition}

  We will single out particular systems of correlators that satisfy the projection property.
  
  \begin{definition}\label{d:projection} 
  Let $ \{ \omega_{g,n} \}_{g \in \frac{1}{2} \N, n\in \N^*}$ be a system of correlators on an admissible local spectral curve $\mc{S}$. We say that it satisfies the \emph{projection property} if for all $ 2g - 2 + n > 0$,
  \begin{equation}
    \omega_{g,n} (z_{[n]}) = \sum_{j} \Res_{z = 0 \in C_j} \Big( \int_0^z \omega_{0,2} (\mathord{\cdot},z_1) \Big) \omega_{g,n} (z, z_2, \dotsc, z_n)
  \end{equation}
\end{definition}

  It is easy to see that the basis of one-forms introduced in \cref{d:xibasis} is well suited to study systems of correlators that satisfy the projection property:
  \begin{lemma}\label{l:finiteness}
  Let $\{ \omega_{g,n} \}_{g \in \frac{1}{2} \N, n\in \N^*}$ be a system of correlators on an admissible local spectral curve $\mc{S}$. The system of correlators satisfies the projection property if and only if it has an expansion of the form
 \begin{equation}\label{eq:expansion}
\omega_{g,n}(z_1,\ldots,z_n) = \sum_{j_1, \ldots, j_n \in [N]} \sum_{k_1, \ldots, k_n \in \mathbb{N}^*} F_{g,n} \begin{bmatrix} j_1 & \ldots & j_n \\ k_1 & \ldots & k_n \end{bmatrix} \xi_{-k_1}^{(j_1)}(z_1) \cdots  \xi_{-k_n}^{(j_n)}(z_n),
\end{equation}
where only a finite number of coefficients are  non-zero. Note that only the one-forms $\xi^{(j)}_k(z)$ with negative $k$ appear in the summation.
  \end{lemma}

  For the purpose of formulating loop equations and topological recursion, we introduce the following particular combinations of correlators.
  \begin{definition}\label{d:EW}
  Let $ \{ \omega_{g,n} \}_{g \in \frac{1}{2} \N, n\in \N^*}$ be a system of correlators on an admissible local spectral curve $\mc{S}$.  For any $i \in \mathbb{N}^*$, we define the objects:
  \begin{align}\label{partiallydisconnected}
     \mc{W}_{g,i,n} (z_{[i]} ; w_{[n]}) \coloneqq& \sum_{\substack{P \vdash z_{[i]} \\ \bigsqcup_{S \in P} N_S = w_{[n]}\\ \sum_{S \in P} (g_S -1) = g-i}}\prod_{S \in P} \omega_{g_S, |S| + |N_S|}(S, N_S) \,, \\
   \mc{W}_{g,i,n}' (z_{[i]} ; w_{[n]}) \coloneqq& \sum'_{\substack{P \vdash z_{[i]} \\ \bigsqcup_{S \in P} N_S = w_{[n]}\\ \sum_{S \in P} (g_S -1) = g-i}}\prod_{S \in P} \omega_{g_S, |S| + |N_S|}(S, N_S) \,,
  \end{align}
  where the sum is 1) over set partitions $P$ of $z_{[i]}$, 2) over all possible splittings of $w_{[n]}$ into possibly empty disjoint subsets $N_S$ where $S$ runs over all parts of $P$ and $\bigsqcup_{S \in P} N_S = w_{[n]}$, 3) over all sets of non-negative half-integers $\{ g_S \}_{S \in P}$ such that $\sum_{S \in P} (g_S - 1) = g-i$. The difference between the first and second object is that the prime over the summation symbol means that the terms with $\omega_{0,1}$ are omitted from the sum.
  
  For each component $C_j$, with $j \in [N]$, and for $i \in [ r_j]$, we also define the objects
    \begin{equation}
    \mc{E}^{i, (j)}_{g,n}(x; z_{[n]}) = \sum_{\substack{Z \subseteq \mf{f}(z) \\ |Z| = i}} \mc{W}_{g,i,n}(Z; z_{[n]}).
  \end{equation}
  \end{definition}

We can now define so-called loop equations, which are particular equations satisfied by systems of correlators.

\begin{definition}
  Let $ \{ \omega_{g,n} \}_{g \in \frac{1}{2} \N, n\in \N^*}$ be a system of correlators on an admissible local spectral curve $\mc{S}$. 
  We say that the system of correlators satisfies the \emph{loop equations} if, for all $j \in [N]$, $i \in [r_j]$, and $2g-2+n > 0$,
  \begin{equation}
    \mc{E}^{i,(j)}_{g,n} (x; z_{[n]}) \in  \mc{O} \left( x^{ \lfloor \frac{s_j(i-1)}{r_j}\rfloor + 1} \right) \left( \frac{dx}{x}\right)^i \,.
  \end{equation}
\end{definition}

The main result of relevance here is that, given an admissible local spectral curve $\mc{S}$, there always exists a single system of correlators that satisfies both the loop equations and the projection property, and this system of correlators can be reconstructed recursively from the data of the spectral curve.

\begin{theorem}[{\cite[Theorem 2.17]{BBCKS}}]\label{UnshiftedTR}
  For an admissible local spectral curve $ \mc{S}$, there exists exactly one system of correlators that satisfies the loop equations and the projection property. It can be calculated recursively by the \emph{topological recursion formula}
  \begin{equation}
    \omega_{g,n+1}(z_0, z_{[n]}) = -\sum_{ j \in [N]} \Res_{z = 0 \in C_j} \sum_{Z \subseteq \mf{f}' (z)} K^{1 + |Z|}(z_0; z, Z) \mc{W}'_{g,1+|Z|,n} (z,Z; z_{[n]})\,,
  \end{equation}
  where $ \mf{f}'(z) = x^{-1}(x(z)) \setminus \{z \}$ and the \emph{recursion kernels} are
  \begin{equation}
    K^{1+|Z|} (z_0; z, Z) \coloneqq \frac{\int_0^z \omega_{0,2} (\mathord{\cdot}, z_0)}{\prod_{z' \in Z} \big( \omega_{0,1}(z') - \omega_{0,1}(z) \big)} \, .
  \end{equation}
\end{theorem}

\subsection{The (deformed) \texorpdfstring{$(r,s)$}{rs}-spectral curves}

From now on we will focus on admissible local spectral curves with only one component ($N=1$); we will therefore drop the superscript $^{(j)}$ from the various expressions.

A particular example of the construction can be obtained from the $(r,s)$-Airy ideals of \cref{s:rsairystruct}. One can show that finding the partition function of the $(r,s)$-Airy ideals of \cref{t:rsAs} is equivalent to topological recursion on the $(r,s)$-spectral curve (see \cite{BBCCN18}), which is defined as follows:
 \begin{definition}\label{d:rs}
 Let $r  \in \mathbb{Z}$ such that $r \geq 2$, and $s \in [r+1]$ with $r = \pm 1 \pmod{s}$.
 The \emph{$(r,s)$-spectral curve} is given by $\mc{S} = (C, x, \omega_{0,1}, \omega_{\frac{1}{2},1} \omega_{0,2} ) $, where  $C$ is a small disk, $ x =  z^r$, $\omega_{0,1} = r z^{s-1}\ dz$, $\omega_{\frac{1}{2},1} = 0$, and
 \begin{equation}
  \omega^{\textup{std}}_{0,2} (z_1, z_2) = \frac{dz_1 dz_2}{(z_1 - z_2)^2}\,.
\end{equation}
 \end{definition}
If we define the meromorphic function $y$ on $C$ by $\omega_{0,1} = y\ dx$, then $y = z^{s-r}$. For $s \in [r-1]$, we can then think of $x(z)$ and $y(z)$ as a parametrization of the algebraic curve
\begin{equation}\label{eq:rsac1}
x^{r-s} y^r - 1 = 0.
\end{equation}
For $s=r+1$, we get a parametrization of the $r$-Airy algebraic curve
\begin{equation}\label{eq:rsac2}
y^r - x =0.
\end{equation}
We call these algebraic curves the $(r,s)$-algebraic curves.

In fact, the correspondence applies more generally to the deformed $(r,s)$-Airy ideals of \cref{p:rsAsgen}, so let us explain it in this more general setting. We define the deformed $(r,s)$-spectral curves in terms of the data introduced in \eqref{eq:complex}.

 \begin{definition}\label{d:rsdeformed}
 Let $r  \in \mathbb{Z}$ such that $r \geq 2$, and $s \in [r+1]$ with $r = \pm 1 \pmod{s}$. Pick complex numbers:
 \begin{equation}
  F_{0,1}[-k]\,, k \geq \min \{ s, r\} \,, \qquad F_{\frac{1}{2},1}[-k]\,, k > 0\,, \qquad F_{0,2}[-k,-l] \,, k,l > 0.
\end{equation}
 The \emph{deformed $(r,s)$-spectral curve} is given by $\mc{S} = (C, x, \omega_{0,1}, \omega_{\frac{1}{2},1} \omega_{0,2} ) $, where $C$ is a small disk, $x=z^r$,
\begin{align}
  \omega_{0,1} (z)
  &= 
  \sum_k F_{0,1} [-k] z^{k-1} dz \,,
  \\
  \omega_{\frac{1}{2},1} (z)
  &=
  \sum_k F_{\frac{1}{2},1} [-k] z^{k-1} dz \,,
  \\
  \omega_{0,2} (z_1, z_2)
  &=
  \omega^{\textup{std}}_{0,2} (z_1, z_2) + \sum_{k,l} F_{0,2}[-k,-l] z_1^{k-1} z_2^{l-1} dz_1 dz_2\,.
\end{align}
The $(r,s)$-spectral curve of \cref{d:rs} is recovered for the choice of numbers:
\begin{equation}
 F_{0,1} [-k] = r \delta_{k,s}, \qquad F_{\frac{1}{2},1}[-k] = 0, \qquad F_{0,2}[-k,-l] =0.
\end{equation}
 \end{definition}

To extract the loop equations from the deformed $(r,s)$-Airy ideal, we start with \cref{p:rsAsgen}. The claim is that the differential constraints  for the partition function of the deformed $(r,s)$-Airy ideal can be recast as the statement that there exists a system of correlators on the deformed $(r,s)$-spectral curve that satisfies the loop equations and the projection property. 

For clarity of notation, let us introduce the notation:
\begin{equation}\label{eq:unshiftedH}
H^i_k := \rho'( W^{\hslash,i}_k) 
\end{equation}
for the operators \eqref{eq:repp} generating the deformed $(r,s)$-Airy ideal. The differential constraints then take the form
\begin{equation}
  H^i_k Z = 0 \,, \quad i \in [r], \quad k \geq - \lfloor \frac{s(i-1)}{r} \rfloor\,.
\end{equation}
We first introduce the following fields constructed out of the differential operators $H^i_k$:
\begin{equation}
H^i(x)
  \coloneqq
  \sum_{k \in \Z} H^i_k \frac{dx^i}{x^{k+i}}.
\end{equation}
We also introduce the following notation, recalling the definition of the modes $J_k$ in \eqref{eq:Js} and the basis of one-forms from \cref{d:xibasis}:
\begin{align}
  \mc{J}_-(z)
  &\coloneqq
  \sum_{k>0} J_k \xi_{-k}(z)\,,
  \\
  \mc{J}_+(z)
  &\coloneqq
  \sum_{k>0} J_{-k} \xi_k (z)\,.
  \end{align}
Using this notation, we can rewrite the differential  operators $H^i_k$ more explicitly in terms of the data of the deformed $(r,s)$-spectral curve from \cref{d:rsdeformed}.
\begin{proposition}[{\cite[Section~4.1]{BKS23}}]\label{p:hh}
  For a set $ S$, let $ \mc{P}(S)$ be the set whose elements are disjoint sets of pairs in $ S$, and for $ P \in \mc{P}(S) $, write $ \sqcup P = \bigsqcup_{p \in P} p \subseteq S$. Then
  \begin{equation}
    \begin{split}
      rH^i(x) 
      &= 
      \sum_{\substack{Z \subseteq \mf{f}(z)\\ |Z| = i}} 
      \sum_{\substack{(\sqcup P) \sqcup A_0 \sqcup A_{\frac{1}{2}} \sqcup A_+ \sqcup A_- = Z \\ P \in \mc{P}(\mf{f}(z))}}
      \prod_{\{ z', z''\} \in P} \hslash^2 \omega_{0,2}(z',z'') 
      \prod_{z' \in A_0} \hslash \omega_{0,1} (z') 
      \\
      &\hspace{4cm}
      \prod_{z' \in A_{\frac{1}{2}}} \hslash^2 \omega_{\frac{1}{2},1}(z') 
      \prod_{z' \in A_+} \mc{J}_+ (z') 
      \prod_{z' \in A_-} \mc{J}_- (z').
    \end{split}
  \end{equation}
\end{proposition}
While this shape may not look very appealing at first, it is a useful form for extracting loop equations. 

Next, out of the partition function 
\begin{equation}
  Z = \exp \bigg( \sum_{ \substack{g \in \frac{1}{2} \N, n \in \N^* \\ 2g-2 + n > 0}} \frac{\hslash^{2g-2+n}}{n!} F_{g,n}[k_1, \dotsc, k_n] \prod_{j=1}^n x_{k_j} \bigg)
\end{equation} 
associated to the deformed $(r,s)$-Airy ideal, we construct a system of correlators on the deformed $(r,s)$-spectral curve.
\begin{definition}\label{d:srs}
Let $Z$ be the partition function associated to the deformed $(r,s)$-Airy ideal. For $2g-2+n>0$, we construct the following symmetric $n$-differentials on the deformed $(r,s)$-spectral curve:
\begin{equation}
  \omega_{g,n} (z_1, \dotsc, z_n ) \coloneqq \sum_{k_1, \dotsc, k_n = 1}^\infty F_{g,n}[k_1, \dotsc, k_n] \prod_{j=1}^n \xi_{-k_j} (z_j).
\end{equation}
\end{definition}

Since the correlators have finite expansions in the $\xi_{-k_j}(z_j)$ with $k_j >0$, it is clear that the system of correlators satisfies the projection property (see \cref{l:finiteness}):
\begin{lemma}\label{l:projprop}
The system of correlators $\{\omega_{g,n} \}_{g \in \frac{1}{2}\mathbb{N}, n \in \mathbb{N}^*}$ constructed above satisfies the projection property.
\end{lemma}

What we need to show is that this system of correlators also satisfies the loop equations, which is the key result:

\begin{proposition}[{\cite[Sections~4.3-4]{BKS23}}]
Let $Z$ be the partition function associated to the deformed $(r,s)$-Airy ideal, and define the system of correlators $\{\omega_{g,n} \}_{g \in \frac{1}{2}\mathbb{N}, n \in \mathbb{N}^*}$ on the deformed $(r,s)$-spectral curve as in \cref{d:srs} . Let
\begin{equation}
  G^i(x) \coloneqq Z^{-1} H^i(x) Z, \qquad i \in [r].
\end{equation}
Decompose the $G^i(x)$ in terms homogeneous separately in $\hslash$ and the $ x_j$ by
\begin{equation}
  G^i(x) \eqqcolon \sum_{g,n} \frac{\hslash^{2g+n}}{n!} G^i_{g,n}(x),
\end{equation}
where $G^i_{g,n}(x)$ is a homogeneous polynomial of degree $n$ in the variables $x_j$.
Then
\begin{equation}
\prod_{j=1}^n \ad_{\hslash^{-1}\mc{J}_-(z_j)} G^i_{g,n}(x) = \mc{E}^i_{g,n}(x; z_{[n]}),
\end{equation}
where $ \mc{E}^i_{g,n}(x; z_{[n]})$ is the object defined in \cref{d:EW} from the system of correlators $\{\omega_{g,n} \}_{g \in \frac{1}{2}\mathbb{N}, n \in \mathbb{N}^*}$ constructed from $Z$. Moreover, the system of correlators satisfies the loop equations:
  \begin{equation}
    \mc{E}^i_{g,n} (x; z_{[n]}) \in \mc{O} \left( x^{ \lfloor \frac{s(i-1)}{r}\rfloor + 1} \right) \left( \frac{dx}{x}\right)^i \,.
  \end{equation}
\end{proposition}

\begin{proof}
A proof of this proposition can be found in \cite{BBCCN18} and \cite{BKS23}. Basically, the differential constraints 
\begin{equation}
  H^i_k Z = 0 \,, \quad i \in [r]\,, \quad k \geq - \lfloor \frac{s(i-1)}{r} \rfloor\,,
\end{equation}
can be recast as the statement that
\begin{equation}
  G^i(x)  \in \mc{O} \left( x^{\lfloor \frac{s(i-1)}{r}\rfloor + 1} \right) \Big( \frac{dx}{x}\Big)^i \,.
\end{equation}
The rest follows combinatorially using \cref{p:hh}.
\end{proof}

What we have found is that, out of the data of the partition function $Z$ associated to the deformed $(r,s)$-Airy ideal, we can construct a system of correlators on the deformed $(r,s)$-spectral curve that satisfies both the loop equations and the projection property. Therefore, it can be calculated recursively from the data of the spectral curve by the topological recursion formula, see \cref{UnshiftedTR}.

\subsection{Shifted loop equations and shifted topological recursion}

We will now consider what happens to the story if we consider instead the shifted $(r,s)$-Airy ideals of \cref{t:shifts}, or, more generally, the deformed and shifted $(r,s)$-Airy ideals of \cref{p:shiftsgen}.

What we will show is that, out of the data of the partition function of the deformed and shifted $(r,s)$-Airy ideals, we can construct a new system of correlators on a shifted version of the deformed $(r,s)$-spectral curve of \cref{d:rsdeformed}. This system of correlators still satisfies the projection property, but it does not satisfy the usual loop equations. Instead, it satisfies a new set of equations, which we call ``shifted loop equations''. We then show there is a unique system of correlators that satisfies the shifted loop equations and the projection property, and it can be constructed recursively from the data of the spectral curve by a shifted version of the topological recursion formula.

We use the notation from the previous section. Let us first define a shifted deformed $(r,s)$-spectral curve:
 \begin{definition}\label{d:rsdefshift}
 Let $r  \in \mathbb{Z}$ such that $r \geq 2$, and $s \in [r+1]$ with $r = \pm 1 \pmod{s}$. Pick complex numbers:
 \begin{gather}
  F_{0,1}[-k]\,,\, k \geq \min \{ s, r\} \,, \qquad F_{\frac{1}{2},1}[-k]\,,\, k > 0\,,\nonumber\\ F_{0,2}[-k,-l] \,,\, k,l > 0, \qquad S_{i,1}\,,\, i \in [r].
\end{gather}
Assume that the set of shifts $\{ S_{i,1} \}_{i \in [r]}$ is $s$-consistent (see \cref{d:consistent}).

 The \emph{shifted deformed $(r,s)$-spectral curve} is given by $\mc{S} = (C, x, \omega_{0,1}, \omega_{\frac{1}{2},1} \omega_{0,2} ) $, where $C$ is a small disk, $x=z^r$,
\begin{align}
 \omega_{0,1} (z)
  &= 
  \sum_k F_{0,1} [-k] z^{k-1} dz \,,
  \\
\omega_{\frac{1}{2},1} (z)
  &=
  \sum_k F_{\frac{1}{2},1} [-k] z^{k-1} dz +\sum_{i=1}^{r}(-1)^{i-1}S_{i,1}\frac{dz}{z^{s(i-1)+1}},
  \\
 \omega_{0,2} (z_1, z_2)
  &=
  \omega^{\textup{std}}_{0,2} (z_1, z_2) + \sum_{k,l} F_{0,2}[-k,-l] z_1^{k-1} z_2^{l-1} dz_1 dz_2\,.
\end{align}
We note that the only difference with \cref{d:rsdeformed} is that $\omega_{\frac{1}{2},1}(z)$ is shifted by terms linear in the constants $S_{i,1}$, which are the $O(\hslash)$ terms in the shifts of the differential operators of the shifted $(r,s)$-Airy ideals. We define the \emph{shifted $(r,s)$-spectral curve} as being the particular case with:
\begin{equation}
 F_{0,1} [-k] = r \delta_{k,s}, \qquad F_{\frac{1}{2},1}[-k] = 0, \qquad F_{0,2}[-k,-l] =0.
\end{equation}
That is, we set the deformations to zero, and recover a shifted version of the original $(r,s)$-spectral curves of \cref{d:rs}. We can still think of the shifted $(r,s)$-spectral curve as a parametrization of the $(r,s)$-algebraic curves of \eqref{eq:rsac1} and \eqref{eq:rsac2}, but with a non-trivial $\omega_{\frac{1}{2},1}(z)$ introduced by the shifts.
 \end{definition}

We start with the differential constraints from \cref{p:shiftsgen}. For simplicity of notation, we write 
 \begin{equation}
  H^i_k :=  \rho'(W^{\hslash,i}_k(S)) =(H^i_k)^{\textup{unshifted}} -\delta_{k,0} \sum_{\ell=1}^\infty \hslash^\ell S_{i,\ell},
 \end{equation}
 where we assume that the set of shifts is $s$-consistent. Here, $(H^i_k)^{\textup{unshifted}}$ refers to the unshifted differential operators of \eqref{eq:unshiftedH}; we just wanted to highlight the fact that the only difference with the previous $(r,s)$ case is that we shift the differential operators $H^i_0$ (the zero modes) by the set of $s$-consistent shifts $\{S_{i,\ell} \}_{i \in [r], \ell \in \mathbb{N}^*}$.
 
We write 
\begin{equation}
 Z = \exp \bigg( \sum_{ \substack{g \in \frac{1}{2} \N, n \in \N^* \\ 2g-2+n > 0}} \frac{\hslash^{2g-2+n}}{n!} F_{g,n}[k_1, \dotsc, k_n] \prod_{j=1}^n x_{k_j} \bigg)
\end{equation}
for the partition function associated to the deformed and shifted $(r,s)$-Airy ideal. It satisfies the differential constraints
\begin{equation}
  H^i_k Z = 0 \,, \quad i \in [r]\,, \quad k \geq - \lfloor \frac{s(i-1)}{r} \rfloor\, .
\end{equation}
As in \cref{d:srs}, out of the partition function we construct a system of correlators on the shifted deformed $(r,s)$-spectral curve. As before, it is clear that the system of correlators satisfies the projection property:
\begin{lemma}\label{l:projpropshift}
The system of correlators $\{\omega_{g,n} \}_{g \in \frac{1}{2}\mathbb{N}, n \in \mathbb{N}^*}$ on the shifted deformed $(r,s)$-spectral curve constructed from the partition function $Z$ as in \cref{d:srs} satisfies the projection property.
\end{lemma}

The question is whether it satisfies loop equations, which is the subject of the next proposition.

\begin{proposition}\label{p:shifteloopeq}
Let $Z$ be the partition function associated to the deformed and shifted $(r,s)$-Airy ideal. Let
\begin{equation}
 G^i(x) \coloneqq Z^{-1} H^i(x) Z  \eqqcolon \sum_{g,n} \frac{\hslash^{2g+n}}{n!} G^i_{g,n}(x)
\end{equation}
where the $G^i_{g,n}(x)$ are homogeneous polynomials of degree $n$ in the variables $x_j$.
Then
\begin{equation}
\prod_{j=1}^n \ad_{\hslash^{-1}\mc{J}_-(z_j)} G^i_{g,n}(x) = \mc{E}^i_{g,n}(x; z_{[n]}) - \delta_{n,0} S_{i,2g} \Big( \frac{dx}{x}\Big)^i ,
\end{equation}
where $ \mc{E}^i_{g,n}(x; z_{[n]})$ is the object defined in \cref{d:EW} from the system of correlators $\{\omega_{g,n} \}_{g \in \frac{1}{2}\mathbb{N}, n \in \mathbb{N}^*}$ constructed from $Z$. 
Moreover, the system of correlators satisfy the \emph{shifted loop equations}:
  \begin{equation}\label{eq:sle}
    \mc{E}^i_{g,n} (x; z_{[n]}) - \delta_{n,0} S_{i,2g} \Big( \frac{dx}{x}\Big)^i  \in \mc{O} \left( x^{ \lfloor \frac{s(i-1)}{r}\rfloor + 1} \right) \left( \frac{dx}{x}\right)^i \,.
  \end{equation}
\end{proposition}

\begin{proof}
  The proof is completely analogous to \cite[Sections~4.3-4]{BKS23}. As the only difference between $ H^i_k$ and $ (H^i_k)^{\textup{unshifted}}$ is an additive constant, i.e. central element, the conjugation by $ Z$ just keeps this constant.\par
  When decomposing $ G^i (x)$ into $G^i_{g,n}$, the shifts are in polynomial degree zero, so they only contribute to $ n =0$, and the $ \hslash^{\ell}$ should be matched to $ \hslash^{2g+n}$, so $ \ell = 2g$ (which is also why the $S_{i,1}$ shifts, i.e. with $g=\frac{1}{2}$, contribute to the initial condition ${\omega}_{\frac{1}{2},1}(z)$ in \cref{d:rsdefshift}). Then the calculation of the $ \mc{E}^i_{g,n}$ is the same as in the unshifted case -- the adjunctions $ \ad_{\mc{J}_-(z_j)}$ only act on the unshifted modes.
\end{proof}

What we have shown is that the differential constraints of the deformed and shifted $(r,s)$-Airy ideals are equivalent to the existence of a system of correlators on the shifted deformed $(r,s)$-spectral curve that satisfies both the projection property and the shifted loop equations. As the shifted loop equations are not the same as the usual loop equations, a natural question then is to determine whether these correlators can be reconstructed recursively via a modification of the topological recursion formula.
%

We start with the following combinatorial lemma, which is essential for proving topological recursion.

\begin{lemma}[{\cite[Lemma~7.6.4]{Kra19}}]\label{CombinatoricsOfE}
Given a system of correlators $ \{ \omega_{g,n} \}_{g \in \frac{1}{2} \N, n\in \N^*}$ on an admissible local spectral curve (with $N=1$), define  $\mc{W}_{g,i,n}' (z_{[i]} ; w_{[n]})$ and $\mc{E}^{i}_{g,n}(x; z_{[n]})$ as  in \cref{d:EW}. Then
  \begin{equation}\label{CombinatorialIdentity}
    \sum_{ \{ z\} \subseteq Z \subseteq \mf{f}(z)} \! \mc{W}_{g,|Z|,n}' (Z; z_{[n]}) \!\! \prod_{z' \in \mf{f}'(z) } \!\! \Big( \omega_{0,1} (z') - \omega_{0,1} (z) \Big)
    =
    \sum_{i=1}^r \mc{E}^i_{g,n}(x; z_{[n]}) \big( - \omega_{0,1}(z)\big)^{r-i}
  \end{equation}
\end{lemma}

In the usual proof that the topological recursion formula reconstructs the unique solution of the loop equations satisfying the projection property, that is, \cref{UnshiftedTR}, a key step is to use the fact that the right side of \cref{CombinatorialIdentity} has a certain vanishing order, which causes it to drop out of a residue formula. For the case of the shifted loop equations, this is no longer the case, because of the extra shift in the shifted loop equations \eqref{eq:sle}. 
As a result, we must add these shifts to the topological recursion formula.

\begin{theorem}\label{ShiftedTR}
Let $\mc{S}$ be the shifted deformed $(r,s)$-spectral curve of \cref{d:rsdefshift}. Let $S = \{S_{i,\ell} \}_{i \in [r], \ell \in \mathbb{N}^*}$ be a set of $s$-consistent shifts. Then there exists exactly one system of correlators $\{\omega_{g,n} \}_{g \in \frac{1}{2}\mathbb{N}, n \in \mathbb{N}^*}$ that satisfies the shifted loop equations \eqref{eq:sle} and the projection property. It can be calculated recursively by the \emph{shifted topological recursion formula} (for $2g-2+n > 0$):
  \begin{equation}
    \begin{split}
     \omega_{g,n+1}(z_0, z_{[n]}) 
      &= 
      -\Res_{z = 0} \Big( \sum_{Z \subseteq \mf{f}' (z)}
       K^{1 + |Z|}(z_0; z, Z) \mc{W}'_{g,1+|Z|,n} (z,Z; z_{[n]})
      \\
      &\hspace{1.2cm}
      - \sum_{i=1}^{r}\delta_{n,0} S_{i,2g} K^r(z_0; \mf{f}(z)) \left( r\frac{dz}{z} \right)^{i} \big( - \omega_{0,1}(z)\big)^{r-i}\Big)\,.
    \end{split}
  \end{equation}
  In particular, this formula does produce symmetric correlators.
\end{theorem}

\begin{proof}
  We again emulate \cite[Proposition~5.10]{BKS23}.\par
  Since by definition $ \mc{W}'_{g,1,n} = \omega_{g,n+1}$, the projection property, the definition of the recursion kernel, and \cref{CombinatorialIdentity} yield
  \begin{equation}
    \begin{split}
      \omega_{g,n+1}(z_0, z_{[n]}) 
      &=
       \Res_{z = 0} \Big( \int_0^z \omega_{0,2} (\mathord{\cdot},z_0) \Big) \mc{W}'_{g,1,n} (z; z_{[n]})
      \\
      &=
       \Res_{z = 0} K^r(z_0; \mf{f}(z)) \mc{W}'_{g,1,n} (z; z_{[n]}) \!\! \prod_{z' \in \mf{f}(z) \setminus \{ z\}} \! \Big( \omega_{0,1}(z') - \omega_{0,1}(z)\Big)
      \\
      &=
      - \Res_{z = 0 } K^r(z_0; \mf{f}(z)) \Big( \sum_{ \{ z\} \subsetneq Z \subseteq \mf{f}(z)} \! \mc{W}_{g,|Z|,n}' (Z; z_{[n]}) \!\! \prod_{z' \in \mf{f}(z) \setminus Z} \!\! \Big( \omega_{0,1} (z') - \omega_{0,1} (z) \Big)
      \\
      & \hspace{5cm}-
      \sum_{i=1}^r \mc{E}^i_{g,n}(x; z_{[n]}) \big( - \omega_{0,1}(z)\big)^{r-i}\Big)
      \\
      &=
      - \Res_{z = 0 } K^r(z_0; \mf{f}(z)) \Big( \sum_{ \{ z\} \subsetneq Z \subseteq \mf{f}(z)} \! \mc{W}_{g,|Z|,n}' (Z; z_{[n]}) \!\! \prod_{z' \in \mf{f}(z) \setminus Z} \!\! \Big( \omega_{0,1} (z') - \omega_{0,1} (z) \Big)
      \\
      &\hspace{1cm}-
      \sum_{i=1}^r \Big( \left\{\mc{E}^i_{g,n}(x; z_{[n]})- \delta_{n,0} S_{i,2g} \big( \frac{dx}{x} \big)^i \right\}+ \delta_{n,0} S_{i,2g} \big( \frac{dx}{x} \big)^i\Big) \big( - \omega_{0,1}(z)\big)^{r-i}\Big)
    \end{split}
  \end{equation}
  The terms in curly brackets in the last line do not contribute, because by the shifted loop equations \eqref{eq:sle} and standard order counting they give holomorphic terms to the integrand. As for the term in the second last line, the kernels can be simplified as usual, giving
  \begin{equation}
    \begin{split}
      \omega_{g,n+1}(z_0, z_{[n]}) 
      &=
      -\Res_{z = 0 } \Big( \sum_{ \{ z\} \subsetneq Z \subseteq \mf{f}(z)} K^{|Z|}(z_0; Z)\mc{W}_{g,|Z|,n}' (Z; z_{[n]})
      \\
      & \hspace{2cm}-
      \sum_{i=1}^{r}\delta_{n,0} S_{i,2g} K^r(z_0; \mf{f}(z)) \big( r\frac{dz}{z} \big)^{i} \big( - \omega_{0,1}(z)\big)^{r-i}\Big)
    \end{split}
  \end{equation}
  We finally change the meaning of $Z$ to not include $z$ to give the statement of the theorem.
\end{proof}

\begin{remark}\label{Rem:GenShiftedTR}
We wrote the derivation of shifted topological recursion from shifted loop equations only for the deformed $(r,s)$-spectral curves, but it can easily be generalized to arbitrary admissible local spectral curves.
\end{remark}

The upshot of the story is the following. On the one hand, from the partition function of the deformed $(r,s)$-Airy ideal, we can construct a system of correlators on the deformed $(r,s)$-spectral curve that satisfies the usual topological recursion formula. On the other hand, from the partition function of the shifted and deformed $(r,s)$-Airy ideal, we can also construct a system of correlators, this time on the shifted deformed $(r,s)$-spectral curve (which is the same spectral curve as  before but with $\omega_{\frac{1}{2},1}(z)$ shifted), but it satisfies a shifted version of topological recursion. The only difference in the shifted topological recursion formula is that we shift the correlators $\omega_{g,1}(z)$ (including the initial condition $\omega_{\frac{1}{2},1}(z)$) by extra terms -- but of course, the shifts propagate through the recursion formula and produce an entirely different system of correlators. It is worth noting that both systems of correlators (shifted and unshifted) satisfy the projection property.

In particular, for the case $s=1$, in which case all shifts are allowed, as we saw in \cref{s:shiftedrs} the partition function $Z$ is a highest weight vector for the $\mathcal{W}(\mathfrak{gl}_r)$-algebra at self-dual level. What we have shown is that the highest weights appear in the topological recursion formula as extra shifts of the correlators $\omega_{g,1}(z)$. Neat!

\section{Quantum curves}
\label{s:qcurves}

In this section, for simplicity we focus on the shifted $(r,s)$-Airy ideals of \cref{t:shifts} and the corresponding system of correlators on the shifted $(r,s)$-spectral curve from \cref{d:rsdefshift} where the deformations are set to zero. 

To recap: in \cref{s:shifted} we constructed a new class of Airy ideals, which we called ``shifted $(r,s)$-Airy ideals'', see \cref{t:shifts}. For $s=1$, the partition function of those shifted Airy ideals gives a highest weight vector for the $\mathcal{W}(\mathfrak{gl}_r)$-algebras at self-dual level. We then showed in \cref{s:shiftedle} that out of the partition function $Z$ of a shifted $(r,s)$-Airy ideal we can construct a system of correlators that lives on a shifted version of the $(r,s)$-spectral curve. We showed that this system of correlators is the unique solution that satisfies the projection property and a variation of the usual loop equations, which we called ``shifted loop equations'' (see \cref{p:shifteloopeq}). Finally, we proved that this solution can be reconstructed recursively from the data of the spectral curve via a ``shifted topological recursion'' formula (see \cref{ShiftedTR})

In this section we show that shifted topological recursion on the shifted $(r,s)$-spectral curve reconstructs the WKB solution to a quantum curve, where the particular quantization depends on the shifts. 

\subsection{The topological recursion/quantum curve correspondence}

Let us start by briefly reviewing the topological recursion/quantum curve correspondence. The intuition is that topological recursion should provide a procedure for quantizing the spectral curve. The statement originates from matrix models \cite{BE09} but can be formulated abstractly in terms of topological recursion itself.

We focus on spectral curves that are constructed as parametrizations of an algebraic curve:
\begin{equation}
  \label{PlaneCurve}
  C=\{P(x,y) = 0 \} \subset \mathbb{C}^2.
\end{equation}
Topological recursion produces a system of correlators $\{\omega_{g,n}\}_{g \in \frac{1}{2}\mathbb{N}, n \in \mathbb{N}^*}$ on the spectral curve. Out of those, one can construct the \emph{wave function}:
\begin{equation}\label{eq:wf}
\psi(z) = \exp \left(\sum_{g \in \mathbb{N}, n \in \mathbb{N}^*} \frac{\hslash^{2g-2+n}}{n!} \left(\int^z_\alpha \cdots \int^z_\alpha \omega_{g,n} - \delta_{g,0} \delta_{n,2} \frac{dx(z_1) dx(z_2)}{(x(z_1)-x(z_2))^2} \right) \right),
\end{equation}
where $\alpha$ is a base point on the normalization of $C$ (that is not a ramification point of $x$) -- it is usually taken to be a pole of $x$. Here we are integrating the correlators $\omega_{g,n}$ in all variables from $\alpha$ to the same variable $z$. 

To state the TR/QC correspondence, we introduce the notion of a quantum curve.

\begin{definition}\label{d:QC}
Let $C = \{P(x,y) = 0\} \subset \mathbb{C}^2$ of degree $d$ in $y$. 
A \emph{quantum curve} $\hat{P}$ of $C$ is an order $d$ linear differential operator in $x$, such that, after normal ordering, it takes the form
\begin{equation}
\hat{P}\left(  x, \hslash \frac{d}{dx}; \hslash \right) = P\left(  x,\hslash \frac{d}{dx}\right) + \sum_{n \geq 1} \hslash^n P_n\left(x,\hslash \frac{d}{dx}\right),
\end{equation}
where the leading term $P$ is the original polynomial defining the spectral curve, and the $P_n$ are (normal-ordered) polynomials of degree $< d$. We usually impose that only finitely many correction terms $P_n$ are non-vanishing.
\end{definition}

This is a quantization of the spectral curve, as it amounts to replacing $(x,y) \mapsto (\hat{x}, \hat{y})= \left( x, \hslash \frac{d}{dx} \right)$. But of course, this process is not unique, since the operators $\hat{x}$ and $\hat{y}$ do not commute, and hence the quantization may include $\hslash$ corrections.

The claim of the TR/QC correspondence is that, given a spectral curve $C$, there exists a quantum curve $\hat{P}$ such that
\begin{equation}
\hat{P} \psi = 0,
\end{equation}
where $\psi$ is the wave function constructed by topological recursion via \eqref{eq:wf}.

This correspondence has been studied in many papers for various spectral curves relevant to enumerative geometry. More generally, the correspondence was proved in \cite{BE17} for a large class of genus zero algebraic spectral curves with arbitrary ramification (the class corresponds to all genus zero spectral curves whose Newton polygon has no interior point and that are smooth as affine curves). More recently, it was proved in \cite{EGMO21} for all algebraic spectral curves (any genus) that only have simple ramification points (for spectral curves of genus $\geq 1$, the definition of the wave function must be modified to take into account non-perturbative contributions).  As a generic spectral curve only has simple ramification points, and in principle spectral curves with higher ramification can be obtained as limit points in families of curves with only simple ramification (see \cite{BBCKS}), the correspondence is expected to hold in full generality for all algebraic spectral curves.\footnote{It is also expected to hold for (at least some) non-algebraic spectral curves, and it has been proved in some such cases relevant to enumerative geometry.}

%

\subsubsection{Choices of ordering in the quantum curve}

It is important to note that given a spectral curve, the construction of a quantum curve $\hat{P}$ is not unique; since the operators $\hat{x} = x$ and $\hat{y} = \hslash  \frac{d}{dx}$ do not commute, there  is an inherent choice of ordering. Topological recursion seems to select a particular choice of ordering. This is however not quite true; as shown in \cite{BE17}, 
for a given spectral curve, different choices of integration divisors in the definition of the wave function \eqref{eq:wf} lead to quantum curves in various choices of orderings.

Nevertheless, if we focus on the $(r,s)$-spectral curves of \cref{d:rs} (or the shifted $(r,s)$-spectral curves of \cref{d:rsdefshift}, as it comes from the same $(r,s)$-algebraic curves \eqref{eq:rsac1} and \eqref{eq:rsac2}), there is only one choice of integration divisor that works, namely the unique pole of $x$ at $z= \infty$. Thus, it seems that topological recursion selects a particular ordering for the quantization of the $(r,s)$-spectral curve. In fact, as shown in \cite{BE17}, this quantization is not the one that you would obtain by simple normal ordering of the operators $\hat{x}$ and $\hat{y}$; instead, the result is a particular quantization in a peculiar choice of ordering. This raises an interesting question: how can we obtain other choices of ordering for these spectral curves, since we cannot consider other integration divisors?

Interestingly, what we will show is that shifted topological recursion produces wave functions that are WKB solutions of quantizations of the  $(r,s)$-algebraic curve in other choices of ordering. In particular, for the cases $s=1$ and $s=r-1$, through shifted topological recursion we obtain all possible choices of ordering of the quantum curve.



\subsection{The quantum curve}

Let us now calculate the quantum curves associated to shifted topological recursion on the shifted $(r,s)$-spectral curves. The calculation will primarily follow the same steps as in \cite[Sections~3-5]{BE17}, and we will simply fill in the details that differ. For each lemma, corollary or theorem that we generalize, we write in square brackets the corresponding statement in \cite{BE17} so that the reader can easily follow and compare.

We start with the shifted $(r,s)$-spectral curve of \cref{d:rsdefshift}, with the deformations set to zero. We can still think of the shifted $(r,s)$-spectral curve as  a parametrization of the $(r,s)$-algebraic curves \eqref{eq:rsac1} and \eqref{eq:rsac2}, but with a non-zero initial condition
\begin{equation}
\omega_{\frac{1}{2},1} (z)
  =\sum_{i=1}^{r}(-1)^{i-1}S_{i,1}\frac{dz}{z^{s(i-1)+1}}
\end{equation}
specified by the $O(\hslash)$ terms in the shifts. 

\begin{remark}\label{r:compact}
In fact, it will be important for us that the correlators $\omega_{g,n}$, which in principle from our definition of admissible local spectral curves are only defined on $C^n$ where $C$ is an open disk, can be extended to symmetric differential forms on the compact Riemann surface $\Sigma = \mathbb{P}^1$, where we think of $z$ as a projective coordinate on $\Sigma$. In other words, we think of the correlators as symmetric differential forms on $\Sigma^n$ with only poles at $z=0$ in each variable. This is standard in the theory of topological recursion, see for instance \cite{BBCKS}.
\end{remark}

We are now in a position to carry out the steps of the calculation. It proceeds in four steps:
\begin{enumerate}
\item From the shifted loop equations, we construct recursion relations for the objects $\mathcal{U}^{i}_{g,n}$ defined in the next section.
\item We integrate these relations and sum over $\hslash$.
\item We rewrite the resulting recursion relations as a system of differential equations.
\item We show that the system is equivalent to a single differential equation for the wave function $\psi(z)$, which defines the quantum curve.
\end{enumerate}

\subsubsection{Recursion relation for the \texorpdfstring{$\mathcal{U}^{i}_{g,n}$}{Uign}}

Let us start by defining the objects $\mathcal{U}^{i}_{g,n}$ from the correlators of the shifted $(r,s)$ spectral curve:

\begin{definition}
For $i = 1, \ldots, r-1$ and all $g, n \geq 0$,
\begin{equation}
\mathcal{U}^{i}_{g,n}(z;z_{[n]}) = \sum_{\substack{Z\subseteq\mathfrak{f}^{\prime}(z)\\|Z|=i}} \mc{W}_{g,i,n}(Z; z_{[n]}),
\end{equation}
where the $\mc{W}_{g,i,n}(Z; z_{[n]})$ were defined in \cref{d:EW},
and we set
\begin{equation}
\mathcal{U}^{0}_{g,n} = \delta_{g,0}\delta_{n,0}.
\end{equation}
\end{definition}

Recall also the objects $\mc{E}^{i}_{g,n}(x; z_{[n]})$ from \cref{d:EW}. For consistency, we define
\begin{equation}
\mathcal{E}^{0}_{g,n} = \delta_{g,0}\delta_{n,0}.
\end{equation}

\begin{remark}
Note that here and throughout this section, we use the notation $x$ to mean $x(z)$, and
 $x_{j}$ to mean $x(z_{j})$ for $j \in [n]$.
\end{remark}

We follow the calculation in \cite{BE17}. The first deviation from the original calculation appears in \cite[Lemma~3.25]{BE17} which now reads:

\begin{lemma}[see {\cite[Lemma~3.25]{BE17}}]\label{E1}
For $2g - 2 + n \geq 0$,
\begin{equation}
\mathcal{E}^{1}_{g,n}(x;z_{[n]}) = \delta_{n,0}S_{1,2g}\frac{dx}{x},
\end{equation}
and the remaining cases are given by
\begin{gather}
\mathcal{E}^{1}_{0,0}(x) = -\frac{p_{1}(x)}{p_{0}(x)}dx,\\
\mathcal{E}^{1}_{\2,0}(x) = S_{1,1}\frac{dx}{x},\\
\mathcal{E}^{1}_{0,1}(x;z_{1}) = \frac{dxdx_{1}}{(x - x_{1})^2}.
\end{gather}
\end{lemma}

\begin{proof}
Nothing has changed for the cases $(g,n) = (0,0), (0,1)$. For $2g - 2 + n \geq 0$, the shifted loop equations \eqref{eq:sle} tell us that
\begin{equation}
\mathcal{E}^{1}_{g,n}(x;z_{[n]})   - \delta_{n,0} S_{1,2g} \frac{dx}{x}\in O(1)\ dx.
\end{equation}
By \cref{r:compact}, the correlators $\omega_{g,n}$ are defined on $\Sigma^n$ where $\Sigma = \mathbb{P}^1$ and only have poles at $z=0$. This means that $\mathcal{E}^{1}_{g,n}(x;z_{[n]}) - \delta_{n,0} S_{1,2g} \frac{dx}{x}$ is holomorphic everywhere except potentially at $z=\infty$ where $\frac{dx}{x}$ has a simple pole. But the residue here is clearly zero since $\Sigma  = \mathbb{P}^1$. Thus, $\mathcal{E}^{1}_{g,n}(x;z_{[n]}) - \delta_{n,0} S_{1,2g} \frac{dx}{x}$ is bounded and entire, so it must be constant. However, if we examine the form of $\mathcal{E}^{1}_{g,n}$ in \cref{d:EW}, it is clear that the constant is simply zero, and then the result follows immediately. Finally, when $(g,n) = (\2,0)$ we have
\begin{equation}
\begin{aligned}
\mathcal{E}^{1}_{\2,0}(x) & = \sum_{z^{\prime}\in\mathfrak{f}(z)}\omega_{\2,1}(z^{\prime})\\
& = \sum_{z^{\prime}\in\mathfrak{f}(z)}\sum_{k=1}^{r}(-1)^{k-1}S_{k,1}\frac{dz^{\prime}}{z^{\prime s(k-1)+1}}\\
& = \sum_{k=1}^{r}(-1)^{k-1}S_{k,1}\frac{dz}{z^{s(k-1)+1}}\sum_{m=1}^{r}\theta^{ms(1-k)}\\
& = \sum_{k=1}^{r}\delta_{k,1}(-1)^{k-1}S_{k,1}r\frac{dz}{z^{s(k-1)+1}}\\
& = S_{1,1}\frac{dx}{x}
\end{aligned}
\end{equation}
as desired.
\end{proof}

We thus obtain:

\begin{corollary}[see {\cite[Corollary~4.6]{BE17}}]\label{E-UCombinatorics}
For $i = 1,\ldots, r$ and all $g,n \geq 0$,
\begin{equation}
\begin{aligned}
\mathcal{E}^{i}_{g,n}(x;z_{[n]}) &-  \delta_{n,0} S_{i,2g}  \Big( \frac{dx}{x}\Big)^i= \mathcal{U}^{i}_{g,n}  (z; z_{[n]}) + \mathcal{U}^{i-1}_{g-1,n+1}(z;z_{[n]},z)\\
 &-\sum_{N_{1}\sqcup N_{2}=z_{[n]}}\sum_{g_{1}+g_{2}=g}\mathcal{U}^{i-1}_{g_{1},|N_{1}|}(z; N_{1})\mathcal{U}^{1}_{g_{2},|N_{2}|}(z; N_{2}) - \frac{p_{1}(x)}{p_{0}(x)}dx \, \mathcal{U}^{i-1}_{g,n}(z; z_{[n]})\\
& + \sum_{h=\2}^{g}S_{1,2h}\frac{dx}{x}\mathcal{U}^{i-1}_{g-h,n}(z; z_{[n]}) + \sum_{j=1}^{n}\frac{dx dx_{j}}{(x - x_{j})^2}\mathcal{U}^{i-1}_{g,n-1}(z; z_{[n]\setminus \{j\}})\\
& - \delta_{n,0}S_{i,2g}\left(\frac{dx}{x}\right)^{i}.
\end{aligned}
\end{equation}
\end{corollary}

\begin{proof}
A simple argument in combinatorics (cf.\ \cite[Lemma~4.5]{BE17}) yields that
\begin{equation}
\begin{aligned}
\mathcal{E}^{i}_{g,n}(x;z_{[n]}) = \mathcal{U}^{i}_{g,n} & (z; z_{[n]}) + \mathcal{U}^{i-1}_{g-1,n+1}(z;z_{[n]},z)\\
& +\sum_{N_{1}\sqcup N_{2}=z_{[n]}}\sum_{g_{1}+g_{2}=g}\mathcal{U}^{i-1}_{g_{1},|N_{1}|}(z; N_{1})\omega_{g_{2},|N_{2}|+1}(z,N_{2}),
\end{aligned}
\end{equation}
which directly implies that
\begin{equation}\label{E1-U1}
\omega_{g,n+1}(z,z_{[n]}) = \mathcal{E}^{1}_{g,n}(x;z_{[n]}) - \mathcal{U}^{1}_{g,n}(z;z_{[n]}).
\end{equation}
Substituting this into the previous expression, applying \cref{E1}, and adding the shifts gives the statement of the corollary.
\end{proof}

\noindent
Before we can find the recursion relations for the $\mathcal{U}^{i}_{g,n}$ to replace \cite[Lemma~4.13]{BE17}, we need to modify \cite[Lemma~4.8]{BE17} and add the case $(g, n) = (\2, 0)$.

\begin{lemma}[see {\cite[Lemma~4.8]{BE17}}]\label{EiOneHalf}
For $i = 1, \ldots, r$,
\begin{equation}
\mathcal{E}^{i}_{0,0}(x)  = (-1)^{i}\frac{p_{i}(x)}{p_{0}(x)}dx^{i} 
\end{equation}
and
\begin{equation}
\mathcal{E}^{i}_{\2,0}(x) =  S_{i,1}  \Big( \frac{dx}{x}\Big)^i.
\end{equation}
\end{lemma}

\begin{proof}
For $(g,n) = (0,0)$, there are no shifts so \cite[Lemma~4.8]{BE17} is unchanged. For the remaining case, we have
\begin{equation}
\begin{aligned}
\mathcal{E}^{i}_{\2,0}(x) & = \sum_{\substack{Z\subseteq\mathfrak{f}(z)\\|Z|=i}}\sum_{z^{\prime}\in Z}\omega_{\2,1}(z^{\prime})\prod_{z^{\prime\prime} \in Z\setminus\{z^{\prime}\}}\omega_{0,1}(z^{\prime\prime})\\
& = \sum_{\substack{Z\subseteq\mathfrak{f}(z)\\|Z|=i}}\sum_{z^{\prime}\in Z}\left(\sum_{k=1}^{r}(-1)^{k-1}S_{k,1}\frac{dz^{\prime}}{z^{\prime s(k-1)+1}}\right)\prod_{z^{\prime\prime} \in Z\setminus\{z^{\prime}\}}\omega_{0,1}(z^{\prime\prime})\\
& = \sum_{k=1}^{r}\sum_{z^{\prime}\in\mathfrak{f}(z)}(-1)^{k-1}S_{k,1}\frac{dz^{\prime}}{z^{\prime s(k-1)+1}}\sum_{\substack{Z\subseteq\mathfrak{f}^{\prime}(z^{\prime})\\|Z|=i-1}}\prod_{z^{\prime\prime} \in Z}\omega_{0,1}(z^{\prime\prime})\\
& = \sum_{k=1}^{r}\sum_{z^{\prime}\in\mathfrak{f}(z)}(-1)^{k-1}S_{k,1}\frac{dz^{\prime}}{z^{\prime s(k-1)+1}}\mathcal{U}^{i-1}_{0,0}(z^{\prime}).
\end{aligned}
\end{equation}
Moreover,
\begin{equation}
\mathcal{U}^{i-1}_{0,0}(z) = (-\omega_{0,1}(z))^{i-1}
\end{equation}
for the shifted $(r,s)$-spectral curve (cf.\ \cite[eq.\ (4.8)]{BE17}), and therefore
\begin{equation}
\begin{aligned}
\mathcal{E}^{i}_{\2,0}(x) -  S_{i,1}  \Big( \frac{dx}{x}\Big)^i & = \sum_{k=1}^{r}(-1)^{k-1}S_{k,1}\frac{dz}{z^{s(k-1)+1}}(-rz^{s-1}dz)^{i-1}\sum_{m=1}^{r}\theta^{ms(i-k)} - S_{i,1}\left(\frac{dx}{x}\right)^{i}\\
& = \sum_{k=1}^{r}\delta_{k,i}(-1)^{i+k}S_{k,1}\left(r\frac{dz}{z}\right)^{i}z^{s(i-k)} - S_{i,1}\left(\frac{dx}{x}\right)^{i}\\
& = S_{i,1}\left(\frac{dx}{x}\right)^{i} - S_{i,1}\left(\frac{dx}{x}\right)^{i}\\
& = 0
\end{aligned}
\end{equation}
as desired.
\end{proof}

\noindent
This brings us to \cite[Theorem~4.12]{BE17} which has a slight modification from the previous lemma. It will be the last piece needed to derive the desired recursion relations. 

\begin{lemma}[see {\cite[Theorem~4.12]{BE17}}]\label{E-UPoleAnalysis}
For $i = 1, \ldots r$ and all $g,n \geq 0$,
\begin{equation}
\begin{aligned}
\frac{p_{0}(x)}{x^{\lfloor \alpha_{r-i+1}\rfloor}}& \left( \frac{\mathcal{E}^{i}_{g,n}(x;z_{[n]})}{dx^{i}} -  \delta_{n,0} \frac{S_{i,2g}}{x^i}  \right)= \sum_{j=1}^{n}  d_{z_{j}}\left(\frac{p_{0}(x_{j})}{x_{j}^{\lfloor \alpha_{r-i+1}\rfloor}}\frac{1}{x-x_{j}}\frac{\mathcal{U}^{i-1}_{g,n-1}(z_{j}; z_{[n]\setminus \{j\}})}{dx_{j}^{i-1}}\right)\\
& + \delta_{g,0}\delta_{n,1}(-1)^{i-1}d_{z_{1}}\left(\frac{1}{x-x_{1}}\left(\frac{p_{i-1}(x)}{x^{\lfloor \alpha_{r-i+1}\rfloor}} - \frac{p_{i-1}(x_{1})}{x_{1}^{\lfloor \alpha_{r-i+1}\rfloor}}\right)\right)\\
&+ \delta_{g,0}\delta_{n,0}\left(\frac{(-1)^{i}p_{i}(x)}{x^{\lfloor \alpha_{r-i+1}\rfloor}} \right) .
\end{aligned}
\end{equation}
Here, $d_{z_j}$ means taking the differential only with respect to the variable $z_j$.
\end{lemma}

\begin{proof}
The proof for the case $(g,n) = (0,1)$ is the same as in \cite[Theorem~4.12]{BE17}, while the other two unstable cases are a result of \cref{EiOneHalf}. For all other cases, the proof is completely analogous to \cite[Theorem~4.12]{BE17} and the result is the same. This is because the shifts only affect terms with $n=0$ which contribute nothing to these expressions for $2g-2+n \geq 0$.
\end{proof}

\begin{lemma}[see {\cite[Lemma~4.13]{BE17}}]\label{URecursion}
For $i = 1, \ldots, r$,
\begin{equation}
\begin{aligned}
&\frac{p_{0}(x)}{x^{\lfloor \alpha_{r-i+1}\rfloor}} \frac{\mathcal{U}^{i}_{g,n}(z;z_{[n]})}{dx^{i}}=  -\frac{p_{0}(x)}{x^{\lfloor \alpha_{r-i+1}\rfloor}dx}\frac{\mathcal{U}^{i-1}_{g-1,n+1}(z;z_{[n]},z)}{dx^{i-1}} + \frac{p_{1}(x)}{x^{\lfloor \alpha_{r-i+1}\rfloor}}\frac{\mathcal{U}^{i-1}_{g,n}(z;z_{[n]})}{dx^{i-1}}\\
& + \frac{p_{0}(x)}{x^{\lfloor \alpha_{r-i+1}\rfloor}}\sum_{N_{1}\sqcup N_{2}=z_{[n]}}\sum_{g_{1}+g_{2}=g}\frac{\mathcal{U}^{i-1}_{g_{1},|N_{1}|}(z; N_{1})}{dx^{i-1}}\frac{\mathcal{U}^{1}_{g_{2},|N_{2}|}(z; N_{2})}{dx}\\
& - \sum_{j=1}^{n}\left(\frac{p_{0}(x)}{x^{\lfloor \alpha_{r-i+1}\rfloor}}\frac{dx_{j}}{(x-x_{j})^{2}}\frac{\mathcal{U}^{i-1}_{g,n-1}(z; z_{[n]\setminus \{j\}})}{dx^{i-1}}\right. \left. - d_{z_{j}}\left(\frac{p_{0}(x_{j})}{x_{j}^{\lfloor \alpha_{r-i+1}\rfloor}}\frac{1}{x-x_{j}}\frac{\mathcal{U}^{i-1}_{g,n-1}(z_{j}; z_{[n]\setminus \{j\}})}{dx_{j}^{i-1}}\right)\right)\\
& - \frac{p_{0}(x)}{x^{\lfloor \alpha_{r-i+1}\rfloor + 1}}\sum_{h=\2}^{g}S_{1,2h}\frac{\mathcal{U}^{i-1}_{g-h,n}(z;z_{[n]})}{dx^{i-1}} + \delta_{g,0}\delta_{n,0}\frac{(-1)^{i}p_{i}(x)}{x^{\lfloor \alpha_{r-i+1}\rfloor}}\\
& + \delta_{n,0}S_{i,2g}\frac{p_{0}(x)}{x^{\lfloor \alpha_{r-i+1}\rfloor + i}} + \delta_{g,0}\delta_{n,1}(-1)^{i-1}d_{z_{1}}\left(\frac{1}{x-x_{1}}\left(\frac{p_{i-1}(x)}{x^{\lfloor \alpha_{r-i+1}\rfloor}} - \frac{p_{i-1}(x_{1})}{x_{1}^{\lfloor \alpha_{r-i+1}\rfloor}}\right)\right).
\end{aligned}
\label{eq:UU}
\end{equation}
\end{lemma}

\begin{proof}
We simply equate the expressions in \cref{E-UCombinatorics} and \cref{E-UPoleAnalysis} and rearrange the terms appropriately.
\end{proof}

\subsubsection{Integrating and summing over \texorpdfstring{$\hslash$}{hbar}}

Next step: we integrate \cref{URecursion}. We first define:

\begin{definition}
For $i = 0, \ldots, r-1$ and all $g, n \geq 0$,
\begin{equation}
\mathcal{G}^{i}_{g,n}(z;z_{[n]}) = \int_{\infty}^{z_{1}}\cdots\int_{\infty}^{z_{n}}\mathcal{U}^{i}_{g,n}(z;z_{[n]}^{\prime})
\end{equation}
where the integrals are with respect to the $z_{[n]}^{\prime}$ variables. We also define the following shorthand notation:
\begin{equation}
\mathcal{G}^{i}_{g,n}(z) =\mathcal{G}^{i}_{g,n}(z;z)\coloneqq \mathcal{G}^{i}_{g,n}(z;z,\ldots, z).
\end{equation}
When necessary, we will assume the integrals are regularized.
\end{definition}

 We now integrate \cref{URecursion} to get a recursion for the $\mathcal{G}^{i}_{g,n}$.

\begin{lemma}[see {\cite[Lemma~5.5]{BE17}}]\label{GRecursion}
For $i = 1, \ldots, r$, the integral $\int_{\infty}^{z}\cdots\int_{\infty}^{z}$ with respect to $z_{[n]}$ of \cref{URecursion} yields
\begin{equation}
\begin{aligned}
\frac{p_{0}(x)}{x^{\lfloor \alpha_{r-i+1}\rfloor}} & \frac{\mathcal{G}^{i}_{g,n}(z)}{dx^{i}}\\
= & -\frac{p_{0}(x)}{(n+1)x^{\lfloor \alpha_{r-i+1}\rfloor}}\frac{d}{dx}\left(\frac{\mathcal{G}^{i-1}_{g-1,n+1}(z^{\prime};z)}{dx^{\prime i-1}}\right)_{z^{\prime}=z} + \frac{p_{1}(x)}{x^{\lfloor \alpha_{r-i+1}\rfloor}}\frac{\mathcal{G}^{i-1}_{g,n}(z)}{dx^{i-1}}\\
& + \frac{p_{0}(x)}{x^{\lfloor \alpha_{r-i+1}\rfloor}}\sum_{m=0}^{n}\sum_{g_{1}+ g_{2}=g}\frac{n!}{m!(n-m)!}\frac{\mathcal{G}^{i-1}_{g_{1},m}(z)}{dx^{i-1}}\frac{\mathcal{G}^{1}_{g_{2},n-m}(z)}{dx}\\
& - n\frac{d}{dx^{\prime}}\left(\frac{p_{0}(x^{\prime})}{x^{\prime\,\lfloor \alpha_{r-i+1}\rfloor}}\frac{\mathcal{G}^{i-1}_{g,n-1}(z^{\prime};z)}{dx^{\prime\,i-1}}\right)_{x^{\prime}=x} - \frac{p_{0}(x)}{x^{\lfloor \alpha_{r-i+1}\rfloor + 1}}\sum_{h=\2}^{g}S_{1,2h}\frac{\mathcal{G}^{i-1}_{g-h,n}(z)}{dx^{i-1}}\\
& + \delta_{g,0}\delta_{n,0}\frac{(-1)^{i}p_{i}(x)}{x^{\lfloor \alpha_{r-i+1}\rfloor}} + \delta_{n,0}S_{i,2g}\frac{p_{0}(x)}{x^{\lfloor \alpha_{r-i+1}\rfloor + i}}\\
& + \delta_{g,0}\delta_{n,1}(-1)^{i-1}\frac{d}{dx}\left(\frac{p_{i-1}(x)}{x^{\lfloor \alpha_{r-i+1}\rfloor}}\right).
\end{aligned}
\end{equation}
\end{lemma}

\begin{proof}
For the most part, this is a very simple integration. However, there are a couple terms that we go into more detail on. On the one hand, integrating the first term on the right-hand-side of \eqref{eq:UU}, we get:
\begin{equation}
\begin{aligned}
-\frac{p_{0}(x)}{x^{\lfloor \alpha_{r-i+1}\rfloor}dx} & \int_{\infty}^{z}\cdots\int_{\infty}^{z}\frac{\mathcal{U}^{i-1}_{g-1,n+1}(z;z_{[n]},z)}{dx^{i-1}}\\
& = -\frac{p_{0}(x)}{rz^{r-1}x^{\lfloor \alpha_{r-i+1}\rfloor}}\frac{d}{dz_{n+1}}\left(\frac{\mathcal{G}^{i-1}_{g-1,n+1}(z;z_{[n+1]})}{dx^{i-1}}\right)_{z_{1},\ldots, z_{n+1}=z}\\
& = -\frac{p_{0}(x)}{(n+1)rz^{r-1}x^{\lfloor \alpha_{r-i+1}\rfloor}}\frac{d}{dz}\left(\frac{\mathcal{G}^{i-1}_{g-1,n+1}(z^{\prime};z)}{dx^{\prime i-1}}\right)_{z^{\prime}=z}\\
& = -\frac{p_{0}(x)}{(n+1)x^{\lfloor \alpha_{r-i+1}\rfloor}}\frac{d}{dx}\left(\frac{\mathcal{G}^{i-1}_{g-1,n+1}(z^{\prime};z)}{dx^{\prime i-1}}\right)_{z^{\prime}=z}
\end{aligned}
\end{equation}
where the second equality used the fact that the $\mathcal{G}^{i}_{g,n}$ are symmetric with respect to interchange of $z_{1}\ldots, z_{n}$. On the other hand, integrating the last term on the right-hand-side of \eqref{eq:UU} we get:
\begin{equation}
\begin{aligned}
\delta_{g,0}\delta_{n,1}(-1)^{i-1} & \int_{\infty}^{z}d_{z_{1}}\left(\frac{1}{x-x_{1}}\left(\frac{p_{i-1}(x)}{x^{\lfloor \alpha_{r-i+1}\rfloor}} - \frac{p_{i-1}(x_{1})}{x_{1}^{\lfloor \alpha_{r-i+1}\rfloor}}\right)\right)\\
& = \delta_{g,0}\delta_{n,1}(-1)^{i-1}\lim_{z_{1}\to z}\left(\frac{1}{x-x_{1}}\left(\frac{p_{i-1}(x)}{x^{\lfloor \alpha_{r-i+1}\rfloor}} - \frac{p_{i-1}(x_{1})}{x_{1}^{\lfloor \alpha_{r-i+1}\rfloor}}\right)\right)\\
& = \delta_{g,0}\delta_{n,1}(-1)^{i-1}\frac{d}{dx}\left(\frac{p_{i-1}(x)}{x^{\lfloor \alpha_{r-i+1}\rfloor}}\right).
\end{aligned}
\end{equation}
A similar computation can be done for the fourth term on the right-hand-side of \eqref{eq:UU}.
\end{proof}

Next, we define:

\begin{definition}
For $i = 0, \ldots, r-1$,
\begin{equation}
\xi^{i}(z) = (-1)^{i}\sum_{g,n}\frac{\hslash^{2g+n}}{n!}\frac{\mathcal{G}^{i}_{g,n}(z)}{dx^{i}}.
\end{equation}
\end{definition}

Multiplying \cref{GRecursion} by $(-1)^{i}\frac{\hslash^{2g+n}}{n!}$ and summing over all $g$ and $n$, we get a recursive relation for $\xi^{i}$. The result is:

\begin{lemma}[see {\cite[Lemma~5.8]{BE17}}]\label{XiRecursion}
After summing, \cref{GRecursion} becomes
\begin{equation}
\begin{aligned}
\frac{p_{0}(x)}{x^{\lfloor \alpha_{r-i+1}\rfloor}} \xi^{i}(z) & - \frac{p_{i}(x)}{x^{\lfloor \alpha_{r-i+1}\rfloor}}\\
= & - \frac{p_{1}(x)}{x^{\lfloor \alpha_{r-i+1}\rfloor}}\xi^{i-1}(z) + \frac{p_{0}(x)}{x^{\lfloor \alpha_{r-i+1}\rfloor}}\xi^{i-1}(z)\xi^{1}(z)\\
& + \frac{p_{0}(x)}{x^{\lfloor \alpha_{r-i+1}\rfloor + 1}}\xi^{i-1}(z)\sum_{g\geq\2}\hslash^{2g}S_{1,2g} + \frac{(-1)^{i}p_{0}(x)}{x^{\lfloor \alpha_{r-i+1}\rfloor + i}}\sum_{g\geq\2}\hslash^{2g}S_{i,2g}\\
& + \hslash\frac{d}{dx}\left(\frac{p_{0}(x)}{x^{\lfloor \alpha_{r-i+1}\rfloor}}\xi^{i-1}(z) - \frac{p_{i-1}(x)}{x^{\lfloor \alpha_{r-i+1}\rfloor}}\right).
\end{aligned}
\end{equation}
\end{lemma}

\subsubsection{A system of differential equations}

Recall the definition of the wave function:

\begin{definition}
Consider the shifted $(r,s)$-spectral curve of \cref{d:rsdefshift}, and let $\{ \omega_{g,n} \}_{g \in \frac{1}{2}\mathbb{N}, n \in \mathbb{N}^*}$ be the system of correlators constructed from shifted topological recursion. 
We define the \emph{wave function} as:
\begin{equation}
\psi(z) = \exp\left(\sum_{g,n}\frac{\hslash^{2g-2+n}}{n!}\int_{\infty}^{z}\cdots\int_{\infty}^{z}\left(\omega_{g,n+1}(z_{0},z_{[n]}) - \delta_{g,0}\delta_{n,1}\frac{dx_{0}dx_{1}}{(x_{0} - x_{1})^2}\right)\right),
\end{equation}
where the integrals of $\omega_{0,1}$ and $\omega_{0,2}$ need to be regularized. We also define
\begin{equation}\label{Psi_i}
\psi_{i}(z) = \frac{p_{0}(x)\xi^{i}(z) - p_{i}(x)}{x^{\lfloor\alpha_{r-i}\rfloor}}\psi(z)
\end{equation}
for $i = 1, \ldots, r$. Here, the functions $p_i(x)$ are defined by
\begin{equation}
\sum_{i=0}^{r}p_{i}(x)y^{r-i} = x^{r-s}y^{r} - 1,
\end{equation}
and the numbers $\alpha_i$ are\footnote{We refer the reader to \cite[eq.\ (2.3)]{BE17} for the definition of $\alpha_{i}$ in the case of a general admissible spectral curve.}
\begin{equation}\label{Alphas}
\alpha_{i} = \frac{i(r-s)}{r}.
\end{equation}
\end{definition}

\Cref{XiRecursion} can be used to produce a system of differential equations for the $\psi_{i}(z)$ which will in turn be used to construct the quantum curve. But before we can go ahead with this construction, we recall that
\begin{equation}\label{Psi_r}
\psi_{r}(z) = -\frac{p_{r}(x)}{x^{\lfloor\alpha_{0}\rfloor}}\psi(z)
\end{equation}
from \cite[Lemma~5.10]{BE17}, and we also require a similar expression for $\psi_{1}(z)$, which is found in the next lemma.

\begin{lemma}[see {\cite[Lemma~5.10]{BE17}}]\label{Psi_1}
Given the definition of $\psi_{i}(z)$ in equation \eqref{Psi_i}, we obtain
\begin{equation}
\psi_{1}(z) = \frac{p_{0}(x)}{x^{\lfloor \alpha_{r-1} \rfloor}}\left(\hslash\frac{d}{dx} - \frac{1}{x}\sum_{g\geq\2}\hslash^{2g}S_{1,2g}\right)\psi(z).
\end{equation}
\end{lemma}

\begin{proof}
We start with the expression in the last line of \cite[eq.\ (5.15)]{BE17}:
\begin{equation}
p_{0}(x)\hslash\frac{d}{dx}\ln \psi(z) = \frac{p_{0}(x)}{dx}\sum_{g,n}\frac{\hslash^{2g+n}}{n!}\int_{\infty}^{z}\cdots\int_{\infty}^{z}\left(\omega_{g,n+1}(z,z_{[n]}) - \delta_{g,0}\delta_{n,1}\frac{dxdx_{1}}{(x-x_{1})^2}\right)
\end{equation}
where the integrals act only on $z_{[n]}$, and use the fact that
\begin{equation}
\begin{aligned}
\omega_{g,n+1}(z,z_{[n]}) = -\mathcal{U}^{1}_{g,n} & (z;z_{[n]}) - \delta_{g,0}\delta_{n,0}\frac{p_{1}(x)}{p_{0}(x)}dx + \delta_{n,0}S_{1,2g}\frac{dx}{x}\\
& + \delta_{g,0}\delta_{n,1}\frac{dxdx_{1}}{(x - x_{1})^2}
\end{aligned}
\end{equation}
as seen in \cref{E1-U1}. Substituting this into the previous expression, we find that
\begin{equation}
p_{0}(x)\hslash\frac{d}{dx}\ln \psi(z) = p_{0}(x) \xi^{1}(z) - p_{1}(x) + \frac{p_{0}(x)}{x}\sum_{g\geq\2}\hslash^{2g}S_{1,2g}.
\end{equation}
Finally, using the definition of $\psi_{1}(z)$ and then rearranging leads to the statement of the lemma.
\end{proof}

We finally obtain a system of differential equations for the $\psi_i(z)$.

\begin{theorem}[see {\cite[Theorem~5.11]{BE17}}]\label{PsiDifferentialSystem}
For $i = 2, \ldots, r$, the following system of linear differential equations holds:
\begin{equation}\label{eq:system}
\begin{aligned}
\hslash\frac{d}{dx}\psi_{i-1}(z) = & \frac{x^{\lfloor \alpha_{r-i} \rfloor}}{x^{\lfloor \alpha_{r-i+1} \rfloor}}\psi_{i}(z) - \frac{p_{i-1}(x)x^{\lfloor \alpha_{r-1} \rfloor}}{p_{0}(x)x^{\lfloor \alpha_{r-i+1} \rfloor}}\psi_{1}(z)\\
& - \frac{p_{i-1}(x)}{x^{\lfloor \alpha_{r-i+1} \rfloor + 1}}\psi(z)\sum_{g\geq\2}\hslash^{2g}S_{1,2g} + \frac{(-1)^{i-1}p_{0}(x)}{x^{\lfloor \alpha_{r-i+1}\rfloor + i}}\psi(z)\sum_{g\geq\2}\hslash^{2g}S_{i,2g}.
\end{aligned}
\end{equation}
\end{theorem}

\begin{proof}
First, we multiply \cref{XiRecursion} by $\psi(z)$, which produces
\begin{equation}
\begin{aligned}{}
\frac{x^{\lfloor \alpha_{r-i}\rfloor}}{x^{\lfloor \alpha_{r-i+1}\rfloor}}\psi_{i} & (z) = \frac{x^{\lfloor \alpha_{r-1}\rfloor}}{x^{\lfloor \alpha_{r-i+1}\rfloor}}\xi^{i-1}(z) \psi_{1}(z)\\
& + \frac{p_{0}(x)}{x^{\lfloor \alpha_{r-i+1}\rfloor + 1}} \xi^{i-1}(z) \psi(z)\sum_{g\geq\2}\hslash^{2g}S_{1,2g} + \frac{(-1)^{i}p_{0}(x)}{x^{\lfloor \alpha_{r-i+1}\rfloor + i}}\psi(z)\sum_{g\geq\2}\hslash^{2g}S_{i,2g}\\
& + \hslash\psi(z)\frac{d}{dx}\left(\frac{p_{0}(x)}{x^{\lfloor \alpha_{r-i+1}\rfloor}}\xi^{i-1}(z) - \frac{p_{i-1}(x)}{x^{\lfloor \alpha_{r-i+1}\rfloor}}\right)
\end{aligned}
\end{equation}
for $i = 2, \ldots, r$. The last term can be written as
\begin{equation}
\hslash\frac{d}{dx}\psi_{i-1}(z) - \hslash\frac{p_{0}(x)\xi^{i-1}(z) - p_{i-1}(x)}{x^{\lfloor \alpha_{r-i+1}\rfloor}}\frac{d\psi(z)}{dx}
\end{equation}
and then the previous lemma implies that this is equal to
\begin{equation}
\begin{aligned}
\hslash\frac{d}{dx}\psi_{i-1}(z) - \frac{x^{\lfloor \alpha_{r-1}\rfloor}}{x^{\lfloor \alpha_{r-i+1}\rfloor}} & \left(\xi^{i-1}(z) - \frac{p_{i-1}(x)}{p_{0}(x)}\right)\psi_{1}(z)\\
& - \frac{p_{0}(x)\xi^{i-1}(z) - p_{i-1}(x)}{x^{\lfloor \alpha_{r-i+1}\rfloor + 1}}\psi(z)\sum_{g\geq\2}\hslash^{2g}S_{1,2g}.
\end{aligned}
\end{equation} 
Putting it altogether, we obtain the statement of the theorem.
\end{proof}

The main difference between \cref{PsiDifferentialSystem} and theorem 5.11 in \cite{BE17} is the appearance of the terms on the second line of \eqref{eq:system} which depend on the shifts $S_{i,2g}$.

\eqref{eq:system} can be simplified significantly by noticing that $p_{i}(x) = 0$ for $i = 1, \ldots, r-1$ and $p_{0}(x) = x^{r-s}$, $p_{r}(x) = 1$. Therefore, we have
\begin{equation}\label{PsiDifferentialSystem-Simplified}
\hslash\frac{d}{dx}\psi_{i-1}(z) = \frac{x^{\lfloor \alpha_{r-i} \rfloor}}{x^{\lfloor \alpha_{r-i+1} \rfloor}}\psi_{i}(z) + \frac{(-1)^{i-1}x^{r-s}}{x^{\lfloor \alpha_{r-i+1}\rfloor + i}}\psi(z)\sum_{g\geq\2}\hslash^{2g}S_{i,2g}
\end{equation}
for $i = 2, \ldots, r$.

\subsubsection{The quantum curve}

Finally, we can write the system of differential equations \eqref{PsiDifferentialSystem-Simplified} as a single $r$th-order differential equation for the wave function $\psi(z)$, which will be the desired quantum curve. 

In the following theorem, we define the shorthand notation
\begin{equation}
  D_{i} \coloneqq \hslash\frac{x^{\lfloor \alpha_{i} \rfloor}}{x^{\lfloor \alpha_{i-1} \rfloor}}\frac{d}{dx}
\end{equation}
for $i=1,\ldots,r$.

\begin{theorem}[see {\cite[Lemma~5.14]{BE17}}]\label{t:QC}
  The system of differential equations in \cref{PsiDifferentialSystem} is equivalent to the $r$th order differential equation
  \begin{equation}
    \left(D_{1}\cdots D_{r} + \sum_{g\geq\2}\sum_{i=1}^{r}(-1)^{i}\hslash^{2g}S_{i,2g}D_{1}\cdots D_{r-i}\frac{x^{r-s}}{x^{\lfloor \alpha_{r-i} \rfloor + i}} - 1\right)\psi(z) = 0
  \end{equation}
  for the shifted wave function $\psi$ constructed from shifted topological recursion on the shifted $(r,s)$-spectral curve. Each set of $s$-consistent shifts $\{S_{i,\ell} \}_{i \in [r], \ell \in \mathbb{N}^*}$ provides a different quantization of the $(r,s)$-algebraic curve.
\end{theorem}

\begin{proof}
  Let's rewrite \cref{PsiDifferentialSystem-Simplified} to get
  \begin{equation}\label{PsiRecursion}
    \psi_{i}(z) = D_{r-i+1} \psi_{i-1}(z) + \frac{(-1)^{i}x^{r-s}}{x^{\lfloor \alpha_{r-i}\rfloor + i}} \psi(z) \sum_{g\geq\2}\hslash^{2g}S_{i,2g}
  \end{equation}
  for $i = 2, \ldots, r$. In particular, we have
  \begin{equation}
    \begin{aligned}
      \psi_{2}(z) 
      &= 
      D_{r-1}\psi_{1}(z) + \frac{x^{r-s}}{x^{\lfloor \alpha_{r-2} \rfloor + 2}}\psi(z)\sum_{g\geq\2}\hslash^{2g}S_{2,2g}
      \\
      &=
      D_{r-1}\left(\frac{x^{r-s}}{x^{\lfloor \alpha_{r} \rfloor}}D_{r} - \frac{x^{r-s}}{x^{\lfloor \alpha_{r-1} \rfloor + 1}}\sum_{g\geq\2}\hslash^{2g}S_{1,2g}\right)\psi(z) + \frac{x^{r-s}}{x^{\lfloor \alpha_{r-2} \rfloor + 2}}\psi(z) \sum_{g\geq\2}\hslash^{2g}S_{2,2g}
      \\
      &=
      \left(D_{r-1}\frac{x^{r-s}}{x^{\lfloor \alpha_{r} \rfloor}}D_{r} - D_{r-1}\frac{x^{r-s}}{x^{\lfloor \alpha_{r-1} \rfloor + 1}}\sum_{g\geq\2}\hslash^{2g}S_{1,2g} + \frac{x^{r-s}}{x^{\lfloor \alpha_{r-2} \rfloor + 2}}\sum_{g\geq\2}\hslash^{2g}S_{2,2g}\right)\psi(z)
    \end{aligned}
  \end{equation}
  where the second equality used \cref{Psi_1}. Iterating until $\psi_{r}(x)$ using equation \eqref{PsiRecursion}, we find that
  \begin{equation}
    \psi_{r}(z) 
    =
    \left(D_{1}\cdots D_{r-1}\frac{x^{r-s}}{x^{\lfloor \alpha_{r} \rfloor}}D_{r} + \sum_{g\geq\2}\sum_{i=1}^{r}(-1)^{i}\hslash^{2g}S_{i,2g}D_{1}\cdots D_{r-i}\frac{x^{r-s}}{x^{\lfloor \alpha_{r-i} \rfloor + i}}\right)\psi(z) \,,
  \end{equation}
  and then the fact that $\psi_{r}(z) = \psi(z)$ as implied in \cref{Psi_r} and $ \alpha_r = r-s$ leads to the statement of the theorem.
\end{proof}

\noindent
If all constants $S_{i,2g}$ are zero as in the unshifted case, then we only get one specific quantization of the $(r,s)$-curve -- and it is a rather non-trivial one. This is the only possibility if $ r = -1 \pmod{s}$, because we require the shifts to be  $s$-consistent, see \cref{d:consistent}, so that shifted topological recursion produces symmetric correlators. If $ r = 1 \pmod{s}$, other shifts produce other quantizations of the spectral curve. In particular, as we now show, for the cases $s=1$ and $s=r-1$, we obtain all possible orderings of the operators $\hat{x} = x$ and $\hat{y} = \hslash \frac{d}{dx}$ as particular choices of the shifts.

\begin{example}[$s=1$]
  For this curve, we have $\alpha_{i} = \frac{i(r-1)}{r}$ which satisfies
  \begin{equation}
    i-1 \leq \alpha_{i} < i
  \end{equation}
  for $i = 1, \ldots, r$, and therefore $\lfloor \alpha_{i} \rfloor = i-1$ while $\lfloor \alpha_{0} \rfloor = 0$. Hence, the quantum curve is
  \begin{equation}
    \begin{aligned}
      & 
      \left(\hslash^{r}\left(\frac{d}{dx}x\right)^{r-1}\frac{d}{dx} + \sum_{g\geq\2}\left(\sum_{i=1}^{r-1}(-1)^{i}\hslash^{r-i+2g}S_{i,2g}\left(\frac{d}{dx}x\right)^{r-i-1}\frac{d}{dx} \right.\right.
      \\
      & 
      \hspace{7cm} \left.\left. + (-1)^{r}\hslash^{2g}S_{r,2g}\frac{1}{x}\right) - 1\right)\psi 
      =
      0\,.
    \end{aligned}
  \end{equation}
  If we now take $S_{i,2g} = 0$ for all $i, g$ except when $i = 2g \leq r-1$, then we obtain the following quantization of the $(r, 1)$-spectral curve
  \begin{equation}
    \left(\hslash^{r}\left(\frac{d}{dx}x\right)^{r-1}\frac{d}{dx} + \hslash^{r}\sum_{i=1}^{r-1}(-1)^{i}S_{i,i}\left(\frac{d}{dx}x\right)^{r-i-1}\frac{d}{dx} - 1\right)\psi 
    =
    0 \,,
  \end{equation}
  which again produces all reorderings by making different choices of $S_{1,1},\ldots, S_{r-1,r-1}$.
\end{example}

\begin{example}[$s=r-1$]
  In this case, we have $\alpha_{i} = \frac{i}{r}$, so $\lfloor \alpha_{i} \rfloor = 0$ for $i = 0, \ldots, r-1$ and $\lfloor \alpha_{r} \rfloor = 1$. The only $s$-consistent shift is $ S^\hslash_1$ (unless $ r = 2$, in which case we can use the previous example). Hence, the quantum curve is
  \begin{equation}
    \left(\hslash^{r}\frac{d^{r-1}}{dx^{r-1}}x\frac{d}{dx} - \sum_{g\geq\2} \hslash^{r-1+2g}S_{1,2g}\frac{d^{r-1}}{dx^{r-1}}  - 1\right)\psi 
    =
    0 \,.
  \end{equation}
  If we assume $S_{1,2g} = 0$ for all $g$ except $ 2g = 1$, then we get a quantization of the $(r,r-1)$-spectral curve of the form
  \begin{equation}
    \left(\hslash^{r}\frac{d^{r-1}}{dx^{r-1}}x\frac{d}{dx} - \hslash^{r}S_{1,1}\frac{d^{r-1}}{dx^{r-1}}  - 1\right)\psi 
    =
    0\,,
  \end{equation}
  which leads to all possible reorderings of the curve for different choices of $S_{1,1}$. Indeed, if we choose $S_{1,1} = m$ for any $m \in \{-1, \ldots, r-1\}$, we will get the quantum curve
  \begin{equation}
    \left(\hslash^{r}\frac{d^{r-m-1}}{dx^{r-m-1}}x\frac{d^{m+1}}{dx^{m+1}} - 1\right)\psi = 0.
  \end{equation}
\end{example}

However, in the other cases, we do not get all the orderings.

\begin{example}[Other $ r = 1 \pmod{s}$]
  For $ s$-consistency, we can again only allow non-trivial $ S_1^\hslash$, so the spectral curve is
  \begin{equation}
    \left(D_{1}\cdots D_{r} - \sum_{g\geq\2} \hslash^{2g}S_{1,2g}D_{1}\cdots D_{r-1}  - 1\right)\psi = 0 \,,
  \end{equation}
  and since $ D_r = \hslash x \frac{d}{dx}$, this can be rewritten as
  \begin{equation}
    \left(  D_{1}\cdots D_{r-1} \Big( \hslash x \frac{d}{dx} - \sum_{g\geq\2} \hslash^{2g}S_{1,2g} \Big)  - 1\right)\psi = 0 \,.
  \end{equation}
  But now, as we assume $ 1 < s < r-1$, at least one more of the $ D_j$ must equal $ \hslash x \frac{d}{dx}$, and the commutator of this $x$ with any $ \hslash \frac{d}{dx}$ cannot be expressed through $ S_1^\hslash $ any more. So we find that in this case, there are reorderings of the normal-ordered quantization that are not covered by $s$-consistent shifts.
\end{example}

\section{Determinantal formulas and non-perturbative loop equations}
\label{s:diffsyst}

In this section we turn the tables around and start directly from the general quantum spectral curve of $(r,s)$-systems, in the differential system form that previously appeared as an intermediate step. We will prove that by defining generating functions with formal genus-counting parameter $\hslash$ for the topological recursion invariants of $(r,s)$-systems, these generating functions can be identified with the non-perturbative amplitudes associated to the differential system obtained by analytic continuation of the quantum curve under consideration. As such, we will be able to express them via determinantal formulas.\par
We will set up the WKB analysis of the quantum curve, introduce the corresponding non-perturbative invariants in the form of well-known determinantal formulas, albeit as $\hslash\rightarrow0$ asymptotics directly, as well as the collection of non-perturbative loop equations we require them to satisfy. A such, the coefficients of these semi-classical expansions will be identified with the topological recursion invariants of interest.

\subsection{Rational formal \texorpdfstring{$\hslash$}{hbar}-connections and their WKB analysis}

In this section, we will consider the setup of $\hslash $-connections for our problem. We will give a definition here that fits our needs.

\begin{definition}
  A \emph{rational formal $\hslash$-connection} is an $\hslash$-connection 
  \begin{equation}
  \label{QuantumCurve}
    \nabla_\hslash \coloneq \hslash d - \Phi_\hslash \,,
  \end{equation}
  on the trivial principal bundle $\mc{E} \coloneq \P^1\times \GL_r(\C)$ over the Riemann sphere with general linear structure group $\GL_r(\C)$, where $ \Phi_\hslash $ is a power series in $ \hslash $ with rational  one-form coefficients, satisfying the deformed Leibniz rule
  \begin{equation}
    \nabla_\hslash(f\sigma) = f \nabla_\hslash (\sigma) + \hslash ( df)\sigma
  \end{equation}
  for all possibly $\hslash$-formal rational functions $f\in \C(x)$, and local sections $\sigma$ of $ \mc{E}$.\par
  The \emph{Higgs field} of a rational $ \hslash$-connection is the leading order
  \begin{equation}
    \phi \coloneq \Phi_0
  \end{equation} 
  and the corresponding \emph{spectral curve} is
  \begin{equation}\label{eq:Sigma}
    \Sigma = \{ E (x, \omega) = \det ( \omega \Id - \Phi_\hslash (x) ) = 0 \} \subset T^*\P^1\,.
  \end{equation}
\end{definition}

The Higgs field fits into the following short exact sequence of sheaves
\begin{equation}
    0 \longrightarrow  \mathcal N \hookrightarrow \mathcal E \xlongrightarrow{[\phi,\bullet]} \mathcal E\otimes\Omega^1 \longrightarrow 0\,,
\end{equation}
where we have denoted the commutant of $\phi$ by $\mathcal N\coloneq \Ker [\phi,\bullet]\subset \mathcal E$. This short exact sequence is one of vector bundles away from the locus of branch points where the rank of $\mc{N}$ jumps.

We are mainly interested in the following example.

\begin{example}
  \label{OurQC}
  Let us consider the quantum curve equation of the previous section, in its differential system form of \cref{PsiDifferentialSystem} with unknown vector function $(\psi,\psi_1,\dots,\psi_{r-1})^T$. 
  Explicitly,
  \begin{equation}
  \label{ConnectionPotential}
    \Phi_\hslash(x)
    \coloneq
    \sum_{k=0}^\infty \hslash^k\Phi_k(x) 
    \coloneq  
      \begin{pmatrix}
        (-1)^0\frac{S_1^\hslash}{x} \cdot \frac{x^{\lfloor \alpha_{r}\rfloor}}{x^{\lfloor \alpha_r\rfloor}}  & \frac{x^{\lfloor \alpha_{r-1}\rfloor}}{x^{\lfloor \alpha_r\rfloor}} &  \cdots & 0
        \\
        \vdots & 0 & \ddots & \vdots
        \\
        (-1)^{r-2}\frac{S_{r-1}^\hslash}{x^{r-1}}\cdot \frac{x^{\lfloor \alpha_{r}\rfloor}}{x^{\lfloor \alpha_2\rfloor}} &\vdots & \ddots & \frac{x^{\lfloor \alpha_{1}\rfloor}}{x^{\lfloor \alpha_2\rfloor}}
        \\
        (-1)^{r-1}\frac{S_{r}^\hslash}{x^{r}}\cdot \frac{x^{\lfloor \alpha_{r}\rfloor}}{x^{\lfloor \alpha_1\rfloor}}+\frac{1}{x^{\lfloor \alpha_{1}\rfloor}} & 0 & \cdots &  0
      \end{pmatrix}  
      dx\,,
  \end{equation}
  where we have introduced the $\hslash$-series $S_k^\hslash \coloneq \sum_{g\geqslant \frac12}^\infty \hslash^{2g} S_{k,2g}$, for each $k\in\{1,\dots,r\}$. Recall from \cref{Alphas} also the rational values $\alpha_k=\frac kr(r-s)$ of which the floors appear in each non-trivial coefficient of the matrix. 
  
  Since all $S_k^\hslash$ are of order at least 1 in $\hslash$, we get the Higgs field
  \begin{equation}
  \label{Higgs}
    \phi(x) 
    =
    \Phi_0(x)
    =
    \begin{pmatrix}
      0 & \frac{x^{\lfloor \alpha_{r-1}\rfloor}}{x^{\lfloor \alpha_r\rfloor}} &  \cdots & 0
      \\
      \vdots & \ddots & \ddots & \vdots 
      \\
      0& \cdots & 0 & \frac{x^{\lfloor \alpha_{1}\rfloor}}{x^{\lfloor \alpha_2\rfloor}} 
      \\
      \frac{1}{x^{\lfloor \alpha_{1}\rfloor}} & 0& \cdots &  0
    \end{pmatrix}  
    dx \,,
  \end{equation}
as well as the spectral curve
  \begin{equation}
    E(x, \omega ) = \omega^r - x^{s-r} dx^r \,.
  \end{equation}
  Interpreting $ \omega $ as $ \omega_{0,1} = y \, dx $, we find that $ E (x, y\, dx) = (y^r - x^{s-r}) dx^r$.
\end{example}

\begin{lemma}
  In the setting of \cref{OurQC}, after restricting to $ \C^\times \subseteq \P^1 $ and pulling back along $ x \colon \C^\times \to \C^\times \colon z \mapsto z^r$, $ \phi $ can be diagonalized as $\phi(x(z))= V(z) Y(z) V(z)^{-1}$, where we have introduced
  \begin{align}
  \label{AbelianizationY}
    Y(z) 
    &\coloneq 
    \begin{pmatrix}
      \theta^0 & & 0 
      \\
      &\ddots&
      \\
      0 && \theta^{r-1}
    \end{pmatrix} r z^{s-1}dz \,,
    \\
    \intertext{and}
  \label{AbelianizationV}
    V(z) 
    & \coloneq 
    \frac{\theta z^{\frac{(r-s)(r+1)}{2}}}{\underset{1\leq a<b\leq r}\prod(\theta^b-\theta^a)^{\frac1r}}
    \begin{pmatrix}
      z^{r\lfloor \alpha_{1}\rfloor} &\cdots & 0 
      \\
      \vdots & \ddots & \vdots 
      \\
      0 & \cdots & z^{r\lfloor \alpha_{r}\rfloor}
    \end{pmatrix}
    \begin{pmatrix} 
      \frac{\theta^0}{z^{r-s}} &  \cdots & \frac{\theta^{r-1}}{z^{r-s}} 
      \\
      \vdots & & \vdots 
      \\
      (\frac{\theta^0}{z^{r-s}})^{r}  & \cdots & (\frac{\theta^{r-1}}{ z^{r-s}})^{r}
    \end{pmatrix} \,,
  \end{align}
  where $ Y(z)$ is a diagonal-matrix valued one form of eigenvalues $Y$ and $V(z)$ is the corresponding invertible Vandermonde matrix of eigenvectors, with a fixed primitive root of unity $\theta^r=1$.\par
  The inverse of $ V(z)$ is
  \begin{equation}
    V(z)^{-1} 
    =
    \frac{\underset{1\leq a<b\leq r}\prod(\theta^b-\theta^a)^{\frac1r}}{r \theta z^{\frac{(r-s)(r+1)}{2}}} 
    \begin{pmatrix} 
      \frac{z^{r-s}}{\theta^0}  & \cdots &  (\frac{z^{r-s}}{\theta^0})^{r}
      \\
      \vdots & &\vdots 
      \\
      \frac{z^{r-s}}{\theta^{r-1}} & \cdots & (\frac{z^{r-s}}{\theta^{r-1}})^{r}
    \end{pmatrix}
    \begin{pmatrix}
      z^{-r\lfloor \alpha_{1}\rfloor} &\cdots & 0
      \\
      \vdots & \ddots & \vdots 
      \\
      0 & \cdots & z^{-r\lfloor \alpha_{r}\rfloor}
    \end{pmatrix}\,. 
  \end{equation}
  These matrices have monodromy 
  \begin{equation}
    \label{DeckAction}
    Y(\theta z) = \tau^{-1} Y(z) \tau \,, \quad \text{and}\quad V(\theta z) = V(z)\tau \,,  
  \end{equation}
  with the permutation matrix defined by
  \begin{equation}\label{PermutationMatrix}
    \tau 
    \coloneq 
    \begin{pmatrix}
      0 & 1 & \cdots & 0
      \\ 
      \vdots& \ddots & \ddots&\vdots
      \\
      0& \cdots & 0& 1
      \\
      1 & 0 & \cdots & 0
    \end{pmatrix}^s \,.
  \end{equation}
  The spectral covering $x$ extends to the origin and infinity where it fully ramifies as $ x \colon \Sigma = \P^1 \to \P^1 \colon z \mapsto z^r$.
\end{lemma}

\begin{proof}
  Explicit standard calculations.
\end{proof}

\begin{remark}\label{WeylIsDeck}
  As usual, the diagonalization is not unique: we can reorder the eigenspaces by the Weyl group action by permutation. Similarly, we have a choice, given a point in the base $ \P^1$, of ordering of the sheets of the spectral curve by deck transformations. The monodromy relation \eqref{DeckAction} equates these two groups, so we do not actually introduce more freedom into our system by passing to the spectral curve.
\end{remark}

We will formally construct an all-order WKB-type solution to the equation $ \nabla_\hslash \Psi_\hslash = 0$. This is an extension of the method used for $r=2$ in \cite{AKT02}.\footnote{We would like to thank J. Hurtubise for explaining this extended method to us.}

\begin{lemma}\label{FormalGauge}
  There is a unique sequence of rational matrix-valued $ u_\ell (z) $ with trivial diagonal coefficients, for $ z \in \C^\times $, such that defining
  \begin{equation}
    \widehat U_\hslash(z) \coloneq V(z) \prod_{\ell\geq1}^{\rightarrow} \exp\Big( \hslash^\ell u_\ell(z)\Big) = V(z)\big(\Id +\mathcal O(\hslash)\big) \,,
  \end{equation}
  the expression  
  \begin{equation}
    \widehat Y_\hslash(z) 
    \coloneq \sum_{\ell\geq0}\hslash^\ell Y_\ell(z) \coloneq
    \widehat U_\hslash(z) ^{-1}\Phi_\hslash(z) \widehat U_\hslash(z) - \hslash \widehat U_\hslash(z)^{-1}d\widehat U_\hslash(z)  \,, 
  \end{equation}
  with $Y_0 \coloneq Y$, is diagonal and rational at each order in its $\hslash$-expansion.\par
  This is equivariant under deck transformations:
  \begin{equation}
    \hat{Y}_\hslash (\theta z) = \tau^{-1} \hat{Y}_\hslash (z) \tau \quad \text{and} \quad \hat{U}_\hslash (\theta z) = \hat{U}_\hslash (z) \tau\,.
  \end{equation}
\end{lemma}

\begin{proof}
  We construct the solution step by step, using the ordered partial products given by $U_\hslash^{(L)}(z) \coloneq V(z) \overset{\rightarrow}{\prod}{}_{\ell=1}^L \exp( \hslash^\ell u_\ell(z))$. The condition that $ \hat{Y}_\hslash (z)$ be rational is equivalent to defining the $u_L$'s by imposing, recursively on $L\geq0$, the triviality of the off-diagonal coefficients of
  \begin{equation}
  \label{Abelianization}
    Y_{L+1}(z) = \Phi_{L+1}^{(L)}(z) +  \big[Y_0(z),u_{L+1}(z)\big],
  \end{equation}
  where $du_0 \coloneq V_0^{-1}dV_0$ by convention, and we have introduced the intermediate connection potentials
  \begin{equation}
    \Phi_\hslash^{(L)} 
    \coloneq 
    \sum_{\ell\geq0}\hslash^\ell\Phi_\ell^{(L)}
    \coloneq
    (U_\hslash^{(L)}) ^{-1}\Phi_\hslash U_\hslash^{(L)} - \hslash  (U_\hslash^{(L)})^{-1}d U_\hslash^{(L)}
    =
    Y_\hslash^{(L)} + \mathcal O(\hslash^{L+1})\,,
  \end{equation}
  with $ Y_\hslash^{(L)} \coloneq \sum_{\ell=0}^L \hslash^\ell Y_\ell$.\par
  This procedure requires $[Y_0(z),\mathord{\cdot}]$ to be invertible on matrices with trivial diagonal coefficients, and is therefore valid everywhere away from zero and infinity, over which it has a non-trivial kernel.\par 
  Rationality of the $u_\ell (z)$ and the $ Y_\ell (z)$ follows by induction: $ du_0 $ is rational because $ V_0$ is, $ \Phi^{(L)}_{L+1}$ is an algebraic combination of the rational $ \Phi_\hslash $ and $ u_j $ for $ j \leq L$, and the inverse of $[Y_0(z),\mathord{\cdot}]$ preserves rationality.\par
  Equivariance follows from the initial step $du_0 \coloneq V_0^{-1}dV_0$ together with the explicit form of (\ref{Abelianization}).
\end{proof}

\begin{corollary}\label{cor:AbelianConnection}
  The formal connection $\widehat\nabla_\hslash$ defined by
  \begin{equation}
  \label{AbelianConnection}
    \widehat \nabla_\hslash 
    \coloneq
    \widehat U_\hslash ^{-1}\nabla_\hslash \widehat U_\hslash
    =
    \hslash d - \widehat Y_\hslash \,,
  \end{equation}
  is abelian, with formal gauge transformation, and diagonal connection potential respectively satisfying 
  \begin{equation}
  \label{DeckActionAb}
    \widehat U_\hslash(\theta\, z)
    = \widehat U_\hslash(z) \tau \,, 
    \quad  \text{and} \quad  
    \widehat Y_\hslash(\theta\,z)
    =
    \tau^{-1} \widehat Y_\hslash(z) \tau \,.  
  \end{equation}
\end{corollary}
In turn, each $Y_\ell $ is rational on the spectral curve with a unique singularity at the origin, with a pole order that is easily seen to grow as a function of $ \ell$. In particular, this means that the divergent integrals in
\begin{equation}
    \label{DivergentIntegral}
    \int_0^z \hat{Y}_\hslash
    =
    \sum_{\ell=0}^\infty \hslash^\ell\int_0^z Y_\ell
\end{equation}
can be regularized \textit{term by term} in the $\hslash$-expansion by adding and subtracting a finite number of counter-terms. This is a fairly standard procedure in topological recursion, cf. e.g. \cite{EO07,EMS11}. More details can be found in \cite[Section~3.2]{Bel24}. For us, the exact method of regularization is not important. We denote this regularized integral as
\begin{equation}
    \widehat{\mathcal J}_\hslash(z)
    \coloneq
    \slashed{\int}_0^z \widehat Y_\hslash
    =
    \sum_{\ell\geq 0} \hslash^\ell\slashed {\int}_0^z Y_\ell.
\end{equation}

\begin{corollary}
  \label{ExistenceWKB}
  The expression 
  \begin{equation}
  \label{PsiAbel}
    \Psi_\hslash(z)
    \coloneq
    \widehat U_\hslash(z) \operatorname e^{\frac1\hslash \widehat{\mathcal J}_\hslash(z)}
    =
    \widehat U_\hslash(z) \exp \Big(\frac1\hslash \slashed \int_0^z \widehat Y_{\hslash}\Big)
  \end{equation}
  is a formal solution to $ \nabla_\hslash \Psi_\hslash (z) = 0$. It can alternatively be expressed in WKB form as
  \begin{equation} \label{PsiWKB}
    \Psi_\hslash(z) = V(z) \hat{\Psi}_\hslash (z) \exp\Big(\frac1\hslash\int_0^z Y_0\Big)\,,
  \end{equation}
  where $ \hat{\Psi}_\hslash (z)$ takes the form of a product of a formal $\hslash$-series whose coefficients are rational matrix-valued functions of $z$ by the exponential of a formal $\hslash$-series each coefficient of which is the regularized integral of a rational one-form. This solution is equivariant:
  \begin{equation}
  \label{Equivariance}
    \Psi_\hslash(\theta z) = \Psi_\hslash(z)\tau \,.
  \end{equation}
\end{corollary}

\begin{proof}
  The first statement is a direct consequence of \cref{cor:AbelianConnection}. The WKB form can be found by equating \eqref{PsiAbel} and \eqref{PsiWKB}, and solving for $ \hat{\Psi}_\hslash (z)$. This yields the claimed rational coefficients, as well as the integrals of rational one-forms by \cref{FormalGauge}.
\end{proof}

\begin{remark}\label{FreedomOfFundamentalSolution}
  Given a fundamental solution $ \Psi_\hslash$ of $ \nabla_\hslash$ as above, $ \Psi_\hslash C$ is still a fundamental solution for any invertible constant matrix $ C$. However, the equivariance of \eqref{Equivariance} reduces this freedom: if we also want $\Psi_\hslash(\theta z)C = \Psi_\hslash(z)C\tau $, we need that $ \tau C = C \tau$.
\end{remark}

\subsection{Determinantal amplitudes and loop equations}

In this subsection, we will explain how to use the formal WKB solutions constructed above through the associated formal solutions
\begin{equation}
  \label{AdjointSolution}
  M_\hslash(z, E) \coloneq \Psi_\hslash(z) E \Psi_\hslash(z)^{-1} \,,
\end{equation}
of the adjoint differential system 
\begin{equation}
  \label{AdjointSystem}
  \hslash d M=[\Phi_\hslash,M],
\end{equation} 
for any $E \in \End (\C^r) $ encoding choices of `initial conditions'.\footnote{Strictly speaking, we do not impose that there is an `initial point' $ z_0$ such that $ \Psi_\hslash (z_0) = \Id$, so $E$ is not quite an initial condition, but it fulfils the same role.}

By construction, $M$ is also equivariant order by order:
\begin{equation}
  M_\hslash (\theta z, \tau^{-1} E \tau) = M_\hslash (z, E ) \,.
\end{equation}
In particular, if $ e_a$ is the $ a$-th diagonal basis matrix,
\begin{equation}
  M_\hslash (z, e_a ) = M_\hslash (\theta z, \tau^{-1} e_a \tau) = M_\hslash (\theta z, e_{a+s} )\,,
\end{equation}
so if $ \tilde{a} $ is the unique solution modulo $ r$ of $ a + \tilde{a}s = r$, then we get
\begin{equation}
\label{MEquiv}
  M_\hslash (z, e_a ) = M_\hslash (\theta^{\tilde{a}} z, e_r )\,.
\end{equation}
This is an implementation of the relation between the Weyl group and the group of deck transformations of \cref{WeylIsDeck}.

Since the only possibly non-rational terms featuring in the WKB solution appear as diagonal multiplicative factors from the right, we get the following proposition.

\begin{proposition}[{\cite[Remark 3.2]{BEM17}}]
  \label{ExpansionOfM}
  The construction via the WKB-type solution \cref{PsiWKB} is equivalent to the expansion of $M_\hslash (z, e_a )$ in powers of $ \hslash$ taking the following shape:
  \begin{equation}
    \label{MExpansion}
    M_\hslash (z, e_a) = V(z) e_a V(z)^{-1} + \sum_{k=1}^\infty M^{(k)}( \theta^{\tilde{a}} z) \hslash^k \,,
  \end{equation} 
  where the $ M^{(k)}$'s are rational functions of $z$.
\end{proposition}

We will use this adjoint system to define non-perturbative connected amplitudes via formulas that are independent of the fact that $M_\hbar$ is $\hbar$-formal, objects that therefore satisfy non-perturbative loop equations~\cite{BEM18}. Under certain assumptions, namely the topological type property of~\cite{BEM17} (following \cite{BBE15} for the $ q$-Gelfand--Dickey hierarchy), the $\hslash$-expansions of these amplitudes yield correlators computed by topological recursion, and we will see exactly when such a topological type property is verified in the setup of the present work.

Recall that we use shorthand notation $ x_j = x(z_j)$ for $ z_j \in \Sigma$.

\begin{definition}
  The \emph{Cauchy kernel} associated to the WKB solution \eqref{PsiWKB} is
  \begin{equation}
    \label{CauchyKernel}
    K_\hslash(z_1,z_2) \coloneq \sqrt{dx_1}\frac{\Psi_\hslash(z_1)^{-1}\Psi_\hslash(z_2)}{x_2-x_1}\sqrt{dx_2} \,.
  \end{equation}
\end{definition}
It has simple poles at each of the pre-images of the diagonal in $\C^\times \times \C^\times$ by the map $x \colon \Sigma \rightarrow \C^\times$, as well as singularities at the origin in each of its variables. We will use matrix elements of the kernel (\ref{CauchyKernel}) to build the non-perturbative amplitudes. We will do so in a symmetric way involving choices of initial conditions, and such that all diagonal singular contributions cancel each other when the corresponding matrices are diagonal. The singularities at the origin will however remain at each order in $\hslash$, to be expected in the context of topological recursion. 

In the setup of \cref{OurQC}, the expansion of $ K_\hslash$ near inverse images of the diagonal is given by
\begin{equation}
\label{KernelSingularity}
    K_\hslash(z_1,z_2) = \frac{r z_1^{r-1}  dz_1}{z_2^r-z_1^r} - \Psi_\hslash(z_1)^{-1}\Phi_\hslash\big(x(z_1)\big) \Psi_\hslash(z_1) \tau^k + \mc{O} (z_2-\theta^k z_1) \,,
\end{equation}
valid in the regimes where $z_2\rightarrow\theta^kz_1$ for some integer $k\in\{1,\dots,r\}$, and where we have used the equivariance \eqref{Equivariance}.

On the diagonal itself, where it is not defined, we prescribe the value of the kernel to be given by the next-to-singular term appearing in the asymptotics (\ref{KernelSingularity}) for $k=0$, that is
\begin{equation}
    \label{DiagonalKernel}
    K_\hslash(z,z) \coloneq - \frac1\hslash\Psi_\hslash(z)^{-1}\Phi_\hslash\big(x(z)\big)\Psi_\hslash(z) \,. 
\end{equation}


\begin{definition}\label{DefNPAmplitudes}
  Define for every $ n \geq 2$, the $n^{th}$ \emph{non-perturbative connected amplitude} as functions of $z_1,\dots,z_n\in\Sigma$, and any choice of matrices $E_1,\dots,E_n$, to be given by the expressions
  \begin{equation}
  \label{NPamplitudes}
    W_n(\overset{E_1}{z_1},\dots,\overset{E_n}{z_n})
    \coloneq
    (-1)^{n-1}\sum_{\sigma\in\mathfrak S_n'}  \frac{\Tr\underset{1\leq i \leq n}{\overset\rightarrow \prod} M_\hslash(z_{\sigma^i}, E_{\sigma^i})}{(x_1 -x_{\sigma^2})\cdots(x_{\sigma^{n}}-x_1)} \prod_{i=1}^n dx_i \,,
  \end{equation}
  involving the set $\mathfrak S_n'$ of all permutations $\sigma=(\sigma^1=1,\sigma^2,\dots,\sigma^n)$ of $\{1,\dots,n\}$ consisting of a single cycle, and the non-commutative products that are always computed in reading order from left to right, as indicated by the arrow.\par
  By the cyclic property of the trace, they can be expressed in terms of the Cauchy kernel, which we use to extend the definitions to points with same base-point projections. For $n=1$, we use this to define
  \begin{equation}
    W_1(\overset Ez)
    \coloneq
    \frac1\hslash\Tr \big( M_\hslash(z, E)\Phi_\hslash(z^r)\big)
    =
    - \Tr \big( K_\hslash(z,z)E\big) \,.
  \end{equation}
  Non-connected amplitudes are defined via the cumulant formula
  \begin{equation}
    \label{DisNPamplitude}
    \widehat W_n(J) \coloneq \underset{\mu\in\text{part}(J)}{\sum} \prod_{i=1}^{\text{length}(\mu)} W_{|\mu_i|}(\mu_i) \,,
  \end{equation}
  summing over set partitions of any given  $J \coloneq \big\{\overset{E_1}{z_1},\dots,\overset{E_n}{z_n}\big\}$.
\end{definition} 

\begin{remark}
  The connected amplitudes are non-perturbative in the sense that the definition work both for a genuine solution $ M_\hslash $ of the differential system or a formal one.\par
  Both cases have singularities at pre-images of the diagonal, but in the formal case, the $\hslash$-expansions coming from considering WKB solutions imply additional singularities at the origin.\par
\end{remark}

These amplitudes (both connected and disconnected) are still equivariant in the same way as $ M_\hslash$ is:
\begin{equation}\label{WEquivariant}
  W_{n+1} (\overset{\tau^{-1} E_0 \tau}{\theta z_0}, J ) = W_{n+1} (\overset{E_0}{z_0}, J )\,.
\end{equation}

Let us now introduce the non-perturbative loop equations, identities between the various $ W_n$'s that imply hierarchies of necessary relations at each order in $\hslash$. This construction is formally similar to the twist-field construction of \cref{s:rsairystruct}, and again makes use of the Casimir elements of the Lie algebra $\mathfrak{gl}_r$, cf. \cref{WrGenerators}. These identities encode the invariance of certain combinations of non-connected amplitudes under the combined action of parallel transport by $\nabla_\hslash$ around the origin and spectral curve deck transformations. We will interpret the corresponding constrained singular profiles in terms of the data of a representation of the W-algebra $\mc{W}(\mathfrak{gl}_r)$. 

Denoting by $\{\mathbf e_{i,j}\}_{i,j=1}^r$ the standard vector-space basis of $r\times r$ matrices, we consider the algebraic generators $C^{(1)},\dots, C^{(r)}$ of the center of $ U(\mf{gl}_r)$,
\begin{eqnarray}
    \label{Casimir}
    C^{(k)} &=& \sum_{1\leqslant i_1,j_1,\dots,i_k,j_k\leqslant r} C^{(k)}_{(i_1,j_1),\dots,(i_k,j_k)}\mathbf e_{i_1,j_1}\otimes \cdots \otimes \mathbf e_{i_k,j_k} \\
 \label{DiagCasimir}   &=&\sum_{1\leqslant a_1<\dots<a_k\leqslant r} e_{a_1}\otimes\cdots\otimes e_{a_k}\,,
\end{eqnarray}
whose coordinates are obtained as coefficients of the characteristic polynomial function
\begin{equation}
  \label{CharacteristicPolynomial}
  \det \big(\omega \Id - E\big) \eqcolon \sum_{k=0}^r(-1)^k \omega^{r-k}\sum_{1\leqslant i_1,j_1,\dots,i_k,j_k\leqslant r} C^{(k)}_{(i_1,j_1),\dots,(i_k,j_k)} E^{i_1,j_1}\cdots E^{i_k,j_k} \,,
\end{equation}
for $E=\sum_{i,j=1}^r E^{i,j} \mathbf e_{i,j}$, and whose diagonal expression \eqref{DiagCasimir} readily follows from expressing $E$ in Jordan normal form.

\begin{proposition} 
  \label{DefNonPertLE}
  For every positive integer $n\geqslant1$, any generic $J=\big\{\overset{E_1}{z_1},\dots,\overset{E_n}{z_n}\big\}$, and any point $z\in\Sigma \setminus \{ 0, \infty, z_1, \dotsc, z_n\}$, the amplitudes constructed as $\hslash$-expansions from copies of the WKB solution necessarily satisfy the \emph{non-perturbative loop equations}, 
  \begin{equation}
    \begin{split}
    \label{NonPertLoopEq}
      \sum_{k=0}^r(-1)^k \omega^{r-k} \widehat W_{k+n}\big(\overset{C^{(k)}}{\overbrace{z,\dots,z}} ,J\big) 
      &=
      [\delta_1\cdots\delta_n] \det \Big(\omega \Id - \Phi_\hslash\big(x(z)\big)-\mathcal M^{(n)}_{\vec\delta}\big(x(z);J\big) \Big)
      \\
      &\eqcolon
      P_n(x(z),\omega ;J) \,,
    \end{split}
  \end{equation}
  where the $k$ first variables of $\widehat W_{k+n}$ are linearly evaluated at the $k^{\text{th}}$ Casimir over the same point, $\vec\delta \coloneq (\delta_1,\dots,\delta_n)$ is a vector of formal variables,
  \begin{equation}
    \mathcal M^{(n)}_{\vec\delta}\big(x;J\big)
    \coloneq
    \sum_{k=1}^n \sum_{1\leq i_1\neq\cdots\neq i_k\leq n} \delta_{i_1}\cdots \delta_{i_k} \frac{\underset{1\leq j\leq k}{\overset\rightarrow \prod} M_\hslash(z_{i_j}, E_{i_j})}{(x-x_{i_1})\cdots (x_{i_k}-x)} dx \prod_{i=1}^k dx_i \,,
  \end{equation}
  and $[\delta_1\cdots\delta_n] P(\vec{\delta}) $ equals the coefficient of the monomial $\delta_1\cdots\delta_n$ in the polynomial expression $ P (\vec{\delta}) $ of $\vec\delta$.
\end{proposition} 

\begin{proof}
  This collection of identities follows from the manipulations introduced in \cite{BEM18} in the proof of Theorem 4.3 of that paper, for amplitudes defined from genuine solutions of differential systems. This proof is still valid in the present context of formal WKB solutions, when understood as manipulations of formal $\hslash$-expansions.
\end{proof}

\begin{remark}
  Generally speaking, the validity of those equations follows from an expansion of the determinant in powers of $ \omega $. The interesting content of this collection of identities is that certain algebraic combinations of non-connected amplitudes exhibit the analytical structure of the expressions appearing in the right-hand side of \eqref{NonPertLoopEq}.
\end{remark}

\subsection{Topological type property}

According to \cite{BBE15}, solutions of \eqref{NonPertLoopEq} that are of topological type can be computed by topological recursion. Furthermore, there exist sets of sufficient conditions on the differential system $\nabla_\hslash$ that ensure that this is the case, see \cite{BEM17}. Our setup satisfies all those sufficient conditions except one. In this section, we recall those sufficient conditions that are satisfied, and describe how to supplement them in order to save the day, showing that the corresponding solutions are indeed computed by a topological recursion, albeit a shifted one.  

The topological type property was defined in \cite[Definition~3.3]{BBE15} and refined in \cite[Definition~5.1]{BEM18}. Here we give a restricted definition that suffices in our context.

\begin{definition}
  \label{TopologicalType}
  A collection  $\{ W_n \}_{n\geq 1}$ of meromorphic symmetric $n$-form sections of the trivial bundle $ \mf{gl}_r \times \P^1 $ satisfies the \emph{topological type property} if
  \begin{enumerate}
    \item There exists a cover $ x \colon \Sigma \to \P^1 $ over which each $W_n$ admits an $\hslash$-expansion whose coefficients are rational functions;
    \item Apart from $ [\hslash^{-1}] W_1$ and $ [\hslash^0] W_2$, the coefficients of the $ W_n$ may only have poles at the ramification points of $ x$. Moreover, $ [\hslash^0] W_2 $ has a residueless double pole over the diagonal and no other singularity;
    \item The $\hslash$-expansion of each $W_n$ has first non-trivial coefficient at order $\mc{O}(\hslash^{n-2})$.
  \end{enumerate}
\end{definition}

\begin{remark}
  \label{Parity}
  One condition is missing here in comparison to \cite{BBE15,BEM18}. This condition concerns the parity of the construction under a sign change of $\hslash$, ensuring that no half-genus invariants appear. We do not mind half-genus terms, since they appear generically, and play no specific role here.
\end{remark}

The reason the topological type property was introduced, is the following theorem.

\begin{theorem}[{\cite[Corollary~3.6]{BBE15}}]
  If a collection of connected amplitudes satisfies the non-perturbative loop equations and the topological type property, and the cover $x \colon \Sigma \to \P^1$ of \cref{TopologicalType} has only simple ramification, then the evaluations of the expansion coefficients of the amplitudes at diagonal basis matrices can be calculated by topological recursion, i.e.
  \begin{equation}
    W_n ( \overset{e_{r}}{z_1}, \dotsc, \overset{e_{r}}{z_n} ) = \sum_{g \in \frac{1}{2} \N } \hslash^{2g-2+n} \omega_{g,n} (z_1, \dotsc, z_n)\,.
  \end{equation}
\end{theorem}

By equivariance of the amplitudes, \cref{WEquivariant}, this in fact determines the $W_n$ on the entire Cartan.

Let us then prove the topological type property. We start with the following lemma.

\begin{lemma}\label{Assumption4Bypass}
  The leading order of $ W_2$ in $ \hslash$, $\omega_{0,2}$, is the unique rational symmetric bi-differential on the spectral curve $\Sigma$ that has a residue-less double pole on the diagonal with unit biresidue. Namely, it is the Bergman kernel on the Riemann sphere, given by
  \begin{equation}
    \label{TRB}
    \omega_{0,2}(z_1,z_2) = \frac{dz_1 dz_2}{(z_1-z_2)^2}\, .
  \end{equation}
\end{lemma}

\begin{proof}
  This follows from two  steps. The first is a direct computation of the leading term of $W_2$ in the WKB approximation, yielding near coinciding point asymptotics. The second step uses the explicit formula for the invertible matrix of eigenvectors of the Higgs field $\varphi$ to calculate the pole order at ramification points.\par
  Indeed, although $V(z)$ does not have the required form to satisfy Assumption 4 of \cite{BEM17}, it is given by the simple expression \eqref{AbelianizationV}, which in particular implies its equivariance under deck transformations. We start by noticing that the leading order WKB approximation
  \begin{equation}
    \label{ClassicalW2}
    \omega_{0,2}(z_1,z_2) = \Tr(e_r V(z_1)^{-1}V(z_2) e_r V(z_2)^{-1}V(z_1)) \frac{dx(z_1)dx(z_2)}{(x(z_1)-x(z_2))^2}
  \end{equation}
  has possible singularities only at the origin, infinity, and along the pre-image of the diagonal over the base. We first study the vicinity of the latter, using the expansions
  \begin{align}
    V(z_1)^{-1}V(z_2) 
    &
    \underset{z_2\sim\theta^p z_1}\sim \tau^p + \frac{x(z_2)-x(z_1)}{x'(z_1)}V(z_1)^{-1}\frac{dV(z_1)}{dz}\tau^p + \mathcal O(x(z_2)-x(z_1))^2\,,
    \\
    V(z_2)^{-1}V(z_1)
    &
    \underset{z_2\sim\theta^p z_1}\sim \tau^{-p} - \frac{x(z_2)-x(z_1)}{x'(z_1)}\tau^{-p}V(z_1)^{-1}\frac{dV(z_1)}{dz} + \mathcal{O} (x(z_2)-x(z_1))^2\,.
  \end{align}
  Introducing the Maurer--Cartan form $\Omega\coloneq V^{-1} \frac{dV}{dz}$, we immediately get
  \begin{equation}
    \begin{split}
      &\omega_{0,2}(z_1,z_2) 
      \underset{z_2\sim \theta ^p z_1}\sim \ \Tr(e_r\tau^p e_r \tau^{-p}) \frac{dx(z_1)dx(z_2)}{(x(z_1)-x(z_2))^2}
      \\
      &
       \qquad \qquad+  \Big(\frac{\Tr(e_r\tau^p e_r \tau^{-p}\Omega(z_1))}{rz_1^{r-1}}- \frac{\Tr(e_r \Omega(z_1)\tau^p e_r\tau^{-p})}{rz_1^{r-1}}\Big) \frac{dx(z_1)dx(z_2)}{x(z_1)-x(z_2)}+\mathcal{O} (1)\,,
    \end{split}
  \end{equation}
  which simplifies to
  \begin{equation}
    \omega_{0,2}(z_1,z_2) \underset{z_2\sim \theta^p z_1}\sim \delta_{p,0}\frac{dx(z_1)dx(z_2)}{(x(z_1)-x(z_2))^2} + \frac{\Tr([e_r,\tau^p e_r \tau^{-p}]\Omega(z_1))}{rz_1^{r-1}}\frac{dx(z_1)dx(z_2)}{x(z_1)-x(z_2)} +\mathcal{O}(1)\,.
  \end{equation}
  Since the commutator appearing in the numerator of the second term of the right-hand side is between two diagonal matrices, it vanishes. So we obtain the equivalence
  \begin{equation}
    \label{NearDiagB}
    \omega_{0,2}(z_1,z_2) \underset{z_2\sim \theta^p z_1}\sim \delta_{p,0}\frac{dx(z_1)dx(z_2)}{(x(z_1)-x(z_2))^2} +\mathcal{O}(1)\,,
  \end{equation}
  implying that $\omega_{0,2}$ has a residue-less double pole with unit biresidue on the diagonal over the spectral curve, but is regular at each other pre-image of the diagonal over the base for which the Kronecker delta vanishes.\par
  It could however still have poles over the origin and infinity; but we will show that it does not. Consider the matrix product
  \begin{equation}
    \label{Vprod}
    [V(z_1)^{-1}V(z_2)]_{i,j} = \frac1r\sum_{k=1}^r \theta^{k(j-i)}\big(\frac{z_2}{z_1}\big)^{r\lfloor\alpha_k\rfloor - k(r-s)}\,,
  \end{equation}
  and replace this expression in that of $\omega_{0,2}$ \eqref{ClassicalW2}. The trace evaluation yields
  \begin{equation}
    \begin{split}
      \omega(z_1,z_2) 
      &=
      [V(z_1)^{-1}V(z_2)]_{r,r} [V(z_2)^{-1}V(z_1)]_{r,r}\frac{dx(z_1)dx(z_2)}{(x(z_1)-x(z_2))^2}
      \\
      &=
      \Big(\sum_{k=1}^r \big(\frac{z_2}{z_1}\big)^{\lfloor\alpha_k\rfloor - k(r-s)}\Big)\Big(\sum_{\ell=1}^r \big(\frac{z_2}{z_1}\big)^{\lfloor\alpha_\ell\rfloor - \ell(r-s)}\Big)\frac{rz_1^{r-1}dz_1 rz_2^{r-1}dz_2}{(z_1^r-z_2^r)^2}
    \end{split}
  \end{equation}
  upon minor simplifications. Furthermore, by definition of the floor, $\alpha_k-1<\lfloor \alpha_k\rfloor \leq \alpha_k$, which implies that the pole order of this expression as a function of $z_1$ near the origin is at most 
  \begin{equation}
    (r-1)+0-(r-1)\, =\, 0,
  \end{equation}
  accounting for each factor of the right-hand side of the last expression. So $\omega_{0,2}(z_1,z_2)$ is regular at $z_1=0$ and generic $z_2$. This same inequality also implies $\omega_{0,2}$ is regular at $z_1=\infty$ at generic $z_2$, as well as $z_2=0$ and $z_2=\infty$ at generic $z_1$ respectively. Therefore, $\omega_{0,2}$ only has the singularities appearing in \eqref{NearDiagB} on $\Sigma^2$. Since the spectral curve has genus zero, there is a unique symmetric bidifferential with this pole behaviour, and it is the one given in the lemma. 
\end{proof}

\begin{lemma}
  In the setting of \cref{OurQC}, conditions (1-2) of \cref{TopologicalType} are satisfied.
\end{lemma}

\begin{proof}
  The $W_n $ have an $ \hbar $-expansion whose coefficients are rational functions in $ z$, the coordinate on the cover $ x \colon \Sigma \to \P^1$, by \cref{DefNPAmplitudes,ExpansionOfM} -- this makes sense because $ g(\Sigma )  = 0$. This proves condition (1).

  The shape of $ [\hslash^0] W_2$ was considered in \cref{Assumption4Bypass}. For the further pole structure, we see from \cref{DefNPAmplitudes} that poles can only be at diagonals or at poles of $M_\hslash $. By the argument in \cite[p. 3230]{BEM17}, diagonal poles do not appear. On the other hand, poles of $M_\hslash $ must come from poles or zeroes of $ \Phi_\hslash $, i.e. the ramification points $ 0$ and $ \infty$, proving condition (2).
\end{proof}

Therefore, in our context and as is usually the case in such problems, the hardest part in proving that the conditions of \cref{TopologicalType} are satisfied is to determine the leading order of the $W_n$ (condition 3). 

Over the years, different methods have been devised to prove this leading order property. Let us mention a few of them. The first method is the enumerative one, cf. e.g. \cite{CEO06,EO07,EO09,EMS11}, using of the interpretation of the amplitudes as generating functions of certain quantities, e.g. in enumerative geometry, when available.\par
In situations where such an interpretation of the amplitudes of a $ \hslash$-connection is not available, \cite{BBE15} introduced a recursive process, using a differential Galois theory approach by integrable loop-insertion operators, when available. This approach was simplified to a combinatorial method in \cite{BEM17}, which is the approach we shall extend in the present work. It uses the combinatorial structure encoded in the loop equations \eqref{NonPertLoopEq}, together with a certain assumption on the expression on the right side of this equation, to prove the leading order property by induction.

This assumption was stated as \cite[Assumption 5]{BEM17}, as a sufficient condition for the leading order property to hold, and it is not satisfied by the differential system of \cref{OurQC}. However, we will adapt it, leading to our notion of shifted perturbative loop equations, and corresponding shifted topological recursion. Let us note additionally that the four assumptions \cite[Assumptions 1--4]{BEM17} preceding this fifth one are satisfied or unnecessary in our situation:
\begin{enumerate}
  \item Assumption 1 states that $ \Phi_\hslash $ has a formal power series expansion in $\hslash$ with coefficients rational functions of $x$. We assume this;
  \item Assumption 2 states that the associated spectral curve is genus $0$. We also assume this;
  \item Assumption 3 is only used to control the behaviour of the spectral curve involved. As our spectral curve is well-behaved already, we do not need it;
  \item Assumption 4 is not satisfied in our context by the invertible matrix of eigenvectors $V(z)$, since it does not take the form required by \cite{BEM17}. There, however, this assumption was only used to determine analytic properties of $\omega_{0,2}$ that are relevant to topological recursion, cf. \cite[Remark 2.6]{BEM17}. We already calculated this in \cref{Assumption4Bypass}, hence bypassing Assumption 4.
\end{enumerate}

There is also an Assumption 6 \cite[Assumption 6]{BEM17}, but it is only relevant to the parity condition, which following \cref{Parity}, we ignore.

As the leading order property deals with the $\hslash$-expansion of the amplitudes, it is only natural that the assumption allowing us to derive it involves the $\hslash$-dependence of the connection $\nabla_\hslash$ seen as $\hslash$-corrections to the Higgs field $\varphi$. 

\begin{definition}[{\cite{BEM17}}]\label{def:Assumption5}
  The $\hslash$-connection $ \nabla_\hslash$ satisfies \emph{Assumption 5} if the following two statements hold:
  \begin{itemize}
    \item The set of singularities of each $\Phi_k$, $k\geq1$,  featuring in the $\hslash$-expansion of $\Phi_\hslash$ is included in that of $\varphi$,
    \item For any $r\times r$ matrix $C$, and any generic base-points $x_0, x_1\in \C^\times$, the $\hslash$-series of rational expressions of the pair $(x, \omega )$ given by
    \begin{eqnarray}
    \label{Assumption5}
        \Big(\det \big(\omega \Id - \Phi_\hslash(x)-\frac{C}{(x-x_0)(x-x_1)}\big) - \det \big(\omega \Id - \varphi(x)\big)\Big)\frac{1}{E_\omega (x, \omega )}
    \end{eqnarray}
    restricts to a one-form on $\Sigma$ that is analytic at each singularity of $\varphi$, with $E_\omega =\frac\partial{\partial \omega }E$ and $E$ defined in \eqref{eq:Sigma}.
  \end{itemize}
\end{definition}

The assumptions are used in \cite{BEM17} in the following way.

\begin{theorem}[{\cite[Theorems 3.1 \& 3.2]{BEM17}}]
  \label{FirstPartOfBEM17}
  For a rational Lax pair system satisfying Assumptions 1, 2, and 4, $ M_\hslash (z,D) $ has an expansion of the shape \eqref{MExpansion}.\par
  If the system also satisfies Assumption 3, then the $ M^{(k)}$ may only have poles at branch points or poles of $ \phi$.
\end{theorem}

\begin{theorem}[{\cite[Section 4]{BEM17}}]
  If a system satisfies Assumptions 1, 2, and 5  and the conclusion of \cref{FirstPartOfBEM17}, then it satisfies the topological type property.\par
  If it also satisfies Assumption 6, then moreover $ W_n \big|_{\hslash \to - \hslash} = (-1)^n W_n$ (this is part of the topological type property in that paper).
\end{theorem}

We reach the same conclusion as \cref{FirstPartOfBEM17} by combining \cref{ExistenceWKB,ExpansionOfM}.

\begin{proposition}
  \label{MExpansionInExample}
  In the situation of \cref{OurQC}, $ M_\hslash $ has an expansion of the shape \eqref{MExpansion}. By construction, the expansion coefficients may only have poles at $ 0$ and $ \infty$.
\end{proposition}

We will however see explicitly in \cref{Assumption5ConstantTerm} that \eqref{Assumption5} is not satisfied by the connection potential \eqref{ConnectionPotential} if the $ S^\hslash_j$ are non-zero. So we will need to relax the conditions of the assumption.

A first hint that it might be too restrictive is that the right-hand side of the non-per\-tur\-ba\-tive loop equations \eqref{NonPertLoopEq} does not feature the expression appearing as second term in the numerator of \eqref{Assumption5}, but only particular coefficients of some polynomial expressions of $\vec\delta=(\delta_1,\dots,\delta_n)$, for each $n\geq0$.\par
An important subtlety is then that when $n=0$, the generic matrix $C$ in \eqref{Assumption5} can be taken to be trivial, but when  $n\neq0$, the right-hand side of \eqref{NonPertLoopEq} can only have lower pole order at $x=0$ than that of the values of the Casimir operators on $\Phi_\hslash$. The refinement we propose focuses on this particular point, is satisfied by \eqref{ConnectionPotential}, and does not affect the sequence of steps in which the assumption is used.

Since the correlators satisfying the (perturbative) shifted abstract loop equations can be computed inductively by the corresponding shifted topological recursion, and since the underlying Airy ideal partition function is unique, it follows that an assumption implying the reduction of the non-perturbative loop equations to the shifted perturbative ones will identify the topological expansion of the non-per\-tur\-ba\-tive connected amplitudes associated to the differential system that constitutes the quantum curve. Let us therefore formulate the sufficiently refined assumption and check that it is indeed satisfied in our case. 

The first step is to notice that the (perturbative) shifted loop equations \eqref{eq:sle} are indexed by two labels $n,g\geq0$ corresponding, from the quantum curve point of view, to the number of spectator variables and order in the $\hslash$-expansion respectively. Multiplying each combination $\mathcal E^i_{g,n}$ by the relevant power of $\hslash$ and summing over all values of the genus label then reproduces the left-hand side of the non-perturbative loop equations, albeit up to the subtraction of the order $i$ differential $S_i^\hslash\big(\frac{dx}x\big)^i$ in the $n=0$ case. This re-summed shift exactly matches the value of the $i^{\text{th}}$ Casimir on $\Phi_\hslash$, encoded in the asymptotic equivalence
\begin{equation}\label{AsymptoticLE}
  \begin{split}
    P_n \big( x(z), \omega (z); J\big) 
    &\eqcolon
    \sum_k \hslash^k P_n^{(k)}(x(z),\omega (z);J\big)
    \\
    &\sim 
    \sum_{g=0}^\infty \hslash^{2g-2+n}\mathcal \sum_{i=0}^r (-1)^i \omega (z)^{r-i} \mc{E}^i_{g,n}\big(x(z);J\big) \, ,
  \end{split}
\end{equation}
with $\mathcal{E}^i_{g,n}$ defined combinatorially in \cref{d:EW}. This asymptotic equivalence follows directly by comparing \cref{d:EW,DisNPamplitude,NonPertLoopEq}.

In \cite{BEM17}, Assumption 5 (\cref{def:Assumption5}) is used only in the form of the following corollary:

\begin{corollary}[{\cite[Corollary 4.4]{BEM17}}]\label{CorToAssumption5}
  The non-perturbative loop equations, \cref{DefNonPertLE}, and Assumption 5, \cref{def:Assumption5}, imply the following property: for every $ (k,n) \in \N \setminus \{ (0,0)\}$ and for every generic $ x_1, \dotsc, x_n$, the function
  \begin{equation}
    P_n^{(k)}(x(z),\omega(z); x_1, \dotsc, x_n) \frac{1}{E_\omega (x(z),\omega(z))}
  \end{equation}
  is a one-form on the Riemann sphere, whose poles may only be at coinciding points or at double points\footnote{Which are absent in our situation.}.
\end{corollary}

As the shifts contribute only to the $ \mc{E}^i_{g,0}$, Assumption 5, or rather its \cref{CorToAssumption5}, is naturally refined by distinguishing the $n=0$ and $n\neq0$ cases of the non-perturbative loop equations, as follows.

\begin{definition}\label{Assumption5*}
  A formal rational $\hslash$-connection written $\nabla_\hslash=\hslash d-\Phi_\hslash$, with $\Phi_\hslash=\sum_{\ell\geq0}\hslash^\ell\Phi_\ell$, satisfies \emph{Assumption 5*} if the following statements are true.
  \begin{itemize}
    \item For all $\ell>0$, all the singularities of $\Phi_\ell$ are among those of $\varphi=\Phi_0$.
    \item For any number $n\geq 1$ of spectator variables, every expression of the form 
    \begin{equation}
        [\delta_1\cdots\delta_n]\det \Big( \omega \Id - \Phi_\hslash\big(x\big)-\mathcal M^{(n)}_{\vec\delta}\big(x(z);J\big) \Big) \frac{1}{E_\omega (x,\omega )} 
    \end{equation}
    restricts to a one-form on the spectral curve $\Sigma$ that is analytic at each singularity of $\varphi$.
    \item There are constants $S_i^\hslash = \sum_g S_{i,2g} \hslash^{2g} \in \C \llbracket \hslash \rrbracket $ such that 
    \begin{equation}\label{eq:Assumption5*n=0}
        \det \Big( \omega \Id - \Phi_\hslash\big(x(z)\big)  -  \sum_{i=1}^r  (-1)^i S_i^\hslash \omega^{r-i} \big( \frac{dx}{x})^i \Big) \frac{1}{E_\omega (x,\omega )}
    \end{equation}
    restricts to a one-form on the spectral curve $\Sigma$ that is analytic at each singularity of $\varphi$.
  \end{itemize}
\end{definition}

As stated before, \cref{Assumption5*} is given to fit the following proposition, adapting known results to our shifted setting.

\begin{proposition}
  \label{Assumption5*ToTT}
  If a rational $ \hslash $-connection has a smooth genus $0$ spectral curve and satisfies
  \begin{enumerate}
    \item Assumption 5*;
    \item The conclusion of \cref{FirstPartOfBEM17};
    \item $ [\hslash^0]W_2 $ is the Bergman kernel $ \frac{dz_1 \, dz_2}{(z_1-z_2)^2}$,
  \end{enumerate}
  then it satisfies the topological type property.
\end{proposition}

\begin{proof}
  Condition (1) of the topological type property, \cref{TopologicalType}, is a consequence of the definition of the amplitudes, \eqref{NPamplitudes}, and \cref{FirstPartOfBEM17}.\par
  Again by the definition of the amplitudes, they can only have poles at poles of the $ M_\hslash$ and at coinciding points. The poles at coinciding points only contribute to $ [\hslash^0]W_2$ by the argument in the second bullet point of \cite[Section 4.4.]{BEM17}, which proves condition (2).\par
  Condition (3) is the hardest to prove. We postpone it to the appendix: see \cref{LeadingOrderProperty}, which clearly implies the leading order property.
\end{proof}

\begin{proposition}
  \label{From5*ToTR}
  If a rational $ \hslash $-connection has a smooth genus $0$ spectral curve, satisfies Assumption 5* and the conclusion of \cref{FirstPartOfBEM17}, and has the Bergman kernel as leading order of $ W_2$, then its non-perturbative connected amplitudes can be expanded in powers of $\hslash$, and the coefficients satisfy the shifted loop equations \eqref{eq:sle}. Hence, if the shifts are $s$-consistent, \cref{d:consistent}, they can be calculated by shifted topological recursion.
\end{proposition}

\begin{proof}
  By \cref{Assumption5*ToTT}, such an $\hslash$-connection satisfies the topological type property, \cref{TopologicalType}. Together with \cref{WEquivariant,MEquiv}, this means that its amplitudes on the Cartan can be expanded as
  \begin{equation}
    W_n ( \overset{e_r}{z_1}, \dotsc, \overset{e_r}{z_n}) = \sum_{g \in \frac{1}{2} \N} \hslash^{2g-2+n} \omega_{g,n}(z_1, \dotsc, z_n)
  \end{equation}
  and that the $ \omega_{g,n}$ for $2g-2+n > 0$ may only have poles at the ramification points. Because the spectral curve is smooth of genus $0$, this means exactly that they satisfy the projection property, \cref{d:projection}.\par
  Moreover, comparing \cref{DefNonPertLE,Assumption5*,AsymptoticLE}, and using that $ r$ and $s$ are coprime, we see that each term
  \begin{equation}
    \frac{\omega(z)^{r-i}}{E_\omega (x(z),\omega(z))} \mc{E}^i_{g,n} (x(z);J)    
  \end{equation}
  with $ n> 0$ is analytic at singularities, i.e.
  \begin{equation}
    \mc{E}^i_{g,n} (x(z);J) 
    \in
    \mc{O}(\frac{E_\omega (x(z),\omega(z))}{\omega(z)^{r-i}} dz) 
    = 
    \mc{O}(\frac{\omega(z)^{r-1}}{\omega(z)^{r-i}}dz) 
    =
    \mc{O}(z^{s(i-1)+1}\big(\frac{dx(z)}{x(z)}\big)^i) \,.
  \end{equation}
  But because it only depends on $ z$ through $x$, we rather have
  \begin{equation}
    \mc{E}^i_{g,n} (x;J) 
    \in
    \mc{O}(x^{\lceil\frac{s(i-1)+1}{r}\rceil}\big(\frac{dx}{x}\big)^i)
    =
    \mc{O}(x^{\lfloor\frac{s(i-1)}{r}\rfloor + 1}\big(\frac{dx}{x}\big)^i) \,,
  \end{equation}
  which is exactly the shifted loop equation \eqref{eq:sle} for these $ n$.\par
  For $n=0$, we do the same thing for
  \begin{equation}
      \frac{\omega(z)^{r-i}}{E_\omega (x(z),\omega(z))} \Big( \mc{E}^i_{g,n} (x(z);J) - S_{i,2g} \Big( \frac{dx(z)}{x(z)}\Big)^i \Big)
  \end{equation}
  to get the shifted loop equation \eqref{eq:sle} for $ n = 0$.

  We then use \cref{ShiftedTR,Rem:GenShiftedTR} to conclude that the $ \omega_{g,n}$ can be calculated by shifted topological recursion.
\end{proof}

\subsection{The main example}

Let us now return to our main case, \cref{OurQC}. Requiring Assumption 5* constrains the values of the parameters $r$ and $s$, as well as the values of the expansion coefficients of the $\hslash$-series $S_i^\hslash$, $i\in\{1,\dots,r\}$.

\begin{proposition}\label{ConditionsForAssumption5}
  Let $ r $ and $ s$ be coprime, and write $ r = r's + r''$ for division with remainder $ 1< r'' < s $ -- so $ r' = \lfloor \frac{r}{s} \rfloor$. Consider the spectral curve $ x(z) = z^r$ and $ y(z) = z^{s-r}$, and given $r$ integers $ \{ a_1, \dotsc, a_r\} $ with $ a_1 = r-s$, consider $ \Phi_\hslash$ given by
   \begin{equation}\label{GenConnection}
    \Phi_\hslash(x)
    \coloneq
    \sum_{k=0}^\infty \hslash^k\Phi_k(x) 
    \coloneq  
      \begin{pmatrix}
        (-1)^0 S_1^\hslash \cdot x^{a_1 - a_1 - 1}  & x^{a_2 - a_1} &  \cdots & 0
        \\
        \vdots & 0 & \ddots & \vdots
        \\
        (-1)^{r-2} S_{r-1}^\hslash \cdot x^{a_1 - a_{r-1} - (r-1) } &\vdots & \ddots & x^{a_r - a_{r-1}}
        \\
        (-1)^{r-1}S_{r}^\hslash \cdot x^{a_1 - a_r - r}+x^{-a_r} & 0 & \cdots &  0
      \end{pmatrix}  
      dx\,,
  \end{equation}
  Then, consider the expression
  \begin{equation}\label{Assumption5det}
    D(z,M) \coloneq \det ( y (z) \, dx(z) \Id - \Phi_\hslash (x(z)) - M dx(z)) \frac{1}{E_\omega (x(z),y(z) \, dx(z))}
  \end{equation}
  with $M$ considered as a matrix of formal variables, with no pole at $ z = 0$. Then
  \begin{equation} \label{Assumption5ConstantTerm}
    D(z,0) = \sum_{j=1}^r  (-1)^j S_j^\hslash z^{ (1-j)s  -1} dz \,.
  \end{equation}
  Moreover, if 
  \begin{itemize}
    \item either for some $j \in [r]$, $a_{j+1} -a_j < -1$\footnote{We use the convention $ a_{r+1} = 0$ here.}; 
    \item or there are more than $ r'$ consecutive $j$ such that $a_{j+1} -a_j = -1$;
    \item or there are more than $ r''$ disjoint subsequences of $ r'$ consecutive $j$ such that $a_{j+1} -a_j = -1$,
  \end{itemize}
  then for any shifts $ S_j^\hbar $, there is an $M$ such that $ D(z,M) $ has a pole at $z =0$.
  Therefore, $D(z,M) - D(z,0)$ is holomorphic at $ z = 0$ only if there are exactly $ r''$ disjoint subsequences of $r'$ consecutive $j$ such that $a_{j+1} -a_j = -1$ (which we call \emph{large blocks}, and $ s - r'' $ subsequences of $ r' -1$ consecutive $j$ such that $a_{j+1} -a_j = -1$ (which we call \emph{small blocks}. The remaining differences are $a_{j+1} - a_j = 0$.

  Assuming the optimal shape from above, the non-constant terms of $ D$ in $M$ have pole order at most
  \begin{itemize}
    \item $ (r'' - \frac{s}{2} )^2 - (1 - \frac{s}{2} )^2 $ if they do not contain any $S^\hslash_j$;
    \item $0$ if $s=1$ and they do contain $ S^\hslash_j$;
    \item $ r^2 - sr + js $ if $ s > 1$ and they contain $ S^\hslash_j$ for $ j > r'+1$;
    \item $s(2-j) + (r'' - \frac{s+1}{2} )^2 - 1 - \big(\frac{s + 1}{2}\big)^2$ if $ s > 1$, $ r'' \neq 1$, the first block is large, and they contain $ S^\hslash_j $ for $ j \leq r'+1$;
    \item $ s(1-j) $ if $ s > 1$, $ r'' = 1$, the first block is large, and they contain $ S^\hslash_j $ for $ j \leq r'+1$;
    \item $s(1-j) + (r'' -s)r'' - 1$ if the first block is small, and they contain $ S^\hslash_j $ for $ j \leq r'+1$
  \end{itemize}
  Hence, $ D(z,M) - D(z,0)$ is holomorphic at $ z = 0$ if and only if $ r'' \in \{ 1, s-1\} $ and one of the following three conditions holds:
  \begin{enumerate}
    \item $ s= 1$;
    \item $ r = 1 \pmod{s} $, the unique large block is in the top left, and $ S^\hslash_j = 0$ for $ j > 1$;
    \item all $ S^\hslash_j = 0$.
  \end{enumerate}
\end{proposition}

This proposition calculates the pole orders involved in the non-perturbative loop equations, \cref{DefNonPertLE}, though Assumption 5*. The first two cases found actually recover cases from \cref{OurQC}. The last one seems slightly more general, as the blocks need not be in the same order. However, because $ r'' \in \{ 1, s-1 \}$, all but one of the blocks are small (resp. large), and therefore if all $ S^\hslash_j = 0$ the blocks can be reordered by conjugation with a constant matrix, which is a power of $ \tau $ from \cref{PermutationMatrix}. So up to a simple gauge transformation, \emph{all matrices of the shape \cref{GenConnection} that satisfy Assumption 5* are of the shape \cref{OurQC}.}\par

The proof strategy of this proposition is straightforward: for each of the cases, we just try to get as high a pole as we can with the given conditions. However, writing it down in general obscures the intuition, so we will give an example to explain the features.

\begin{example}
  All $ \Phi_\hslash$ as in \cref{ConnectionPotential} fall in the scope of this proposition, with $a_j = \lfloor \alpha_{r+1-j} \rfloor$. We will look at some of these cases.
  Let us first consider the case $ (r,s) = (5,2)$. Then
  \begin{equation}
    \omega \Id - \phi 
    = 
    \begin{pNiceMatrix}
      \Block[draw]{3-3}{} z^{-3} & - z^{-5} & 0 & 0 & 0
      \\
      0 & z^{-3} & - z^{-5} & 0 & 0
      \\
      0 & 0 & z^{-3} & -1 & 0
      \\
      0 & 0 & 0 & \Block[draw]{2-2}{} z^{-3} & - z^{-5}
      \\
      1 & 0 & 0 & 0 & z^{-3}
    \end{pNiceMatrix}
    dx
  \end{equation}
  For the case $(r,s) = (7,4)$, we find
  \begin{equation}
    \omega \Id - \phi 
    = 
    \begin{pNiceMatrix}
      \Block[draw]{2-2}{} z^{-3} & - z^{-7} & 0 & 0 & 0 & 0 & 0
      \\
      0 & z^{-3} & - 1 & 0 & 0 & 0 & 0
      \\
      0 & 0 & \Block[draw]{2-2}{} z^{-3} & - z^{-7} & 0 & 0 & 0
      \\
      0 & 0 & 0 & z^{-3} & -1 & 0 & 0
      \\
      0 & 0 & 0 & 0 & \Block[draw]{2-2}{} z^{-3} & - z^{-7} & 0
      \\
      0 & 0 & 0 & 0 & 0 & z^{-3} & -1
      \\
      1 & 0 & 0 & 0 & 0 & 0 & \Block[draw]{1-1}{} z^{-3}
    \end{pNiceMatrix}
    dx
  \end{equation}
  In both of these matrices, all of the possible pole contributions are in the indicated blocks. The sizes of the blocks are either $ r' $ or $ r' + 1$, and there are $ s - r''$ of the first case and $ r''$ of the second. In the first case, the diagonal will give a higher pole order, while in the second case, the off-diagonal (supplemented by a non-pole from $M$) will give a higher pole order.\par
  The blocks are ordered by size, with larger ones coming first, but this only happens if $ r = \pm 1 \pmod{s}$. It is related to the result of \cite{BBCCN18}, cf. \cref{t:rsAs}, that only for these we get partitions (i.e. non-increasing tuples of numbers). As an example, consider the case $(r,s) = (7,5)$:
  \begin{equation}
    \omega \Id - \phi 
    = 
    \begin{pNiceMatrix}
      \Block[draw]{2-2}{} z^{-2} & - z^{-7} & 0 & 0 & 0 & 0 & 0
      \\
      0 & z^{-2} & - 1 & 0 & 0 & 0 & 0
      \\
      0 & 0 & \Block[draw]{1-1}{} z^{-2} & -1 & 0 & 0 & 0
      \\
      0 & 0 & 0 &  \Block[draw]{2-2}{} z^{-2} & - z^{-7} & 0 & 0
      \\
      0 & 0 & 0 & 0 & z^{-2} & -1 & 0
      \\
      0 & 0 & 0 & 0 & 0 & \Block[draw]{1-1}{} z^{-2} & -1
      \\
      1 & 0 & 0 & 0 & 0 & 0 & \Block[draw]{1-1}{} z^{-2}
    \end{pNiceMatrix}
    dx
  \end{equation}

  For $ \omega \Id - \Phi_\hslash$, the $ S^\hslash_j$ also contribute. These are all  in the first column. This requires a bit more analysis, but it turns out that these contributions are only allowed if they fit in the top-left block anyway.
\end{example}

\begin{proof}
  Rewrite $ \Phi_\hslash \eqcolon F_\hslash dx $. Then 
  \begin{equation}
    D(z,M) = \det ( y(z) \Id - F_\hslash -M ) \frac{dx(z)}{P_y(x(z),y(z))} \,.
  \end{equation}
  The last factors can be calculated to give
  \begin{equation}\label{Assumption5Prefactor}
    \frac{dx(z)}{P_y (x(z),y(z))} = \frac{dz^r}{r y(z)^{r-1}} = \frac{r z^{r-1} dz}{r z^{(r-1)(s-r)}} = z^{(r-1)(r+1-s)} dz = z^{r^2 - 1 - rs + s} dz \,.
  \end{equation}
  Now for the determinants. Write $ Y = y(z)\Id$. Its non-zero entries are clearly $ Y_{k,k} = y(z) = z^{s-r}$.
  The matrix $ F_\hslash (x(z)) $ has non-zero entries
  \begin{align}
    F_{1,j} 
    &=
    (-1)^{j-1} S_j^\hslash x^{a_1 - a_j -j} + \delta_{j,r} x^{-a_r} = (-1)^{j-1} S_j^\hslash z^{r(a_1 - a_j -j)} + \delta_{j,r}z^{-ra_r}
    \\
    F_{k+1,k} &= x^{a_{k+1} - a_k \rfloor} = z^{r(a_{k+1} - a_k)} \,.
  \end{align}
  To calculate $ \det (Y - F_\hslash (x(z)))$, we first develop with respect to the first column. Given our entry there, we develop successively by rows, starting at the top: all of these choices will be unique. This gives
  \begin{equation}
    \begin{split}
      \det (Y - F_\hslash (x(z)) 
      &= 
      \sum_{j=1}^r (-1)^{j-1} \Big( -F_{1,j} + \delta_{j,1} Y_{1,1} \Big) \prod_{k=1}^{j-1} -F_{k+1,k} \prod_{l=j+1}^r Y_{l,l}
      \\
      &=
      \sum_{j=1}^r \Big( -S_j^\hslash z^{r(a_1 - a_j -j)} - \delta_{j,r} (-1)^{r-1} z^{-ra_r} + \delta_{j,1} z^{s-r} \Big) 
      \\
      & \hspace{4cm}
      \cdot \Big( \prod_{k=1}^{j-1} -z^{r(a_{k+1} - a_k} \Big) z^{(s-r)(r-j)}
      \\
      &=
      \sum_{j=1}^r \Big( -S_j^\hslash z^{r(a_1 - a_j -j)} - \delta_{j,r} (-1)^{r-1} z^{-ra_r} + \delta_{j,1} z^{s-r} \Big)  
      \\
      & \hspace{4cm}
      \cdot (-1)^{j-1} z^{r(a_j - a_1)} z^{(s-r)(r-j)}
      \\
      &=
      z^{(s-r)r} + \sum_{j=1}^r (-1)^j \Big( S_j^\hslash z^{-rj} + \delta_{j,r} (-1)^{r-1} z^{-r a_1}  \Big) z^{(s-r)(r-j)}
      \\
      &=
      z^{(s-r)r} + \sum_{j=1}^r  (-1)^j S_j^\hslash z^{-rj +(s-r)(r-j)} - z^{-r (r-s)} 
      \\
      &=
      \sum_{j=1}^r  (-1)^j S_j^\hslash z^{r(s-r) -sj}
    \end{split}
  \end{equation}
  Combining this with \cref{Assumption5Prefactor} gives \cref{Assumption5ConstantTerm}.\par
  Now let us consider the part of $ D (z,M)$ that is not constant in $M$, i.e. $ D (z,M) - D (z,0)$. Any term contributing to the development of this difference of determinants has to have at least one factor $ M_{j,k}$, which does not contribute a pole.\par
  In case there is a $j \in [r] $ such that $ a_{j+1} - a_j < -1$, the product
  \begin{equation}
    F_{j+1,j} M_{j,j+1} \prod_{k \neq j, j+1} Y_{k,k}\,,
  \end{equation}
  or in the case $ j = r$, the product
  \begin{equation}
    F_{1,r} M_{r,1} \prod_{k = 2}^{r-1} Y_{k,k} \,,
  \end{equation}
  has pole order
  \begin{equation}
    2r + (r-2)(r-s) = r^2 - rs + 2s \,, 
  \end{equation}
  which is larger than the vanishing order of \cref{Assumption5Prefactor}. So taking a matrix $M$ with only non-zero entry $ M_{j,j+1}$ will yield a pole.

  In case there are $ r'+1 $ consecutive entries $ x^{-1}$, we again complete to a block by using the right $ M_{j,k}$ to get a pole order
  \begin{equation}
    r(r'+1) + (r-s) ( r- 1 -r'-1 ) = r^2 -rs -r + s (1 +r'+1) \geq r^2 -rs +s \,,
  \end{equation}
  which is again greater than the vanishing order of \cref{Assumption5Prefactor}. So we get the same conclusion.

  In case there are more than $ r'' $ disjoint subsequences of $ r'$ consecutive $j$ such that $a_{j+1} -a_j = -1$, we can complete all of these by a non-zero element $ M_{j,k}$ and take $Y_{k,k}$ for the remainder, to obtain a pole order
  \begin{equation}
    \begin{split}
      r (r''+1) r'  + (r-s) & \Big( r - (r''+1) (1 + r' ) \Big)
      \\
      &=
      (r-s) r + \Big( r^2 - 1  +s - r' s(r-1) \Big) r' - (r +1) 
      \\
      &>
      (r-s) r + (r - 1  +s ) r' - (r +1)
      \\
      &\geq
      (r-s) r + 2r - 2\,,
    \end{split}
  \end{equation}
  which is again greater than the vanishing order of \cref{Assumption5Prefactor}. So again we get the same conclusion.

  To get the exact sizes of the blocks, we note that the exponents of the superdiagonal are all $ x^{\geq 0}$ or $ x^{-1}$, and we need at least $ r-s$ of the latter. The only way we can do this with at most $ r'' $ times $ r' $consecutive entries $ x^{-1}$ is to have that exactly $ r'' $ times, and $ s - r''$ times $ r' -1$ consecutive entries $ x^{-1}$, because of the identities
  \begin{align}
    r'' \big( r' + 1 \big) + (s - r'' ) r' &= r
    \\
    r'' r' + (s - r'' ) \big( r' -1 \big) &= r - s \,.
  \end{align}
  In particular, this means we need exactly $ r-s$ times $ x^{-1}$, and hence all other coefficients are $ x^0$.

  Now, assume that there are exactly $ r''$ large blocks, and $ s - r'' $ small blocks.\par
  Let us consider the $ S_j^\hslash$-independent part.  We see that the maximal pole order comes from taking the product of $ F_{k+1,k}$ in the large blocks and $Y_{k,k}$ in the small blocks. So the pole order of the determinant becomes
  \begin{equation}
    \begin{split}
    y^{(s-r'') r'} x^{-r'' r'} 
      &= 
      z^{\big( (s-r) (s-r'') - r r''\big) r'}
      \\
      &=
      z^{( s -r- r'' ) sr'}
      \\
      &=
      z^{-( r - (s - r'') ) (r - r'')}
    \end{split}
  \end{equation}
  and the total highest pole is
  \begin{equation}
    z^{(r-1)(r+s-1)-( r - (s - r'') ) (r - r'')}dz = z^{(r'' + \frac{s}{2} )^2 - ( 1 - \frac{s}{2})^2}dz
  \end{equation}
  as was to be proved. This power of $ z$ is non-negative if and only if $ | r'' - \frac{s}{2}| \geq | 1 - \frac{s}{2} | $, which combined with the fact that $ 1 \leq r'' \leq s-1$ gives $ r'' \in \{ 1, s-1\}$.\par
  \vspace{11pt}
  Then, we will consider terms that do contain $ S^\hslash_j$.

  First take $ s= 1$. In this case, 
  \begin{align}
    F_{1,j} 
    &= (-1)^{j-1} S_j^\hslash z^{-r} + \delta_{j,r} 
    \\
    F_{k+1,k} &= z^{-r} \,.
  \end{align}
  We see that in any column, the pole contributions is at most $ z^{-r}$. As we need an $M_{j,k}$ in at least one column, the maximal pole order in the determinant is $ z^{-r(r-1)}$. As $ \frac{dx}{P_y} = z^{(r-1)r}$ in this case, this proves that for $s=1$, $ D(z,M) - D(z,0)$ is holomorphic.\par

  Now, assume that $ s > 1$. For $ j > \lceil \frac{r}{s} \rceil$, the term
  \begin{equation}
    F_{1,j} \prod_{k=1}^{k_1-1} F_{k+1,k} \cdot M_{k_1+1,k} \prod_{k=\lceil \frac{r}{s} \rceil+1}^j F_{k+1,k} \prod_{l=j+1}^r Y_{l,l} \frac{dx}{P_y}
  \end{equation}
  has pole order $ r^2 + js - sr > 0$.\par
  If $ j \leq r'+1$, we use the same argumentation as for the $S_j^\hslash $-independent term, to divide the determinant into $ s$ blocks, of which $ s - r'' $ are diagonal products of $ Y_{l,l}$ and $ r'' $ are products of $ F_{k+1,k} $ and a $ M_{k_{l-1}+1,k_l}$. However, now for the first block we will use
  \begin{equation}
    F_{1,j} \prod_{k=1}^{j-1} F_{k+1,k} \prod_{l=j+1}^{r'(+1)} Y_{l,l} \,,
  \end{equation}
  with the $ (+1)$ depending on whether the first block is large or small.
  In case of a large first block, the total vanishing order is now (first line for the special block, second line is analogous to the $ S^\hslash_j$-independent term)
  \begin{equation}
    \begin{split}
      r(a_1 - a_j - j ) 
      &+ r ( a_j - a_1 ) + (r'+1-j)(s-r) 
      \\
      & +(s-r)(s-r'')r' - r (r''-1) r'+ (r+1-s)(r-1)
      \\
      &=
      -rj + (1-j) (s-r) + ( s^2 - rs -r''s +s ) r'+ (r+1-s)(r-1)
      \\
      &=
      - sj + (s-r) + ( s^2 - rs -r''s +s ) r'+ (r+1-s)(r-1)
      \\
      &=
      s(1-j) -r + ( s - r -r'' +1 ) (r-r'')+ (r+1-s)(r-1)
      \\
      &=
      s(1-j) -r + rs - r^2 - rr'' + r - ( s - r -r'' +1 )r'' + r^2 -1 - rs + s
      \\
      &=
      s(2-j) + ( - s + r'' - 1 )r'' -1
      \\
      &=
      s(2-j) + (r'' - \frac{s+1}{2} )^2 - 1 - \big(\frac{s + 1}{2}\big)^2 \,.
    \end{split}
  \end{equation}
  We see that this is always negative: the maximal value we can obtain with $ 1\leq r'' \leq s-1 $ requires $ |r''  - \frac{s+1}{2}|$ to be maximal, i.e. $ r'' = 1$, so that we get
  \begin{equation}
    s(1-j) - 1 < 0 \,.
  \end{equation}
  However, in case that $r'' = 1$, and a large first block, the term we considered was actually constant in $ M$: the only block that contained $M$ and $F$ wast the first one, and we exchanged the $M$ for $ F_{1,j}$. To obtain a term with at least one $M$ while keeping the maximal pole order, we should add one more $ F$ block, trading $ r' $ factors of $ y $ for $ r' -1 $ factors of $ x^{-1}$ to obtain
  \begin{equation}
    s(1-j) - 1 + r'(r-s) - (r' - 1) r
    =
    s(1-j)\,,
  \end{equation}
  which is non-negative only if $ j = 1$.\par
  In the case that the first block is small, we get
  \begin{equation}
    \begin{split}
      r(a_1 - a_j - j ) 
      &+ r ( a_j - a_1 ) + (r'-j)(s-r) 
      \\
      & +(s-r)(s-r''-1)r' - r r'' r'+ (r+1-s)(r-1)
      \\
      &= 
      -jr  + -j(s-r) +(s^2 - r''s -rs)r' + (r+1-s)(r-1)
      \\
      &= 
      -js +(s - r'' -r)(r-r'') + r^2 - 1 -rs + s
      \\
      &= 
      s(1-j) + (r'' -s)r'' - 1 \,.
    \end{split}
  \end{equation}
  Here, we see even more clearly, because $ 1 \geq r'' \geq s$ and $ j \geq 1$, that this is always negative.\par
\end{proof}

From all of this, we find the following result.

\begin{theorem}\label{t:ds}
  Let $ r $ and $ s$ be coprime, and write $ r = r's + r''$ for division with remainder $ 1< r'' < s $. Consider the spectral curve $ x(z) = z^r$ and $ y(z) = z^{s-r}$, with rational $ \hslash$-connection  $ \hslash d + \Phi_\hslash$ as in \cref{GenConnection}, which includes all cases of \cref{ConnectionPotential}. Then the $\hslash$-expansions of the non-perturbative amplitudes of this connection can be computed by shifted topological recursion of \cref{ShiftedTR} if one of the following three conditions hold:
  \begin{enumerate}
    \item $ s= 1$;
    \item $ r = 1 \pmod{s} $, the unique large block is in the top left,\footnote{The unique large block is always top left for \cref{ConnectionPotential}.} and $ S^\hslash_j = 0$ for $ j > 1$;
    \item $ s > 2$, $ r = -1 \pmod{s} $, and all $ S^\hslash_j = 0$.
  \end{enumerate}
\end{theorem}

\begin{proof}
  We first prove that Assumption 5*, \cref{Assumption5*}, holds in this setting. The first part of the assumption evidently holds: the singularities of $ \Phi_\ell$ are at $ z=0$, which is also a singularity of $ \phi$.\par
  The conditions given are those needed in \cref{ConditionsForAssumption5} to prove that $ D(z,M) - D(z,0)$ is holomorphic for any matrix $M$ with no pole at $z=0$. This in particular means that it holds for $ \mc{M}_{\vec{\delta}}^{(n)} (x(z);J)$ in the second part of Assumption 5*. In that second part, we only consider $ n \geq 1$, which means that we need to take a non-constant coefficient in $ \vec{\delta} $ in the determinant, which in turn means that we may consider $ D(z,M) - D(z,0)$ instead of just $ D(z,M)$. Therefore, \cref{ConditionsForAssumption5} implies the second part of Assumption 5*.\par
  For the third part of Assumption 5*, we note that
  \begin{equation}
    \begin{split}
      \frac{\omega(z)^{r-i}}{E_\omega (x(z),\omega(z))} S_{i,2g} \Big( \frac{dx(z)}{x(z)}\Big)^i
      &=
      \frac{\omega(z)^{1-i}}{r } S_{i,2g} \Big( \frac{dx(z)}{x(z)}\Big)^i
      \\
      &=
      S_{i,2g} z^{s(1-i)-1} dz
      \,.
    \end{split}
  \end{equation}
  so that \cref{eq:Assumption5*n=0,Assumption5ConstantTerm} match, and the third part of Assumption 5* is satisfied.\par
  We conclude by invoking \cref{MExpansionInExample,Assumption4Bypass,From5*ToTR}, noting that the numbered conditions in the theorem are exactly $s$-consistency.
\end{proof}

Note that the particular shape of the $n=0$ loop equation, or rather the $n=0$ condition in Assumption 5*, \cref{eq:Assumption5*n=0}, dictates through \cref{Assumption5ConstantTerm} the exact shape of the entries in the first column of $ \Phi_\hslash$; any other powers of $ x$ would yield a wrong shape.

\begin{remark}
  Interestingly, we note that the conditions that we obtained in \cref{t:ds} are used twice: to invoke \cref{t:shifts} and \cref{Assumption5*ToTT}. However, we used these conditions in very different ways. On the one hand, in \cref{Assumption5*ToTT} the conditions are required to get Assumption 5* to obtain the topological type property, so that the $\hslash$-expansion of the non-perturbative amplitudes of the $\hslash$-connection can be computed by shifted topological recursion. On the other hand, in \cref{t:shifts} the conditions are required for the left ideal to be an Airy ideal, which is in turn equivalent to showing that shifted topological recursion produces symmetric differentials. It is quite satisfying that the two sets of conditions are precisely the same!
\end{remark}

\backmatter

\bmhead{Acknowledgments}

We would like to thank N. Chidambaram, A. Giacchetto, J. Hurtubise, P. Lorenzoni and S. Shadrin for interesting discussions, and J. Hurtubise in particular for explaining the construction of formal WKB solutions relevant to our situation. We would also like to thank the anonymous referee for numerous thoughtful corrections.

\begin{appendices}

\section{Topological Type from the refined Assumption 5}
\label{TTAppendix}

In this appendix we give a proof of the topological type property based on Assumption 5*. This is a modification of an argument previously published in \cite{BEM17}. We follow the exact same steps, albeit not repeating them all, to explain how our refinement, distinguishing between the absence and presence of spectator variables in the loop equations, does not spoil the proof of the leading order property.\par
The validity of the non-perturbative loop equations \eqref{NonPertLoopEq} imply that the only steps of the proof that require being checked are those where Assumption 5 \eqref{Assumption5} was applied, that is to obtain equation (4.39) and the direct consequence of equation (4.45) in the original paper \cite{BEM17}. They correspond to equations (\ref{LeadOrd1}) and (\ref{LeadOrd2}) below, but to reach them we will first need to introduce some notations, and derive some intermediate results.\par
In this appendix, we write $ z.E $ for $ \overset{E}{z}$ as arguments of $ W_n$.

\begin{definition}
  For every $n\geq1$, define the \textit{primed} correlators by
  \begin{equation}
    W'_n \coloneq W_n - \frac1\hslash \delta_{n,1}\omega_{0,1} \,,
  \end{equation}
  as well as the \textit{partially disconnected} correlators, given for all $n\geq0$ by
  \begin{equation} 
  \label{PartDisc}
    {\mathcal{W}}'_{|I|,n}(I;J) 
    \coloneq 
    \sum_{\substack{(I_1,\dots, I_l)\vdash I \\ J_1\sqcup\dots\sqcup J_l=J}} \quad \prod_{i=1}^l W'_{|I_i|+|J_i|}(I_i,J_i) \,,
  \end{equation}
  for any subset $I\subset D \coloneq \big\{z. e_1,\dots,z. e_{r}\big\}$, where $( e_1,\dots, e_{r})$ is still the canonical basis of $r\times r$ diagonal matrices, and $ J = \big\{z_1.E_1,\dots,z_n.E_n\big\} $ is generic. We again use the equivariance property \eqref{WEquivariant} to identify the evaluation at the subset $I$ with evaluations at different points in the fiber $\mathfrak f(z)$, now all evaluated at the same diagonal matrix $e_r$. This is equivalent to the identification $D=\{\vartheta^{\tilde 1}z.e_r, \vartheta^{\tilde 2}z.e_r, \dots,\vartheta^{\tilde r}z.e_r\}\simeq\mathfrak f(z)$ with the notation of \eqref{MEquiv}. Furthermore, in this last expression, none of the $I_i$ featuring the underlying set partitions are allowed to be empty.
\end{definition} 

The original proof of the topological property was done by (nested) induction, and we now adapt it in the form of the following proposition. It implies the leading order property as a particular case, and its proof will make use of two intermediate lemmas, and one proposition.

\begin{proposition}
  \label{LeadingOrderProperty}
  In the situation of \cref{Assumption5*ToTT}, the proposition $\mathcal{P}_k$ given by
  \begin{equation} \label{TheoInduction}
    \mathcal{P}_k : \text{For all }j\geq k,\, W_j=\mathcal O(\hslash^{k-2})\,,
  \end{equation}
  holds for every $k\geq 1$.
\end{proposition}

\begin{proof}
  We see that $\mathcal{P}_1$ and $\mathcal{P}_2$ are trivial. Indeed, by definition $W_1(z_1.E_{1})$ is of order $\hslash^{-1}$ while all other correlation functions $W_n(z_1.E_{1},\dots,z_n.E_{n})$ with $n\geq 2$ are at least of order $\hslash^0$.

  We now assume as induction hypothesis that each proposition from $\mathcal{P}_1$ up to $\mathcal{P}_n$ for some $n\geq 2$ holds. Let us estimate the order of the second term of the right-hand side of (4.28) in \cite{BEM17}. The equation reads
  \begin{multline}
  \label{Core}
    P_n(x(z),\omega(z);J)  
    =
    \hslash W_{n+1}(z.e_{r},J) E_\omega (x(z),\omega (z)) \\ + \sum_{\{z\}\subsetneq I\subset \mathfrak f(z)} \hslash^{|I|} {\mathcal W}'_{|I|,n}(I;J)\prod_{z'\notin I}( \omega (z)- \omega (z')) \,.
  \end{multline}
  Its topological expansion will eventually yield \eqref{CombinatorialIdentity}. 

  As in \cite{BEM17}, there are three different cases to consider, depending on the set $I$ over which the underlying sum runs.
  \begin{itemize}
    \item[1)] $|I_i|=1$ and $J_i=\emptyset$: this corresponds to a single term of the form $W'_{1}(\vartheta^{\tilde a}z.e_r)$ which is of order at least $\hslash^0$, because in $W'_1$ the leading order term has been removed.
    \item[2)] $1<|I_i|+|J_i|\leq n$: here we can apply the induction hypothesis: $\mathcal{P}_{|I_i|+|J_i|}$ is assumed, so we get an order of $\hslash^{|I_i|+|J_i|-2}$.
    \item[3)] $|I_i|+|J_i|>n$: closely related to the previous case, we apply $\mathcal{P}_n$, from which we get an order of $\hslash^{n-2}$.
  \end{itemize} 
  Putting those three estimates together, it follows that
  \begin{equation}
    W'_{|I_i|+|J_i|}(I_i,J_i)
    =
    \mathcal O\Big(\hslash^{\min (n,|I_i|+|J_i|)-2+\delta_{|I_i|+|J_i|=1}}\Big) \,.
  \end{equation}
  In turn, for any integer $l\geq1$ labelling the length of the set partition featuring in the definition \eqref{PartDisc},
  \begin{equation}
  \label{OrdreEnCours}
    \hslash^{|I|}\prod_{i=1}^l W'_{|I_i|+|J_i|}(I_i,J_i)
    =
    \mathcal{O}\Big(\hslash^{\sum_{i=1}^l \left(\min (n,|I_i|+|J_i|)-2+\delta_{|I_i|+|J_i|=1}\right)+|I| }\Big) \,.
  \end{equation}

  Controlling this term then follows from a lemma, the proof of which we do not repeat here.

  \begin{lemma}[{\cite[p.~3232]{BEM17}}]
    For any $l\geq1$, the inequality
    \begin{equation}
    \label{Inequality} 
      \sum_{i=1}^l \left(\min(n,|I_i|+|J_i|)-2+\delta_{|I_i|+|J_i|=1}\right)+|I|-n\geq 0,
    \end{equation}
    holds whenever $\sum_{i=1}^l |J_i|=n$, $|I_i|\geq 1$ and $\sum_{i=1}^l |I_i|=|I|$.
  \end{lemma}

  Returning to \eqref{OrdreEnCours} and inserting \eqref{Inequality} implies that the second term of the right-hand side of the identity \eqref{Core} is at least of order $\mathcal O(\hslash^{n})$. It follows that for any positive $k\geq1$, the order $\hslash^{n-k+1}$ component of \eqref{Core} reads
  \begin{equation}
  \label{PreContradiction}
    P_n^{(n-k+1)}(x(z),\omega(z);J)
    =
    W_{n+1}^{(n-k)}(z.e_{r},J) E_\omega (x(z),\omega(z)) \,.
  \end{equation}
  From $\mathcal P_n$ it follows that $W_{n+1}=\mathcal O(\hslash^{n-2})$, and therefore that the right-hand side of \eqref{PreContradiction} vanishes for $k>2$. Hence we find that the first possibly non-vanishing term is actually for $k=2$, which leads to
  \begin{equation}
  \label{Contradiction}
    P_n^{(n-1)}(x(z),\omega(z);J)\frac{1}{E_\omega (x(z),\omega(z))}
    =
    W_{n+1}^{(n-2)}(z.e_{r},J) \,.
  \end{equation}
  Our knowledge on the pole structure of the non-perturbative correlators and their expansions, from the first part of the proof of \cref{Assumption5*ToTT}, guarantees that $W_{n+1}^{(n-2)}(z.e_{r},J)$ can only have poles at singularities of $ \phi$, while Assumption 5*, \cref{Assumption5}, implies that the left-hand side cannot have poles there. Thus, we get that $W_{n+1}^{(n-2)}(z.e_{r},J)dx$ defines a rational one-form without any poles. The only holomorphic one-form on $ \P^1$ is zero, so
  \begin{equation}
  \label{LeadOrd1}
    W_{n+1}^{(n-2)}(z.e_{r},J) = 0 \,.
  \end{equation}
  Therefore, $W_{n+1}(z.e_r,J)$ is at least of order $\hslash^{n-1}$, concluding the first part of the proof.

  \medskip

  The second part of the proof consists in extending the previous argument to higher correlators of the form $W_{n+p}$, with $p>1$,  and is also proved by induction, making use of the following proposition.

  \begin{proposition}
    The proposition ${\mathcal{P}}_{n,m}$ defined by
    \begin{center} 
      ${\mathcal{P}}_{n,m}$ : $W_m = \mc{O}(\hslash^{n-1})$ 
    \end{center}
    holds for $m\geq n+1$.
  \end{proposition}
 
  \begin{proof}
    The last identity  \eqref{LeadOrd1} is equivalently formulated as proposition ${\mathcal P}_{n,n+1}$, so the initial step of this second induction process holds.

    Now with $m\geq n+1$,  let us assume ${\mathcal P}_{n,n+1}, \dots, {\mathcal P}_{n,m}$ to hold, and prove that $\mathcal P_{n,m+1}$ holds as well. In order to do so, consider a set of $m$ distinct points $J=\big\{z_1.E_1,\dots,z_m.E_m\big\}$, and recall \eqref{Core}.\par
    Similarly to the previous proof, there are now four cases to consider, three of which already appeared above. Only the third one is new. 
    \begin{itemize}
      \item[1)] $|I_i|=1$ and $J_i=\emptyset$: this corresponds to a single term of the form $W'_{1}(\vartheta^{\tilde a}z.e_r)$ which is still of order at least $\hslash^0$, by definition.
      \item[2)] $1<|I_i|+|J_i|\leq n$: again the situation in which $\mathcal{P}_{|I_i|+|J_i|}$ applies, yielding an order of $\hslash^{|I_i|+|J_i|-2}$.
      \item[3)] $n<|I_i|+|J_i|\leq m$: In that case, we can apply $ {\mathcal{P}}_{n,|I_i|+|J_i|}$ and thus we get an order of $\hslash^{n-1}$
      \item[4)] $|I_i|+|J_i|> m$: this again corresponds to the case where we can $\mathcal{P}_n$, yielding an order of $\hslash^{n-2}$.
    \end{itemize}
    Again following \cite{BEM17}, denote by $L_1$ the set of subscript labels for which $1<|I_i|+|J_i|\leq n$, by $L_2$ the set of those for which $n<|I_i|+|J_i|\leq m$, and finally by $L_3$ that for which $|I_i|+|J_i|> m$. Introducing the cardinalities $l_1=|L_1|$, $l_2=|L_2|$, and $l_3=|L_3|$, they are non-negative integers that satisfy $l_1+l_2+l_3=l$. 
    With those notations in hand, gathering the conclusions of considering the four distinct cases now yields
    \begin{equation}
    \label{EnCours2}
      \hslash^{|I|}\prod_{i=1}^l W'_{|I_i|+|J_i|}(I_i,J_i)
      =
      \mathcal{O} \left(\hslash^{\underset{i\in L_1}{\sum}(|I_i|+|J_i|-2+\delta_{|I_i|+|J_i|=1})+\underset{i\in L_2}{\sum}(n-1)+\underset{i\in L_3}{\sum}(n-2)+|I| }\right) \,,
    \end{equation}
    that is again controlled by making use of a lemma that we recall without proof.
    \begin{lemma}[{\cite[p.~3234]{BEM17}}] 
      The inequality
      \begin{equation}
      \label{Inequality2} 
        \sum_{i\in L_1} (|I_i|+|J_i|-2+\delta_{|I_i|+|J_i|=1})+l_2(n-1)+l_3(n-2)+|I|-n \geq 0
      \end{equation}
      holds whenever we have $ \sum_{i=1}^l|J_i|=m$, $|I_i|\geq 1$, $ \sum_{i=1}^l |I_i|=|I|$, with $l_1+l_2+l_3=l$.
    \end{lemma}

    Inequality \eqref{Inequality2} together with \eqref{EnCours2} now implies that the expression $\hslash^{|I|} \mathcal W'_{|I|,j_0}(I;J)$ has at least order $\hslash^n$. Since $m+1> m\geq n+1> n$, and by proposition ${\mathcal P}_n$ \eqref{TheoInduction}, we also have that $W_{m+1}(z.e_r,J)$ is of order at least $\mathcal O(\hslash^{n-2})$. Writing the order $\hslash^{n-1}$ component of \eqref{Core} then leads to
    \begin{equation} 
    \label{Contradiction2}
      P_{m}^{(n-1)}(x(z),\omega (z);J)\frac{1}{E_\omega(x(z),\omega(z))}
      =
      W_{m+1}^{(n-2)}(z.e_r,J).
    \end{equation}
    The argument that was used to obtain \eqref{LeadOrd1} from \eqref{Contradiction} still applies, allowing us to conclude that
    \begin{equation}
    \label{LeadOrd2}
      W_{m+1}^{(n-2)}(z.e_r,J) = 0.
    \end{equation}

    \medskip

    We have finally obtained that assuming each ${\mathcal{P}}_{n,j}$ for $n+1\leq j\leq m$, it follows that ${\mathcal{P}}_{n,m+1}$ also holds. Since we had already proved the initial proposition $ {\mathcal{P}}_{n,n+1}$, we conclude by induction on $m$ that for all $m\geq n+1$, $ {\mathcal{P}}_{n,m}$ is in fact true. 
  \end{proof}
  
  Returning to the proof of \cref{TheoInduction}, this means that for all $m\geq n+1$, $W_{m+1}(z.e_{r},J)$ is of order at least $\hslash^{n-1}$, which is precisely the statement of proposition $\mathcal{P}_{n+1}$. We finally conclude by induction on $n$ that proposition $\mathcal{P}_n$ is valid for all $n\geq 1$.
\end{proof}

\end{appendices}

\section*{Statements and declarations}

The authors have no competing interests to declare that are relevant to the content of this article.

No data was generated or used for this article.

\bibliography{Refs}

\end{document}